\documentclass[a4paper,USenglish,cleveref,authorcolumns,autoref,thm-restate]{lipics-v2021}

\usepackage{graphicx} 

\usepackage{float}
\usepackage{tikz}
\usetikzlibrary{arrows.meta}
\usetikzlibrary{decorations.pathmorphing}
\usepackage{amsmath,amsthm,amssymb,amsfonts}
\usepackage{xcolor}
\usepackage{multirow}
\usepackage{subcaption}

\newcommand{\N}{\mathbb{N}}

\newcommand{\crn}[1]{\mathcal{#1}}
\newcommand{\species}{\Lambda}
\newcommand{\reactions}{\Gamma}
\newcommand{\reaction}{\mathcal{\gamma}}
\newcommand{\size}[1]{\lvert #1 \rvert}
\newcommand{\config}[1]{\overrightarrow{#1}}
\newcommand{\single}[1]{\vec{#1}}
\newcommand{\reactants}{\config{R}}
\newcommand{\products}{\config{P}}

\newcommand{\undefConfig}{\bot}

\newcommand{\cmap}{M}

\newcommand{\stepsto}{\mathop{\rightarrow}\limits}
\newcommand{\reaches}{\mathop{\leadsto}\limits}
\newcommand{\rxn}{\mathop{\longrightarrow}\limits}

\newcommand{\macrotrans}{\Rightarrow}

\newcommand{\Zero}{Z_{\emptyset}}

\newcommand{\SeqReacts}{\dashrightarrow}
\newcommand{\SeqProds}{\raisebox{2.5pt}{\tikz[baseline]{
    \draw[->] (0,0) -- (0.2,0);
    \draw[->] (0,0) -- (0.4,0);
    \draw[->] (0,0) -- (0.6,0);}}}

\newcommand{\RM}{\mathcal{R}} 
\newcommand{\PM}{\mathcal{P}} 
\newcommand{\Set}[1]{\{#1\}}
\newcommand{\para}[1]{\vspace*{.2cm}\noindent{\textbf{#1.}}}
\newcommand{\fillhor}{{\par\setlength{\parfillskip}{0pt}}}

\newcommand{\iCRNsimKVG}{
        \begin{tabular}{| l  l  l |}\hline
         $\forall \reaction_i \in \reactions:$ & $1.$ & $\config{\RM}_i \xrightarrow{I} \config{e}_{\Set{\RM_i}} + \config{\PM}_i + \size{\Set{\RM_i}} \cdot \single{I}$ \\ \hline
        \multirow{2}{*}{$\forall \lambda_i \in \species_1:$} & $2.$ & $\single{e}_{\lambda_i} + \single{I} + \single{\lambda}_i \rightarrow \single{\lambda}_i$ \\
        & $3.$ & $\single{e}_{\lambda_i} + \single{I} \xrightarrow{\lambda_i} \single{z}_{\lambda_i}$ \\ \hline
        \end{tabular}
}

\newcommand{\KVGsimiCRN}{
\begin{tabular}{| l  l  l |}\hline
            \multirow{2}{*}{$\forall \reaction_i \in \reactions:$} & $1.$ & $\single{G} + \config{\RM_i} + \config{Z_i} \rightarrow \single{G}_i + \config{\RM_i} + \config{Z_i}$ \\
            & $2.$ & $\single{G}_{i} + \config{\RM_i} \SeqProds \single{G} + \config{\PM_i}$ \\ \hline
        \end{tabular}
}

\newcommand{\iCRNiffKVG}{
    \begin{tabular}{|p{1.3cm} p{0.4cm} @{} p{4cm}| @{}p{.02cm}@{} | p{1.1cm} p{0.4cm} @{} p{3.8cm}|}\cline{1-3}\cline{5-7}
         $\forall \reaction_i \in \reactions:$ & $1.$ & $\config{\RM}_i \xrightarrow{I} \config{e}_{\Set{\RM_i}} + \config{\PM}_i + \size{\Set{\RM_i}} \cdot \single{I}$ & &
         \multirow{2}{*}{$\forall \reaction_i \in \reactions:$} & $1.$ & $\single{G} + \config{\RM_i} + \config{Z_i} \rightarrow \single{G}_i + \config{\RM_i} + \config{Z_i}$ \\ \cline{1-3}
         
        \multirow{2}{*}{$\forall \lambda_i \in \species_1:$} & $2.$ & $\single{e}_{\lambda_i} + \single{I} + \single{\lambda}_i \rightarrow \single{\lambda}_i$ 
        & &
        & $2.$ & $\single{G}_{i} + \config{\RM_i} \SeqProds \single{G} + \config{\PM_i}$ \\ 
        
        & $3.$ & $\single{e}_{\lambda_i} + \single{I} \xrightarrow{\lambda_i} \single{z}_{\lambda_i}$ 
        & & & &\\ \cline{1-3}\cline{5-7}
        \multicolumn{3}{l}{\textbf{(a)} iCRN simulating $k$-VG} & \multicolumn{1}{c}{}& \multicolumn{3}{l}{\textbf{(b)} $k$-VG simulating iCRN}\\
    \end{tabular}
}

\newcommand{\CRsimKVG}{
\begin{tabular}{| l  l  l |}\hline
        \multicolumn{3}{|c|}{\textbf{Fast Reactions} (Rank 2)} \\ 
        $\forall$ $\reaction_i \in \reactions:$ & $1.$ & $\single{G} + \config{\RM_i} \rightarrow \single{x} + \config{e}_{\{\RM_i\}} + \config{\PM_i}$ \\ \hline
        \multirow{2}{*}{$\forall \lambda_i \in \species_1:$} & $2.$ & $\single{x} + \single{e}_{\lambda_i} + \single{\lambda}_i \rightarrow \single{x} + \single{\lambda}_i$ \\
        & $3.$ & $\single{y} + \single{e}_{\lambda_i} \rightarrow \single{y} + \single{z}_{\lambda_i}$ \\ \hline
        \multicolumn{3}{|c|}{\textbf{Slow Reactions} (Rank 1)} \\ 
        & $4.$ & $\single{x} \rightarrow \single{y}$ \\
        & $5.$ & $\single{y} \rightarrow \single{G}$ \\ \hline
    \end{tabular}
}

\newcommand{\KVGsimCR}{
\begin{tabular}{| l  l  l |}\hline
            & $1.$ & $\single{G} + \single{g}_i + \config{\RM_i} \rightarrow \single{r}_1 + \config{\PM_i}$ \\
            $\forall \reaction_i^2 \in \reactions^2$,   & $1b.$ & $\single{G} + \single{g}_{i} + \single{z}_k \rightarrow \single{G} + \single{z}_k$ \\
            $z_k \in \{\RM_i\}$: & $2.$ & $\single{r}_j + \single{g}_j \rightarrow \single{r}_{j+1} + \single{g}_j$ \\
            & $3.$ & $\single{r}_j + \single{z}_{g_j} \rightarrow \single{r}_{j+1} + \single{g}_j$ \\ \hline
            & $4.$ & $\single{r}_{|\reactions^2|+1} \rightarrow \single{G}$ \\ 
            & $5.$ & $\single{G} + \single{z}_{g_1} + \ldots + \single{z}_{g_{|\reactions^2|}} \rightarrow \single{s}$ \\ 
            & $6.$ & $\single{t} \rightarrow \single{G} + \single{g}_1 + \ldots + \single{g}_{|\reactions^2|}$ \\ \hline 
            $\forall \reaction_i^1 \in \reactions^1$: & $7.$ & $\single{s} + \config{\RM_i} \SeqProds \single{t} + \config{\PM_i}$ \\ \hline
        \end{tabular}
}

\newcommand{\CRiffVG}{
\begin{tabular}{| @{}l @{} l @{ } l | @{}l@{} |@{}l @{}l@{ } p{4cm} |}\cline{5-7}
        \multicolumn{3}{c}{} & &
        & $1.$ & $\single{G} + \single{g}_i + \config{\RM_i} \rightarrow \single{r}_1 + \config{\PM_i}$ \\ \cline{1-3}
        
        \multicolumn{3}{|c|}{\textbf{Fast Reactions} (Rank 2)} & &
        $\forall \reaction_i^2 \in \reactions^2$,   & $1b.$ & $\single{G} + \single{g}_{i} + \single{z}_k \rightarrow \single{G} + \single{z}_k$ \\
        
        $\forall$ $\reaction_i \in \reactions:$ & $1.$ & $\single{G} + \config{\RM_i} \rightarrow \single{x} + \config{e}_{\{\RM_i\}} + \config{\PM_i}$ & &
        $z_k \in \{\RM_i\}$: & $2.$ & $\single{r}_j + \single{g}_j \rightarrow \single{r}_{j+1} + \single{g}_j$
        \\ \cline{1-3}
        
        \multirow{2}{*}{$\forall \lambda_i \in \species_1:$} & $2.$ & $\single{x} + \single{e}_{\lambda_i} + \single{\lambda}_i \rightarrow \single{x} + \single{\lambda}_i$ & &
        & $3.$ & $\single{r}_j + \single{z}_{g_j} \rightarrow \single{r}_{j+1} + \single{g}_j$ \\ \cline{5-7}
        
        & $3.$ & $\single{y} + \single{e}_{\lambda_i} \rightarrow \single{y} + \single{z}_{\lambda_i}$ & &
        & $4.$ & $\single{r}_{|\reactions^2|+1} \rightarrow \single{G}$ \\ \cline{1-3}
        
        \multicolumn{3}{|c|}{\textbf{Slow Reactions} (Rank 1)} & &
        & $5.$ & $\single{G} + \single{z}_{g_1} + \ldots + \single{z}_{g_{|\reactions^2|}} \rightarrow \single{s}$ \\ 
        
        & $4.$ & $\single{x} \rightarrow \single{y}$ & &
         & $6.$ & $\single{t} \rightarrow \single{G} + \single{g}_1 + \ldots + \single{g}_{|\reactions^2|}$ \\ 
         
        & $5.$ & $\single{y} \rightarrow \single{G}$ & &
        $\forall \reaction_i^1 \in \reactions^1$: & $7.$ & $\single{s} + \config{\RM_i} \SeqProds \single{t} + \config{\PM_i}$ \\ \cline{1-3}\cline{5-7}
        \multicolumn{3}{l}{\textbf{(a)} Coarse-Rate CRN simulating $k$-VG CRN} & \multicolumn{1}{c}{}& \multicolumn{3}{l}{\textbf{(b)} $k$-VG CRN simulating Coarse-Rate CRN}\\
    \end{tabular}

}

\newcommand{\SCsimKVG}{
\begin{tabular}{| l @{} l @{} l |}\hline
        $\forall \reaction_i \in \reactions:$ & $1.$ & $\single{G} + \config{\RM_i} \rightarrow \config{e}_{\{\RM_i\}} + \config{\PM_i}$ \\ \hline
        \multirow{2}{*}{$\forall \lambda_i \in \species_1:$} & $2.$ & $\single{\lambda}_i + \single{e}_{\lambda_i} \rightarrow \single{\lambda}_i$ \\
        & $3.$ & $\single{y} + \single{e}_{\lambda_i} \rightarrow \single{y} + \single{z}_{\lambda_i}$ \\ \hline
        & $4.$ & $\single{y} + \single{y} \rightarrow \single{G}$ \\
        & $5.$ & $\single{G} + \single{y} \leftrightarrow \single{G} + \single{w}$ \\ 
        & $6.$ & $\single{w} \leftrightarrow x$ \\ \hline
        \end{tabular}
}
\newcommand{\KVGsimSC}{
        \begin{tabular}{|l @{}l @{}l|}
            \hline
            & $1.$ & $\single{G} + \single{g}_i \SeqReacts \single{r}_1 +  \config{\PM_i} $ \\
            \multirow{1}{*}{$\forall \reaction_i \in \reactions$, $z_k \in \{\RM_i\}$:}  
            & $2.$ & $\single{r}_j + \single{g}_j \rightarrow \single{r}_{j+1} + \single{g}_j$ \\
            & $3.$ & $\single{r}_j + \single{z}_{g_j} \rightarrow \single{r}_{j+1} + \single{g}_j$ \\ 
            \hline
            & $4.$ & $\single{r}_{|\reactions|+1} \rightarrow \single{G}$ \\
            & $5.$ & $\single{G} + \single{z}_{g_1} + \ldots + \single{z}_{g_{|\reactions|}} \rightarrow \single{s}$ \\
            & $6.$ & $\single{t} \rightarrow \single{G} + \single{g}_1 + \ldots + \single{g}_{|\reactions|}$ \\ 
            \hline
            \multirow{1}{*}{$\forall S_i \in S \setminus \{ S_{k-1} \}:$}  
            & $7.$ & $\single{s} + \single{s}_i \SeqProds \single{t} + \single{s}_{i+1} + \config{S}_i$ \\ 
            \hline
            & $8.$ & $\single{s} + \single{s}_{k-1} \SeqProds \single{t} + \single{s}_0 + \config{S}_{k-1}$ \\ 
            \hline
        \end{tabular}
}

\newcommand{\SCiffKVG}{
\begin{tabular}{| l @{}l @{}l | l | l @{} l @{} l |}\cline{5-7}
        \multicolumn{3}{l}{} & &
        & $1.$ & $\single{G} + \single{g}_i \SeqReacts \single{r}_1 +  \config{\PM_i} $ \\
        
        \multicolumn{3}{l}{} & &
        \multirow{1}{*}{$\forall \reaction_i \in \reactions$, $z_k \in \{\RM_i\}$:}  
        & $2.$ & $\single{r}_j + \single{g}_j \rightarrow \single{r}_{j+1} + \single{g}_j$ \\ \cline{1-3}

        $\forall \reaction_i \in \reactions:$ & $1.$ & $\single{G} + \config{\RM_i} \rightarrow \config{e}_{\{\RM_i\}} + \config{\PM_i}$ & &
        & $3.$ & $\single{r}_j + \single{z}_{g_j} \rightarrow \single{r}_{j+1} + \single{g}_j$ \\ \cline{1-3}\cline{5-7}
        
        \multirow{2}{*}{$\forall \lambda_i \in \species_1:$} & $2.$ & $\single{\lambda}_i + \single{e}_{\lambda_i} \rightarrow \single{\lambda}_i$  & &
         & $4.$ & $\single{r}_{|\reactions|+1} \rightarrow \single{G}$ \\

        & $3.$ & $\single{y} + \single{e}_{\lambda_i} \rightarrow \single{y} + \single{z}_{\lambda_i}$ & &
        & $5.$ & $\single{G} + \single{z}_{g_1} + \ldots + \single{z}_{g_{|\reactions|}} \rightarrow \single{s}$ \\ \cline{1-3}

        & $4.$ & $\single{y} + \single{y} \rightarrow \single{G}$ & &
        & $6.$ & $\single{t} \rightarrow \single{G} + \single{g}_1 + \ldots + \single{g}_{|\reactions|}$ \\ \cline{5-7}
        
        & $5.$ & $\single{G} + \single{y} \leftrightarrow \single{G} + \single{w}$ & &
        \multirow{1}{*}{$\forall S_i \in S \setminus \{ S_{k-1} \}:$} 
        & $7.$ & $\single{s} + \single{s}_i \SeqProds \single{t} + \single{s}_{i+1} + \config{S}_i$ \\ \cline{5-7}

        & $6.$ & $\single{w} \leftrightarrow x$ & &
         & $8.$ & $\single{s} + \single{s}_{k-1} \SeqProds \single{t} + \single{s}_0 + \config{S}_{k-1}$ \\ \cline{1-3}\cline{5-7}
         \multicolumn{3}{l}{\textbf{(a)} Step-Cycle CRN sim. $k$-VG CRN} & \multicolumn{1}{c}{}& \multicolumn{3}{l}{\textbf{(b)} $k$-VG CRN simulating a Step-Cycle CRN}\\
    \end{tabular}
}

\newcommand{\CheckReactants}{
\begin{tabular}{| l  l |}\hline
        & $\single{G} + \single{g}_i \rightarrow \single{R}_1^{1}$ + $\single{z}_{g_i}$ \\ \hline
        \multirow{2}{*}{$\forall \lambda_j \in \RM_i:$}
        & $\single{R}_i^k + \single{\lambda}_j \rightarrow \single{R}_i^{k+1} + \single{\lambda}_j'$ \\ 
        & $\single{R}_{i}^{|\RM_i|+1} + \single{\lambda}_1' + \ldots + \single{\lambda}_{|\RM_i|}' \rightarrow \single{r}_1 + \config{\PM}_i$ \\ \hline
    \end{tabular}
}

\newcommand{\UndoReactants}{
\begin{tabular}{| l  l |}\hline
        & $\single{R}_i^k + \single{z}_{\lambda_j} \rightarrow \single{R}_i^{k^-} + \single{z}_{\lambda_j}$ \\ 
        \multirow{2}{*}{$\forall \lambda_j \in \RM_i:$} 
        & $\single{R}_i^{k^-} + \single{z}_{\lambda_j} \rightarrow \single{R}_i^{k-1^-} + \single{\lambda}_j$ \\
        & $\single{R}_i^{k^-} + \single{\lambda}_j' \rightarrow \single{R}_i^{k^-} + \single{\lambda}_j$ \\ 
        & $\single{R}_1^{1^-} \rightarrow \single{G}$ \\ \hline
    \end{tabular}
}

\newcommand{\SCReactants}{
\begin{tabular}{| l  l | l |l l|}\cline{4-5}
        \multicolumn{2}{c}{} & &
        & $\single{R}_i^k + \single{z}_{\lambda_j} \rightarrow \single{R}_i^{k^-} + \single{z}_{\lambda_j}$ \\ \cline{1-2}

        & $\single{G} + \single{g}_i \rightarrow \single{R}_1^{1}$ + $\single{z}_{g_i}$ & &
        \multirow{2}{*}{$\forall \lambda_j \in \RM_i:$} 
        & $\single{R}_i^{k^-} + \single{z}_{\lambda_j} \rightarrow \single{R}_i^{k-1^-} + \single{\lambda}_j$ \\ \cline{1-2}
        \multirow{2}{*}{$\forall \lambda_j \in \RM_i:$}
        & $\single{R}_i^k + \single{\lambda}_j \rightarrow \single{R}_i^{k+1} + \single{\lambda}_j'$ & &
        & $\single{R}_i^{k^-} + \single{\lambda}_j' \rightarrow \single{R}_i^{k^-} + \single{\lambda}_j$ \\ 
        & $\single{R}_{i}^{|\RM_i|+1} + \config{\RM_i}' \SeqProds \single{r}_1 + \config{\PM}_i$ & &
        & $\single{R}_1^{1^-} \rightarrow \single{G}$ \\ \cline{1-2}\cline{4-5}
        \multicolumn{2}{l}{\textbf{(a)} Check reactants} & \multicolumn{1}{c}{}& \multicolumn{2}{l}{\textbf{(b)} Undo reactants}\\
    \end{tabular}
}

\newcommand{\UIsimsVG}{
    \begin{tabular}{|l @{}l l|} \hline
        $\forall \reaction_i \in \reactions:$ & $1.$ & $\single{G} + \config{\RM_i} \rightarrow \single{G_i} + \config{e}_{\{\RM_i\}} + \config{\PM_i}$ \\ \hline
        {$\forall \lambda_j \in \species_1:$} & $2.$ & $\single{e}_{\lambda_j} \rightarrow \single{r}^1_j + \single{t}^1_j$ \\ \cline{2-3}
         & $3.$ & $\single{r}^1_j + \single{\lambda}_i \rightarrow \single{r}^2_j$ \\
         &      & $\single{t}^1_j \rightarrow \single{t}_2$ \\ \cline{2-3}
         & $4.$ & $\single{t}_2 + \single{r}^1_j \rightarrow \single{z} + \single{E}$ \\ 
         &      & $\single{t}_2 + \single{r}^2_j \rightarrow \single{\lambda}_i + \single{E}$ \\ \hline
         & $5.$ & $\single{G_i} + |\{\RM_i\}|\cdot \single{E} \rightarrow  \single{G}$ \\ \hline
    \end{tabular}
}
\newcommand{\KVGsimsUI}{
        \begin{tabular}{|l l l|}\hline
            $\forall \gamma_i \in \reactions,$   & $1.$ & $\single{G}_i + \single{z}_{r_j} \rightarrow \single{N}_i + \single{X}_i + \single{z}_{r_j}$  (can not run)  \\ 
            $\forall r_j \in \RM_i:$ && \\ \hline
            seq. & $2.$ & $\single{G}_i + \single{r}_1 \rightarrow \single{R}_i^2$\\
            reacts. & $3.$ & $\single{R}_i^j + \single{r}_j \rightarrow \single{R}_i^{j+1}$ \\
             & $4.$ & $\single{R}_i^j + \single{z}_{r_j} 
             \rightarrow \single{N}_i + \single{X}_i + \single{r}_1 + \ldots + \single{r}_{j-1} + \single{z}_{r_j}$ \\
            $j = |\RM_i|$  & $5.$ & $\single{R}_i^j + \single{r}_{j} \rightarrow \single{I}_i + \single{X}_i$ (rule selected) \\ \hline
             & $6.$ & $X_1 + \ldots + X_{|\reactions|} \rightarrow \single{F}$ (rules selected)\\ \hline
            $\forall \gamma_i \in \reactions$   & $7.$ & $\single{I}_i + \single{F} \rightarrow \single{F}_i + \single{F} + \config{\PM}_i$ (output)\\
            & $8.$ & $\single{N}_i + \single{F} \rightarrow \single{F}_i + \single{F}$ (not run)\\ \hline
            reset & $9.$ & $\single{F} + \single{F}_1 + \ldots + \single{F}_{|\Gamma|} \rightarrow \single{G}_i + \ldots + \single{G}_{|\reactions|}$\\ \hline
        \end{tabular}
   }

\newcommand{\UIiffVG}{
    \begin{tabular}{|@{}l@{}l@{}l@{}|@{}l@{}|l@{}l@{}|} \cline{5-6}
        \multicolumn{3}{c}{}& & 
        $\forall \gamma_i \in \reactions,$  & $ \single{G}_i + \single{z}_{r_j} \rightarrow \single{N}_i + \single{X}_i + \single{z}_{r_j}$  (can not run)  \\ 
        \multicolumn{3}{c}{}& &
         $\forall r_j \in \RM_i:$  &   \\ \cline{5-6}
        
        \multicolumn{3}{c}{}& & 
        seq. & $\single{G}_i + \single{r}_1 \rightarrow \single{R}_i^2$\\ \cline{1-3}
        $\forall \reaction_i \in \reactions:$ & $1.$ & $\single{G} + \config{\RM_i} \rightarrow \single{G_i} + \config{e_{\{\RM_i\}}}+ \config{\PM_i} $ 
        & &
        reacts. & $\single{R}_i^j + \single{r}_j \rightarrow \single{R}_i^{j+1}$ \\ \cline{1-3}
        {$\forall \lambda_j \in \species_1:$} & $2.$ & $\single{e}_{\lambda_j} \rightarrow \single{r}^1_j + \single{t}^1_j$  & &
        & $\single{R}_i^j + \single{z}_{r_j} \rightarrow \single{N}_i + \single{X}_i + \single{r}_1 + \ldots + \single{r}_{j-1} + \single{z}_{r_j}$ \\ \cline{2-3}
        & $3.$ & $\single{r}^1_j + \single{\lambda}_i \rightarrow \single{r}^2_j$ & &
        $j = |\RM_i|$  & $\single{R}_i^j + \single{r}_{j} \rightarrow \single{I}_i + \single{X}_i$ (rule selected) \\ \cline{5-6}
        &  & $\single{t}^1_j \rightarrow \single{t}_2$ & &
        & $X_1 + \ldots + X_{|\reactions|} \rightarrow \single{F}$ (rules selected) \\ \cline{2-3}\cline{5-6}
        & $4.$ & $\single{t}_2 + \single{r}^2_j \rightarrow \single{\lambda}_i + \single{E}$ & &
        $\forall \gamma_i \in \reactions$   & $\single{I}_i + \single{F} \rightarrow \single{F}_i + \single{F} + \config{\PM}_i$ (output)\\
        &      & $\single{t}_2 + \single{r}^1_j \rightarrow \single{z} + \single{E}$ & &
        & $\single{N}_i + \single{F} \rightarrow \single{F}_i + \single{F}$ (not run)\\ \cline{1-3}\cline{5-6}
        & $5.$ & $\single{G_i} + |\{\RM_i\}|\cdot \single{E} \rightarrow  \single{G}$ & &
        reset & $\single{F} + \single{F}_1 + \ldots + \single{F}_{|\Gamma|} \rightarrow \single{G}_i + \ldots + \single{G}_{|\reactions|}$\\ \cline{1-3}\cline{5-6}
        \multicolumn{3}{l}{\textbf{(a)} UI parallel CRN simulating a VG CRN} & \multicolumn{1}{c}{}& \multicolumn{2}{l}{\textbf{(b)} $k$-VG simulating UI parallel CRN}\\
    \end{tabular}
   }

\title{Polynomial Equivalence of Extended Chemical Reaction Models}

\author{Divya Bajaj}{University of Texas Rio Grande Valley}{divya.bajaj@utrgv.edu}{}{}

\author{Jose-Luis Castellanos}{University of Texas Rio Grande Valley}{joseluis.castellanos01@utrgv.edu}{}{}

\author{Ryan Knobel}{University of Texas Rio Grande Valley}{ryan.knobel01@utrgv.edu}{}{}

\author{Austin Luchsinger}{University of Texas Rio Grande Valley}{austin.luchsinger@utrgv.edu}{}{}

\author{Aiden Massie}{University of Texas Rio Grande Valley}{aiden.massie01@utrgv.edu}{}{}

\author{Adrian Salinas}{University of Texas Rio Grande Valley}{adrian.salinas08@utrgv.edu}{}{}

\author{Pablo Santos}{University of Texas Rio Grande Valley}{pablo.santos01@utrgv.edu}{}{}

\author{Ramiro Santos}{University of Texas Rio Grande Valley}{ramiro.santos01@utrgv.edu}{}{}

\author{Robert Schweller}{University of Texas Rio Grande Valley}{robert.schweller@utrgv.edu}{}{}

\author{Tim Wylie}{University of Texas Rio Grande Valley}{timothy.wylie@utrgv.edu}{}{}

\authorrunning{D.~Bajaj et al.}

\Copyright{Divya Bajaj \and Jose Luis Castellanos \and Ryan Knobel \and Austin Luchsinger \and Aiden Massie \and Adrian Salinas \and Pablo Santos \and Ramiro Santos \and Robert Schweller \and Tim Wylie}



\category{}

\relatedversion{}

\funding{This research was supported in part by National Science Foundation Grant CCF-2329918.}

\nolinenumbers 

\EventEditors{Ho-Lin Chen, Wing-Kai Hon, and Meng-Tsung Tsai}
\EventNoEds{3}
\EventLongTitle{36th International Symposium on Algorithms and Computation (ISAAC 2025)}
\EventShortTitle{ISAAC 2025}
\EventAcronym{ISAAC}
\EventYear{2025}
\EventDate{December 7--10, 2025}
\EventLocation{Tainan, Taiwan}
\EventLogo{}
\SeriesVolume{359}
\ArticleNo{40}

\begin{document}
\maketitle
\begin{abstract}
The ability to detect whether a species (or dimension) is zero in Chemical Reaction Networks (CRN), Vector Addition Systems, or Petri Nets is known to increase the power of these models--- making them capable of universal computation.
While this ability may appear in many forms, such as extending the models to allow transitions to be inhibited, prioritized, or synchronized, we present an extension that directly performs this zero checking.
We introduce a new \emph{void genesis} CRN variant with a simple design that merely increments the count of a specific species when any other species' count goes to zero. As with previous extensions, we show that the model is Turing Universal. We then analyze several other studied CRN variants and show that they are all equivalent through a \textit{polynomial} simulation with the void genesis model, which does not merely follow from Turing-universality. Thus, inhibitor species, reactions that occur at different rates, being allowed to run reactions in parallel, or even being allowed to continually add more volume to the CRN, does not add additional \textit{simulation power} beyond simply detecting if a species count becomes zero.

\end{abstract}

\newpage

\section{Introduction}\label{sec:intro}

\para{Background}
Chemical Reaction Networks \cite{Aris:1965:ARMA}, Vector Addition Systems \cite{karp1967parallel}, and Petri nets \cite{petri1966communication} are three formalisms that arose in different disciplines to study distributed/concurrent systems.
Despite their distinct origins, these models are mathematically equivalent in expressive power \cite{cook2009programmability, hack1976decidability}.
While these models are capable of very complex behavior --- e.g., deciding if one configuration is reachable from another via a sequence of transitions was recently shown to be Ackermann-complete \cite{czerwinski2022reachability,leroux2022reachability} --- they fall short of Turing-universality.

Even though these base models are not capable of universal computation, they are each right on the cusp of doing so.
It is well known that extending the models in any way that allows ``checking for zero'' immediately results in Turing-universality \cite{agerwala1974complete,hack1976petri,peterson1981petri}.
This has been shown for model extensions allowing transition inhibition \cite{agerwala1974complete, calabrese2024inhibitory, esparza2024decidability, hack1976petri}, transition prioritization \cite{hack1976petri, senum2011rate, shea2010writing}, synchronous (parallel) transitions \cite{berthomieu2024sleptsov,thomas1976petri}, or even continual volume increase \cite{anderson2024steps, anderson2024computing}.
These results are typically established by showing how each extended model can simulate a register (counter) machine \cite{peterson1981petri}.

However, Turing equivalence alone is a blunt instrument.
It tells us that models can emulate one another, but not how efficiently this can be done.
We often care not just that two models are Turing-complete, but if they can be easily transformed into one another.
For instance, Hack gave one such transformation between Inhibitory and Priority Petri nets \cite{hack1976petri}, but discussed their equivalence in terms of languages.
For this paper, we focus on the concept of simulation to draw comparisons between models.

In the literature, several notions of simulation have been proposed to capture structural or behavioral equivalence between systems.
While concepts like strong/weak bisimulation \cite{antoniotti2004taming, dong2012bisimulation,fernandez1990implementation,johnson2019verifying,milner1989communication} and pathway decomposition \cite{shin2012compiling,shin2019verifying} have been used to compare the expressiveness of different systems, there is typically a tradeoff between reasonably preserving dynamics and maintaining structural correspondence.
Furthermore, efficiency of simulation is often either implicitly included or sometimes omitted altogether.
Part of our work aims to introduce a more wieldy definition of simulation that explicitly accounts for efficiency.
We give a more detailed discussion on simulation later in the paper.

\para{Our Contributions}
In this work, we make two main contributions:
\begin{enumerate}
    \item We define a general-purpose notion of polynomial efficient simulation that is intended to capture a broader notion of simulation while still remaining reasonable.  Our definition ensures that the simulation respects a polynomial correspondence in both time (dynamics) and space (structure).  This allows us to formally compare the simulation efficiency of different CRN variants.
    \item We introduce a new model, which we call Void Genesis (VG), that makes zero-testing explicit: it creates a designated species whenever a tracked species reaches zero. We use this model as a unifying model that captures the essence of zero-testing in a clean and modular way. We show that the Void Genesis model is polynomially equivalent to other known Turing-universal extensions of CRNs.
\end{enumerate}

Since all of these extended models can be simulated by VG (and vice versa) with only polynomial overhead, our results establish a hub of polynomial simulation equivalence among them.
This yields a stronger and more precise understanding of the relationships between these CRN variants and highlights VG simulation as a useful proof technique for adding models to this hub.
Figure \ref{fig:vg_hub} summarizes the simulation relationships we establish here.

\para{Organization}
Given the number of models and results, we have arranged the paper in a way that systematically builds up understanding and techniques. Section \ref{sec:definitions} covers the formal definitions of the models and a simulation of one model with another. Section \ref{sec:zeros} provides minimum working examples in the models. The polynomial equivalence between models is proven through a series of simulations in Section \ref{sec:equivalence}. In Section \ref{sec:conc}, we discuss our results and some open problems.

\begin{figure}[t]
    \centering
    \includegraphics[width=0.6\linewidth]{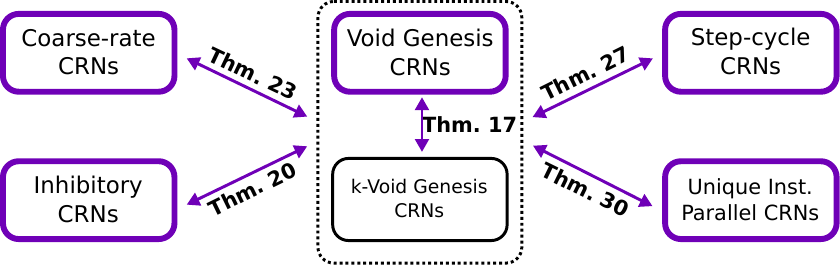}
    \caption{CRN Model variants and their connections to Void Genesis CRNs. Theorems 1-5 each show a polynomial equivalence between two CRN variants. With the Void Genesis model as a central hub, the implication of these results is a polynomial equivalence between all models.}
    \label{fig:vg_hub}
\end{figure}
\section{Preliminaries}\label{sec:definitions}
In Section~\ref{subsec:models}, we define the extended chemical reaction network models considered in this paper with examples, and in Section~\ref{subsec:sim} we define the concept of inter-model simulation.

\subsection{Models}\label{subsec:models}
Here, we define the six chemical reaction networks considered in this paper: the basic chemical reaction network model (Section~\ref{subsec:crn}), the Void Genesis model (Section~\ref{subsec:vg}), the Inhibitory CRN model (Section~\ref{subsec:inhib}), the Coarse-Rate CRN model (Section~\ref{subsec:coarse}), the Step-Cycle CRN model (Section~\ref{subsec:step-cycle}), and the Unique Instruction Parallel model (Section~\ref{subsec:parallel}).

\subsubsection{Chemical Reaction Networks}\label{subsec:crn}

A \emph{chemical reaction network} (CRN) $\crn{C} = (\species, \reactions)$
is defined by a finite set of species $\species$
, and a finite set of reactions $\reactions$ where each reaction is a pair $( \reactants,\products) \in \N^\species \times \N^\species$, sometimes written $\config{R} \rxn \config{P}$, that denotes the \textit{reactant} species consumed by the reaction and the \textit{product} species generated by the reaction.
For example, given $\species = \{a,b,c\}$, the reaction $( (2,0,0), (0,1,1) )$ represents $2a \rxn b + c$; 2 $a$ species are removed, and 1 new $b$ and $c$ species are created.   

A \textit{configuration} $\config{C}\in \N^\species$ of a CRN assigns integer counts to every species $\lambda \in \species$, and we use notation $\config{C} \, [\lambda]$ to denote that count.
For a species $\lambda \in \species$, we denote the configuration consisting of a single copy of $\lambda$ and no other species as $\single{\lambda}$.  
It is often useful to reference the set of species whose counts are not zero in a given configuration.
In such cases, the notation $\Set{\config{C}}$ is used.
Formally, $\Set{\config{C}} = \{\lambda \in \species \mid \config{C}[\lambda] > 0\}$, and when convenient and clear from the context, we further use $\{\config{C}\}$ to denote the configuration (vector) representation in which each element has a single copy.  Finally, let $|\config{C}| = \sum_{\lambda\in \species} C[\lambda]$ denote the total number of copies of all species in a configuration, sometimes referred to as the \emph{volume} of $\config{C}$.

A reaction $(\reactants,\products)$ is said to be \textit{applicable} in configuration $\config{C}$ if $\reactants \leq \config{C}$; in other words, a reaction is applicable if $\config{C}$ has at least as many copies of each species as $\reactants$. If the reaction $(\reactants,\products)$ is applicable, it results in configuration $\config{C'} = \config{C} - \reactants + \products$ if it occurs, and we write $\config{C} \rightarrow_{crn}^{(\species,\reactions)} \config{C'}$, or simply $\config{C} \stepsto \config{C'}$ when the model and CRN are clear from context. We use the same notation as configuration vectors to denote the size ($\size{\reactants}$ and $\size{\products}$) and explicit content ($\{\reactants\}$ and $\{\products\}$) of reactants and products for a reaction.

\begin{definition}[Discrete Chemical Reaction Network]  A discrete chemical reaction network (CRN) is an ordered pair $(\species, \reactions)$ where $\species$ is an ordered alphabet of species, and $\reactions$ is a set of rules over $\species$.
\vspace{-.1cm}
\end{definition}
\begin{definition}[Basic CRN Dynamics]  For a CRN $(\species, \reactions)$ and configurations $\config{A}$ and $\config{B}$, we say that $\config{A} \rightarrow_{crn}^{(\species,\reactions)} \config{B}$ if there exists a rule $(\reactants,\products)\in \reactions$ such that $\reactants \leq \config{A}$, and $\config{A} - \reactants + \products = \config{B}$.
\end{definition}

If there exists a finite sequence of configurations such that $\config{C} \stepsto \config{C}_1 \stepsto \dots \stepsto \config{C}_n \stepsto \config{D}$, then we say that $\config{D}$ is \textit{reachable} from $\config{C}$ and we write $\config{C} \reaches \config{D}$.
A configuration is said to be \textit{terminal} if no reactions are applicable. We also define an \emph{initial configuration} for a CRN as its starting configuration. A \emph{CRN System} $T$ is then defined as a pair of a CRN model and its initial configuration. The following sections define extensions of the basic CRN model by way of defining modified dynamics.

\subsubsection{Void Genesis CRNs}\label{subsec:vg}

A \emph{Void Genesis CRN} $\crn{C_{VG}} = ((\species, \reactions), z)$ is a basic CRN with a \emph{zero} species $z \in \species$ whose count is incremented whenever the count of any species other than $z$ goes to zero. See Figure \ref{fig:void_genesis_example} for an example.

\begin{definition}[Void-Genesis Dynamics]
    For a Void-Genesis CRN $((\species, \reactions), z\in \species)$ and configurations $\config{A}$, $\config{B_t}$ and $\config{B}$, we say that $\config{A} \rightarrow_{\crn{C_{VG}}}^{(\species, \reactions)} \config{B} $ if there exists a rule $(\reactants,\products) \in \reactions$ such that $\reactants \leq \config{A}$, $\config{A}-\reactants+\products=\config{B_t}$, and $\config{B} = \config{B_t} + n\cdot \single{z}$ where n is the cardinality of \{$\lambda \in \species\setminus\{z\}$ $|$ $\config{A}[\lambda] \not= 0$ and $\config{B_t}[\lambda]=0$\}.
\end{definition}

It is straightforward to show that the Void Genesis model is Turing-universal via simulating a register machine, and we do so in Section \ref{sub:undecide}.

\subsubsection{Inhibitory CRNs}\label{subsec:inhib}

A reaction $\reaction$ is said to be \textit{inhibited} by a species $\lambda$ when the reaction $\reaction$ may only be applied if $\lambda$ is absent in the system. We define an inhibitor mapping $\mathcal{I} : \reactions \rightarrow \mathbb{P}(\species)$ that maps a reaction to a subset of species that inhibit the reaction. An \emph{Inhibitory CRN} $\crn{C_{IC}} = ((\species, \reactions), \mathcal{I})$ as defined by \cite{calabrese2024inhibitory} is then a basic CRN along with the mapping $\mathcal{I}$. See Figure \ref{fig:inhibitory_example} for an example.

\begin{definition}[Inhibitory Dynamics]\label{def:inhibitory-dynamics}
    For a Inhibitory CRN $((\species, \reactions), \mathcal{I})$ and configurations $\config{A}$ and $\config{B}$, we say that $\config{A} \rightarrow_{\crn{C_{IC}}}^{(\species, \reactions)} \config{B}$ if there exists a rule $\reaction = (\reactants, \products) \in \reactions$ such that $\reactants \leq \config{A}$, $\config{A} -\reactants+\products=\config{B}$, and $A[\lambda] = 0, \forall \lambda \in \mathcal{I}(\reaction)$.
\end{definition}

\subsubsection{Coarse-Rate CRNs} \label{subsec:coarse}

A \emph{Coarse-Rate CRN} $\crn{C}_{CR} = ((\species, \reactions), rank)$ as introduced by \cite{shea2010writing} is a basic CRN along with a function $rank : \reactions \rightarrow \N$. We define a set of reactions $\reactions^l$ as the set of all reactions $\reaction$ where $rank(\reaction) = l$. The set of reactions $\reactions$ is then defined as an ordered partition set given by $\{\reactions^1, \reactions^2, \ldots, \reactions^n\}$. Any applicable reaction $\reaction_i^\ell$ may only be applied if no reaction $\reaction_j^k \in \reactions^k, \forall k \in [\ell+1, n]$ is applicable. We use $(\reactants, \products)^\ell$ to denote a reaction $\reaction^\ell \in \reactions^\ell$. In the context of this paper we focus on models with rank at most $2$. For clarity, we will refer to reactions with rank $2$ as \emph{fast} reactions, and the ones with rank $1$ as \emph{slow} reactions. See Figure \ref{fig:priority_example} for an example.

\begin{definition}[Coarse-Rate Dynamics]
For a Coarse-Rate CRN $((\species, \reactions), rank)$ and configurations $\config{A}$ and $\config{B}$, we say that $\config{A} \rightarrow_{\crn{C_{CR}}}^{(\species,\reactions)} \config{B}$ if there exists a rule $(\reactants, \products)^\ell \in \reactions$ such that $\reactants \leq \config{A}$, $\config{A}-\reactants+\products=\config{B}$, and $\not\exists$ an applicable reaction $\reaction^{k > \ell}$.
\end{definition}

\subsubsection{Step-Cycle CRNs}\label{subsec:step-cycle}
A \emph{step-cycle CRN} is a step CRN \cite{anderson2024steps} that infinitely repeats steps 0 through $k-1$: that is, once $\config{S}_{k-1}$ is added to the terminal configuration in the ${(k-1)}^{th}$ step, the resulting configuration is treated as the new initial configuration for the step CRN. More formally, a \emph{step-cycle CRN} of $k$ steps is an ordered pair $( (\species, \reactions), (\config{S}_0, \config{S}_1, \ldots, \config{S}_{k-1}))$, where the first element is a normal CRN $(\species, \reactions)$ and the second is a sequence of length-$|\species|$ vectors of non-negative integers denoting how many copies of each species type to add after each step. We define a \textit{step-configuration} $\config{C}_i$ for a step-cycle CRN as a valid configuration $\config{C}$ over $(\species, \reactions)$ along with an integer $i \in \{0,\ldots, k-1\}$ that denotes the configuration's step.

\begin{definition}[Step-Cycle Dynamics]
For a step-cycle CRN $((\species, \reactions), (\config{S}_0, \config{S}_1, \ldots, \config{S}_{k-1}))$ and step-configurations $\config{A}_i$ and $\config{B}_j$, we say that $\config{A}_i \rightarrow_{\crn{C_{SC}}}^{(\species,\reactions)} \config{B}_j$ if either
    \begin{enumerate}
        \item $i=j$, there exists a rule $(\reactants, \products) \in \Gamma$ s.t. $\reactants \leq \config{A}_i$, and $\config{A}_i-\reactants+\products=\config{B}_j$, or
        \item $(i+1) \mod k=j$, $\config{A}_i$ is terminal, and $\config{A}_i + \config{S}_i=\config{B}_j$.
    \end{enumerate}
\end{definition}

\subsubsection{Unique Instruction Parallel CRNs}\label{subsec:parallel}
The Unique Instruction Parallel model modifies the dynamics of a normal CRN $(\species, \reactions)$ by applying a maximal set of compatible rules as a single transition.  In this paper, we restrict this maximal set to contain only one application of any given rule, leaving the study of more relaxed parallel models for future work.

\begin{definition}[Plausibly Parallel Rules] A multiset of $n$ (not necessarily distinct) rules $\{(\reactants_1, \products_1),\ldots, (\reactants_n,\products_n)\}$ are plausibly parallel for a configuration $\config{C}$ over $\species$ if the vector $\reactants=\sum_{i=1}^{n} \reactants_i$ is such that $\reactants \leq \config{C}$.
\end{definition}

\begin{definition}[Unique-Instruction Plausibly Parallel] A plausibly parallel multiset is said to be Unique-Instruction if it contains at most one copy of any given rule (i.e., it is a set).  It is considered unique-instruction maximal if it is not a proper subset of any other unique-instruction plausibly parallel set.
\end{definition}

\begin{definition}[Unique-Instruction Parallel Dynamics]\label{def:unique-dynamic}
For a CRN $(\species, \reactions)$ and configurations $\config{A}$ and $\config{B}$, we say that $\config{A} \rightarrow_{\crn{C_{UI}}}^{(\species,\reactions)} \config{B}$ if there exists a unique-instruction maximal plausibly parallel set $\{ (\reactants_1, \products_1), \ldots , (\reactants_k, \products_k)\}$ for configuration $\config{A}$ and rule set $\reactions$ such that $\config{B}= \config{A}-\sum_{i=1}^{k} \reactants_i + \sum_{i=1}^{k} \products_i$.
\end{definition}

\begin{figure}[t]
\vspace{-.2cm}
    \centering
    \begin{subfigure}{1\textwidth}
        \centering
        \includegraphics[width=0.85\textwidth]{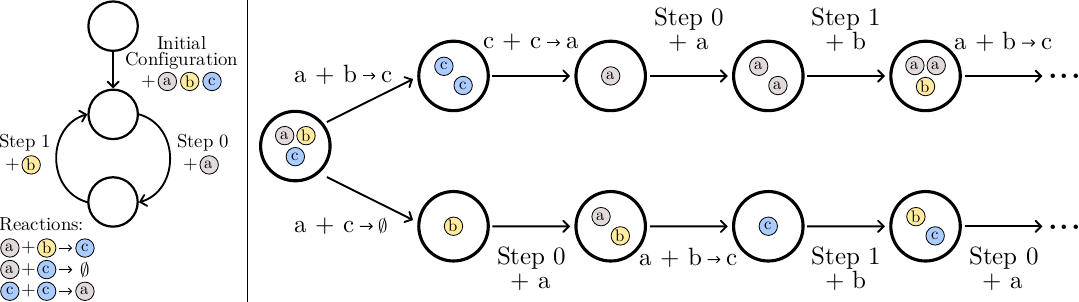}
        \subcaption{Step-Cycle Example}
        \label{fig:step_cycle_example}
    \end{subfigure}
    \begin{subfigure}{.47\textwidth}
        \centering
        \includegraphics[width=.8\textwidth]{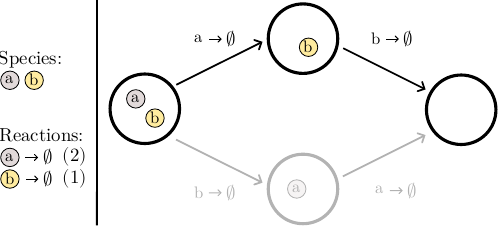}
        \subcaption{Coarse-Rate Example}
        \label{fig:priority_example}
    \end{subfigure}
    \begin{subfigure}{.47\textwidth}
        \centering
        \includegraphics[width=.8\textwidth]{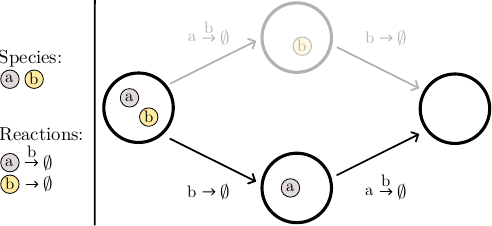}
        \subcaption{iCRN Example}
        \label{fig:inhibitory_example}
    \end{subfigure}
    \begin{subfigure}{.47\textwidth}
        \centering
        \includegraphics[width=.9\textwidth]{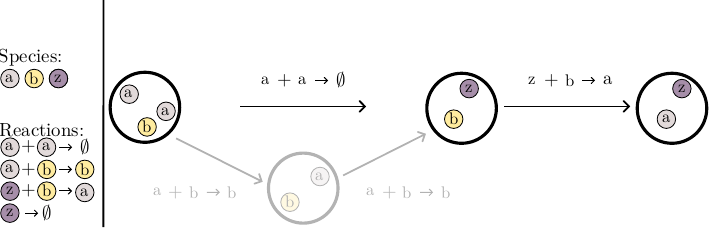}
        \subcaption{Void Genesis Example}
        \label{fig:void_genesis_example}
    \end{subfigure}
    \begin{subfigure}{.47\textwidth}
        \centering
        \includegraphics[width=.8\textwidth]{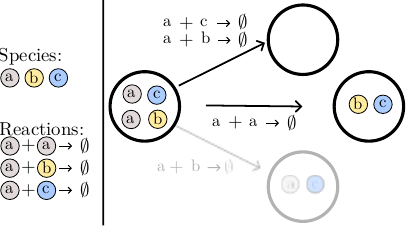}
        \subcaption{Unique Instruction Parallel Example}
        \label{fig:UI-example}
    \end{subfigure}
    \vspace{-.2cm}
    \caption{Example systems for the 5 CRN models. The blurred portions of the figure represent invalid reaction sequences. (a) Step-cycle with $2$ steps. Species are added when configurations reach a terminal state. (b) Coarse-rate. The numbers next to each reaction denote the rank. The bottom reaction sequence is invalid as the reaction with rank $2$ must occur first. (c) Inhibitory. The top reaction sequence is invalid as there exists a $b$ in the initial configuration, which is an inhibitor for the rule $a \rightarrow \emptyset$. (d) Void Genesis. Once the $a$ species reaches a count of 0, the $z$ species is created. In the bottom reaction sequence, the $z$ species reacts with the $b$ species, and another $z$ species is created. (e) Unique-Instruction Parallel. The bottom reaction sequence is invalid as the rule $a + b \rightarrow \emptyset$ is not a maximal set.}
    \label{fig:crn_models}
    \vspace{-.2cm}
\end{figure}

\subsection{Simulation}\label{subsec:sim}
By way of Petri nets, discrete CRNs have seen various model extensions.
To meaningfully compare the computational capabilities of these variants, we turn to the notion of simulation, which serves as a tool to compare the relative expressive power of each model.  
However, existing definitions in the literature vary in scope and applicability.
Some emphasize strict structural correspondence while others focus purely on dynamic behavior.  
Thus, it is worthwhile to discuss why we formulate our own definition of simulation and equivalence.

Borrowed from classical process theory, (strong) \emph{bisimulation} \cite{milner1989communication,fernandez1990implementation,antoniotti2004taming} is perhaps the most stringent form of equivalence.
It requires that for every state and transition in one system, there exists a matching state and transition in the other, and vice versa.
This strong bidirectional constraint means bisimulation ensures both behavioral and structural fidelity, and notably also implicitly guarantees efficiency.

Weak bisimulation \cite{dong2012bisimulation,johnson2019verifying}, on the other hand, relaxes the strict step-by-step matching of bisimulation.
Instead, a transition in one system may correspond to a \textit{macrotransition} in the other: a sequence of transitions possibly allowing ``hidden'' or ``silent'' intermediate steps.
This makes bisimulation more flexible and applicable to realistic implementations, but it is not without its own limitations.
Pathway decomposition, as presented in \cite{shin2012compiling,shin2019verifying}, takes a different approach.
Rather than comparing reactions directly, they identify a CRN's \emph{formal basis} --- a set of ``indivisible'' formal pathways that collectively define its behavior.
This allows pathway decomposition to capture phenomena such as \emph{delayed choice}, in which nondeterministic behavior is distributed across multiple steps.

Our framework seeks to strike a middle ground.
We define simulation in terms of a configuration map and macrotransitions, retaining the core idea of weak bisimulation but allowing for a more general structural correspondence between systems.
Since this correspondence is so general, we explicitly include a measure of efficiency with our definition.
We define polynomial efficient simulation that permits abstraction and internal nondeterminism, like weak bisimulation and pathway decomposition, but captures both dynamic behavior and bounded structural transformation.
We say two systems are polynomially equivalent if they can each simulate each other via our definition.
This is analogous to bisimilarity \cite{milner1989communication}, but in the context of efficient, weak simulation.

\para{Simulation Definition}
To define the concept of one CRN system $T'$ simulating another CRN system $T$ we introduce \emph{configuration maps} and \emph{representative configurations}.  A configuration map is a polynomial-time computable function $\cmap: \textrm{configs}_{T'} \rightarrow \textrm{configs}_{T} \bigcup \{\undefConfig\}$ mapping at least one element in $\textrm{configs}_{T'}$ to each element in $\textrm{configs}_{T}$, and the \emph{representative} configurations for a configuration $\config{C}\in \textrm{configs}_{T}$ are $[\![\config{C}\,]\!] = \{ \config{C'} \mid \config{C}=\cmap(\config{C'})\}$, also computable in polynomial time.\footnote{We write $M(\config{C'}) = \undefConfig$ to mean that the mapping is undefined, which is not the same as $M(\config{C'}) = \config{0}$.} Finally, we define the concept of \emph{single-step transition}, $\config{A} \rightarrow_{T} \config{B}$, to mean that $\config{A}$ transitions to $\config{B}$ in system $T$. And, the concept of \emph{macro transitionable}, $\config{A'} \macrotrans_{T'} \config{B'}$, to mean that $\config{B'}$ is reachable from $\config{A'}$ in system $T'$ through a sequence of $k$ intermediate configurations $\langle \config{A'}, \config{X_1'},\ldots, \config{X_k'}, \config{B'}\rangle$ 
such that $M(\config{A'}) \neq \undefConfig$, $M(\config{B'}) \neq \undefConfig$, and each $M(\config{X_i'}) \in \{ M(\config{A'}), \undefConfig\}$. 
Note that $k$ can be zero, resulting in the sequence $\langle \config{A'}, \config{B'} \rangle$.

\para{Intuition}  Each representative configuration set $[\![\config{C}\,]\!]$ is a particular collection (of at least 1 configuration) that adheres to the strictest modeling of system $T$ within system $T'$. That is, any $\config{C'} \in [\![\config{C}\,]\!]$ must be able to grow into anything that $\config{C}=M(\config{C'})$ can grow into, and no $\config{C'} \in [\![\config{C}\,]\!]$ may grow into something that $\config{C}=M(\config{C'})$ cannot grow into.

For a given configuration map and collection of representative configurations, we define the concepts of \emph{following} and \emph{modeling}, followed by our definition of \emph{simulation}.

\begin{definition}[It Follows]  We say system $T$ follows system $T'$ if whenever $\config{A'}\macrotrans_{T'} \config{B'}$ and $\cmap(\config{A'})\neq \cmap(\config{B'})$, then $\cmap(\config{A'})\rightarrow_T \cmap(\config{B'})$.
\end{definition}

\begin{definition}[Models] We say system $T'$ models system $T$ 
    if $\config{A} \rightarrow_T \config{B}$ implies that $\forall \config{A'}\in [\![\config{A}\,]\!]$, $\exists \config{B'}\in [\![\config{B}\,]\!]$ such that $\config{A'} \Rightarrow_{T'} \config{B'}$.
\end{definition}

\begin{definition}[Simulation]
A system $T'$ simulates a system $T$ if there exists a polynomial time computable function $M: \textrm{configs}_{T'} \rightarrow \textrm{configs}_T$ and polynomial time computable set $[\![\config{C}\,]\!] = \{ \config{C'} \mid \cmap(\config{C'}) = \config{C}\}$ for each $\config{C} \in \textrm{configs}_T$, such that:

\vspace{-.1cm}

\begin{enumerate} \setlength\itemsep{0em}
    \item $T$ follows $T'$.
    \item $T'$ models $T$.
\end{enumerate}
\end{definition}

\begin{figure}[t]
\vspace{-.2cm}
    \centering
    \includegraphics[width=0.75\textwidth]{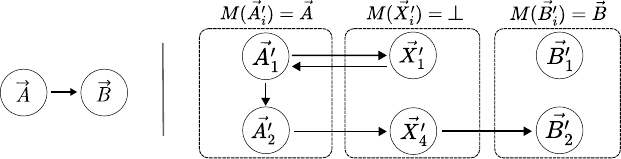}
    \vspace{-.2cm}
    \caption{(Left) A system $T$ with states $\config{A}$ and $\config{B}$, and transition $\config{A} \rightarrow_T \config{B}$. (Right) A state-space diagram for system $T'$ that simulates $T$. Here, each arrow represents some transition $\config{X} \rightarrow_{T'} \config{Y}$ under the dynamics of $T'$. Observe how $T$ follows $T'$ and $T'$ models $T$.}
    \label{fig:general_simulation}
    \vspace{-.2cm}
\end{figure}

\para{Polynomial Simulation}  
We say a simulation is \emph{polynomial efficient} if the simulating system adheres to the following:
\vspace{-.1cm}

\begin{description} \setlength\itemsep{0em}
    \item [polynomial species and rules.] The number of species and rules is at most polynomial in the number of species and rules of the simulated system.
    \item [polynomial rule size.] The maximum rule size (number of products plus number of reactants) is polynomial in the maximum rule size of the simulated system.
    \item [polynomial transition sequences.] For all $B$ such that $A \rightarrow_T B$, the expected number of transitions taken to perform a macro transition from $M(A) \macrotrans_{T'} M(B)$, conditioned that $M(A)$ does macro transition to $M(B)$, has expected number of transitions polynomial in the number of rules and species of the simulated system based on a uniform sampling of applicable rules.
    \item [polynomial volume.] Each $\config{C'} \in [\![\config{C}\,]\!]$ has a volume that is polynomially bounded by the volume of $C$, and for any macro transition $A' \macrotrans_{T'} B'$, any intermediate configuration within this macrotransition has volume polynomially bounded in the volume of $M(A')$ and $M(B')$.
\end{description} 

\begin{restatable}{theorem}{transitivity} \label{thm:transitivity}
    (\textbf{Transitivity}.) Given three CRN systems $T_1, T_2$ and $T_3$ such that $T_2$ simulates $T_1$ under polynomial simulation, and $T_3$ simulates $T_2$ under polynomial simulation, then $T_3$ simulates $T_1$ under polynomial simulation.
\end{restatable}
\begin{proof}
    Given three CRN systems $T_1$, $T_2$ and $T_3$ such that $T_2$ simulates $T_1$ under polynomial simulation, and $T_3$ simulates $T_2$ under polynomial simulation. Let $\cmap_{21} : \textrm{configs}_{T_2} \rightarrow \textrm{configs}_{T_1}$ and $\cmap_{32} : \textrm{configs}_{T_3} \rightarrow \textrm{configs}_{T_2}$ be the polynomial-time computable configuration mappings. We define a configuration mapping function $\cmap_{31} : \textrm{configs}_{T_3} \rightarrow \textrm{configs}_{T_1}$ by composing functions $\cmap_{21}$ and $\cmap_{32}$ as follows. 

    \[
        \cmap_{31}(\config{C_3}) =
        \begin{cases}
            (\cmap_{21} \circ \cmap_{32})(\config{C_3}) & \text{if}\; \cmap_{32}(\config{C_3})\not=\phi \wedge \cmap_{21}(\cmap_{32}(\config{C_3}))\not=\phi\\
            \phi & \text{otherwise}
        \end{cases}
    \]

    And the set of representative configurations for a configuration $\config{C_1}\in \textrm{configs}_{T_1}$ as $[\![\config{C_1}\,]\!] = \{ \config{C_3} \mid \config{C_1}=\cmap_{31}(\config{C_3})\}$. If the configuration $\config{C_1}$ has a non-empty set of representative configurations in $T_2$, and each of those configurations have representative configurations in $T_3$, then this set will contain all such configurations. Therefore, if both $\cmap_{21}$ and subsequently $\cmap_{32}$ are defined, this set will be non-empty.
    
    We now show that $T_3$ simulates $T_1$ by proving that $T_1$ follows $T_3$ and $T_3$ models $T_1$ for the configuration mapping and representative configurations defined above. We then show that the simulation is polynomial efficient.

    \para{$T_1$ follows $T_3$} $T_1$ follows $T_3$ if for any two configurations $\config{A_3}$ and $\config{B_3}$ in $T_3$ where $\cmap_{31}(\config{A_3})$ and $\cmap_{31}(\config{B_3})$ are defined, such that $\config{A_3} \macrotrans_{T_3} \config{B_3}$, and $\cmap_{31}(\config{A_3}) \not= \cmap_{31}(\config{B_3})$, then $\cmap_{31}(\config{A_3}) \rightarrow_{T_1} \cmap_{31}(\config{B_3})$. For these configurations if $\cmap_{32}(\config{A_3}) \not= \cmap_{32}(\config{B_3})$ then $\cmap_{32}(\config{A_3}) \rightarrow_{T_2} \cmap_{32}(\config{B_3})$. This is true because $T_2$ follows $T_3$.\\
    Because $T_1$ follows $T_2$, for any two configurations $\config{A_2}$ and $\config{B_2}$ in $T_2$, if $\config{A_2} \macrotrans_{T_2} \config{B_2}$, and $\cmap_{21}(\config{A_2}) \not= \cmap_{21}(\config{B_2})$, then $\cmap_{21}(\config{A_2}) \rightarrow_{T_1} \cmap_{21}(\config{B_2})$. A single-step transition is simply a macro-transition with one step. Therefore, if $\cmap_{21}(\cmap_{32}(\config{A_3})) \not= \cmap_{21}(\cmap_{32}(\config{B_3}))$, then $\cmap_{21}(\cmap_{32}(\config{A_3})) \rightarrow_{T_1} \cmap_{21}(\cmap_{32}(\config{B_2}))$. When $\cmap_{31}$ is defined we can infer based on the definition of $\cmap_{31}$ that, if $\cmap_{31}(\config{A_3}) \not= \cmap_{31}(\config{B_3})$ then $\cmap_{31}(\config{A_3}) \rightarrow_{T_1} \cmap_{31}(\config{B_3})$. Therefore, $T_1$ follows $T_3$.

    \para{$T_3$ models $T_1$} $T_3$ models $T_1$ if for any two configurations $\config{A_1}$ and $\config{B_1}$ in $T_1$ such that $\config{A_1}\rightarrow_{T_1}\config{B_1}$, $\forall \config{A_3}\in[\![\config{A_1}\,]\!]$, $\exists\config{B_3}\in[\![\config{B_1}\,]\!]$ under $\cmap_{31}$ such that $\config{A_3}\macrotrans_{T_3}\config{B_3}$. For all such configurations $\config{A_1}$ and $\config{B_1}$ in $T_1$, $\forall \config{A_2}\in [\![\config{A_1}\,]\!], \exists \config{B_2}\in [\![\config{B_1}\,]\!]$ under $\cmap_{21}$ such that $\config{A_2}\macrotrans_{T_2}\config{B_2}$ because $T_2$ models $T_1$. This macro transition $\config{A_2}\macrotrans_{T_2}\config{B_2}$ is represented as a sequence of single-step transitions $\config{X_2}^i\rightarrow_{T_2}\config{X_2}^{i+1}$ where $\cmap_{21}(\config{X_2}^i), \,\cmap_{21}(\config{X_2}^{i+1}) \in \{\config{A_1}, \phi \}$ for $0 \leq i < k$ as shown in Figure \ref{fig:models-transitivity}. Furthermore, $\forall \config{X_3^j}\in[\![\config{X_2}^i\,]\!]$, $\exists \config{X_3}^\ell\in[\![\config{X_2}^{i+1}\,]\!]$ under $\cmap_{32}$ such that $\config{X_3}^j\macrotrans_{T_3}\config{X_3}^\ell$ because $T_3$ models $T_2$. 
    
    Because $\cmap_{32}(\config{X_3}^j)=\config{X_2}^i$ and $\cmap_{32}(\config{X_3}^\ell)=\config{X_2}^{i+1}$ and $\cmap_{21}(\config{X_2}^i), \,\cmap_{21}(\config{X_2}^{i+1}) \in \{\config{A_1}, \phi \}$, therefore, $\cmap_{31}(\config{X_3}^j), \,\cmap_{31}(\config{X_3}^\ell) \in \{\config{A_1}, \phi \}$ for $0 \leq j < k$ and $j < \ell \leq k$. As shown in Figure \ref{fig:models-transitivity}, the sequence of macro-transitions $\config{X_3}^j\macrotrans_{T_3}\config{X_3}^\ell$ can be represented as a single macro-transition $\langle\config{A_3},\config{X_3}^1,\ldots,\config{X_3}^k,\config{B_3}\rangle$, where $\cmap_{31}(\config{X_3}^j)\in\{\config{A_1}, \phi\}$ for $0\leq j \leq k$. Also, $\cmap_{32}(\config{A_3})=\config{A_2}$ and $\cmap_{32}(\config{B_3})=\config{B_2}$. We know that $\cmap_{21}(\config{A_2})=\config{A_1}$ and $\cmap_{21}(\config{B_2})=\config{B_1}$, therefore, $\cmap_{31}(\config{A_3})=\config{A_1}$ and $\cmap_{31}(\config{B_3})=\config{B_1}$ when both $\cmap_{21}$ and $\cmap_{32}$ are defined. Hence, such a macro-transition $\config{A_3}\macrotrans_{T_3}\config{B_3}$ models the single-step transition $\config{A_1}\rightarrow_{T_1}\config{B_1}$ when $\config{A_3}\in[\![\config{A_1}\,]\!]$ and $\config{B_3}\in[\![\config{B_1}\,]\!]$. Therefore, $T_3$ models $T_1$.
    \begin{figure}[t!]
        \centering
        \includegraphics[width=0.7\linewidth]{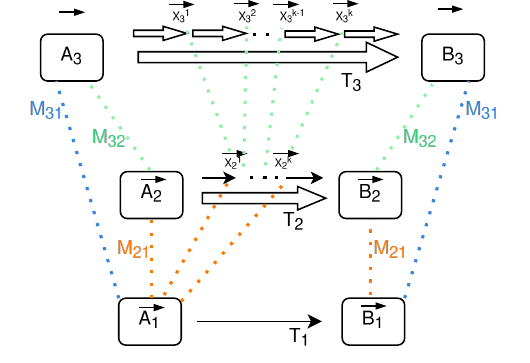}
        \caption{A single-step transition $\config{A_1}\rightarrow_{T_1}\config{B_1}$ is simulated using a sequence of transitions starting at $\config{A_2}$ through $\config{B_2}$ in $T_2$. Each of these single-step is represented using a sequence of transitions in $T_3$.}
        \label{fig:models-transitivity}
    \end{figure}
    
    \para{Polynomial Simulation} Given that $T_2$ simulates $T_1$ under polynomial simulation, and $T_3$ simulates $T_2$ under polynomial simulation, we now show that $T_3$ simulates $T_1$ under polynomial simulation. Let $\species_i$ and $\reactions_i$ be the set of species and reactions for the model $T_i$ for $1\leq i\leq 3$.
    \begin{enumerate}
        \item \textbf{polynomial species and rules:} We know that $\size{\species_2} = \mathcal{O}(\size{\species_1}) = \size{\species_1}^{c_1}$ and $\size{\species_3} = \mathcal{O}(\size{\species_2}) = \size{\species_2}^{c_2}$, therefore, $\size{\species_3} = \size{\species_1}^{c_1\cdot c_2}$. Similarly, $\size{\reactions_3} = \size{\reactions_1}^{c_1\cdot c_2}$. Therefore, $\size{\species_3} = \mathcal{O}(\size{\species_1})$ and $\size{\reactions_3} = \mathcal{O}(\size{\reactions_1})$.
        \item \textbf{polynomial rule size:} Let $\RM_i$ be the largest rule in the model $T_i$. We know that $\size{\RM_2} = \mathcal{O}(\size{\RM_1})$ and $\size{\RM_3} = \mathcal{O}(\size{\RM_2})$, therefore, $\size{\RM_3} = \mathcal{O}(\size{\RM_1})$.
        \item \textbf{polynomial transition sequences:} The length of the transition sequence for a macro-transition in $T_3$ is bounded by $\mathcal{O}(\size{\species_2} + \size{\reactions_2})$. Given that $\size{\species_2}=\size{\species_1}^{c_1}$ and $\size{\reactions_2}=\size{\reactions_1}^{c_1}$, therefore, each macro-transition in $T_3$ simulating a single-step in $T_1$ uses $\mathcal{O}(\size{\species_1}+\size{\reactions_1})$ single-step transitions.
        \item \textbf{polynomial volume:} For a macro-transition $\config{A_2}\macrotrans_{T_2}\config{B_2}$ representing a transition $\config{A_1}\rightarrow_{T_1}\config{B_1}$, the volume of representative configurations $\config{A_2}$ and $\config{B_2}$ is polynomial in the volume of configurations $\config{A_1}$ and $\config{B_1}$ respectively. Similarly, for a macro-transition $\config{A_3}\macrotrans_{T_3}\config{B_3}$ as shown in Figure \ref{fig:models-transitivity}, the volume of $\config{A_3}$ and $\config{B_3}$ is polynomial in volume of $\config{A_2}$ and $\config{B_2}$ respectively. Therefore, the volume of representative configurations $\config{A_3}$ and $\config{B_3}$ is polynomial in the volume of configurations $\config{A_1}$ and $\config{B_1}$, where $\config{A_3}\macrotrans_{T_3}\config{B_3}$ models $\config{A_1}\rightarrow_{T_1}\config{B_1}$. The volume of all intermediate configurations are bound by $\config{A_3}$ and $\config{B_3}$.
    \end{enumerate}
\end{proof}
\section{Detecting Zero: Examples in Each Model}\label{sec:zeros}

The power each extended CRN model has over the basic CRN model stems from the ability to detect when a species has reached a count of zero. Thus, the focus of this section is twofold: to detail how each model is capable of detecting zero and to provide minimum working examples (MWEs) for ease of comprehension. 
For all examples, we want to check if species $s$ is in the system. The idea is that there is a species $c_s$ that turns to $n_s$ if there are no $s$’s and $y_s$ if there are $s$’s in the system while nothing else in the system changes.

\para{Void Genesis}
Detecting zero is explicitly done by the model. For the simplest version of the model, we assume we have no $z$ species.
To check whether $s$ is zero, simply use either as a catalyst (examples on the left side).
For this to work, you must maintain a single $z$. Thus, when producing any species $s$, continue to check with $s$ as a catalyst and $z$ as a reactant. As an example, the rule $x \rightarrow s$ would be modified to be implemented with the two rules on the right side.

\begin{tabular}{  p{.5\textwidth} p{.3\textwidth}  }
$c_s + s \rightarrow y_s + s$ & $x + s \rightarrow 2s$\\
$s_s + z \rightarrow n_s + z$ & $x + z \rightarrow s$\\
\end{tabular}

\para{Inhibitory CRNs}
We can use an inhibited rule and a catalyst to detect whether a species $s$ exists or not.

\begin{tabular}{  p{2.5cm} l }
 $c_s \xrightarrow{s} n_s$ & runs when $s$ does not exist \\
 $c_s + s \rightarrow y_s + s$  & runs when $s$ does exist\\
\end{tabular}

\para{Course-Rate CRNs}
Detecting a species $s$ only requires that there is a fast rule that uses it and a slow rule that does not. 

\begin{tabular}{  p{2.5cm} p{10cm}  }
    Fast Reaction & $c_s + s \rightarrow y_s + s$ (will always execute first)\\
    Slow Reaction & $c_s \rightarrow n_s$ (will only execute if no $s$ exists in the system)\\
\end{tabular}

\para{Step-Cycle CRNs}
Detecting zero is simple by using $s$ in a reaction and then going to another step. Since each step reaches a terminal configuration, it must not exist if the reaction did not occur.

\vspace{.2cm}
\begin{tabular}{ | c | p{5.5cm} | p{1.5cm} | p{4cm} | }
\hline
  Step & Description & Add & Rules \\ \hline
   
  1 & Add something that only reacts with $s$ & a single $c_s$ & $s + c_s\rightarrow y$\\ \hline
  2 & Check if it reacted & a single $w$ 
  	& 	$w + y \rightarrow y_s + s$ ($s$ exists)\fillhor\
        $w + c_s \rightarrow n_s$ ($s$ not exists) 
       \\ \hline
\end{tabular}

\para{Unique-Instruction Parallel CRNs}
The parallel model UI takes advantage of how rules are applied to force reactions to run in a specific order depending on whether the count of a certain species is zero or not. 
Since we are limiting the possible rules that can run by our species selection, detecting zero can be done in this manner. The intuition is to run two independent rules, followed by a rule that uses the output of both. The `Round' indicates that all rules in this round would execute before the next round due to the maximal selection.

\vspace{.2cm}
\begin{tabular}{| c | p{5cm} | p{6cm} |  }
    \hline
  Round & Description & Rules \\ \hline
  1 & Create a timing/clock species ($t_i$) and a species to use $s$ with. & $c_s \rightarrow r_c + t_1$\\ \hline
  
  2 & Try to use $s$ and use the timer $t$ in another rule.
  	&  $r_c + s \rightarrow r_s$ (can only run if $S$ exists)\fillhor\
       $t_1 \rightarrow t_2$ (runs even if the other rule can not)\\ \hline 
        
  3 & Now we use the timer to see if the other rule ran 
  	&  $t_2 + r_s \rightarrow y_s + s$ (there is an $s$ in the system) \fillhor\
    $t_2 + r_c \rightarrow n_s$ (there are no $s$'s in the system)\\ \hline
\end{tabular}

\section{Equivalence}\label{sec:equivalence}

This section shows equivalence between the CRN models, represented as the purple bi-directional arrows in Figure \ref{fig:vg_hub}. We first introduce a more general Void Genesis model, along with some useful techniques, before presenting our equivalence results. Due to the number of results and limited space, construction details and formal proofs are placed in Appendix \ref{constructions}.


\subsection{Equivalence Preliminaries}\label{sub:technique}

\subsubsection{$k$-Void Genesis} \label{sub:kvg-def}
Due to the complexity of some of the simulations, we first provide a more general version of the VG model that makes simulation easier.
We will prove that the standard Void Genesis model can still simulate the more general model with at most a polynomial blow-up in rules, species, and expected rule applications. Essentially, the only difference is that rather than a single zero species $z$, there can be a different zero species for every species.

Formally, a \emph{$k$-Void Genesis CRN} $\crn{C_{KVG}} = ((\species, \reactions), Z_{\emptyset})$ is a CRN with a partial mapping function $\Zero:\species_1 \rightarrow \species_2$, such that $\species_1 \cup \species_2 = \species$ and $\species_1 \cap \species_2 = \emptyset$, 
that indicates which species is created whenever another species count goes to zero (if mapped). The partition creates a distinction between \emph{normal} species and the special \emph{zero-counting} species, which eliminates chaining effects that could arise with zero-species being created from the elimination of other zero-species. For convenience, we use the notation $\Zero(\lambda)\rightarrow z_\lambda$ to indicate that if the count of species $\lambda$ in the system goes to zero, a $z_\lambda$ species is created.

\begin{definition}[$k$-Void-Genesis Dynamics]\label{def:kvdynamics}
    For a $k$-Void-Genesis CRN $((\species, \reactions), \Zero)$ and configurations $\config{A}$, $\config{B_t}$ and $\config{B}$, we say that $\config{A} \rightarrow_{\crn{C_{KVG}}}^{(\species, \reactions)} \config{B} $ if there exists a rule $(\reactants,\products) \in \reactions$ such that $\reactants \leq \config{A}$, $\config{A}-\reactants+\products=\config{B_t}$, and $\config{B} = \config{B_t} + \config{Z}$ where $\config{Z} = \sum_{\lambda \in C} \single{z}_\lambda$, and $C =\{\lambda \in \species | \config{A}[\lambda] \not= 0$ and $\config{B_t}[\lambda]=0\}$. 

\end{definition}

\subsubsection{Techniques} \label{subsec: techniques}
One technique that is used in our simulations is to maintain a single global leader species $G$ that selects which rule $\mathcal{G}_i$ to execute, and then a sequence of rules to either sequentially consume the reactants or sequentially break down the product of the other models' rules to execute zero-checking before managing zero species. 
Table \ref{tab:react_break} describes both of these processes.

\para{Notation}
To simplify the proofs and for consistency, we use the following notation for rules. For rule $i$, we have $\mathcal{G}_i=(\config{\RM_i},\config{\PM_i})$. We make use of the $\Set{\RM}$ and $|\RM|$ notation defined in the preliminaries, as well the difference between a configuration with many species $\config{X}$ versus a configuration with only a single species $\single{X}$ by the over arrow used. 

\begin{table}[t]
\vspace{-.2cm}
        \centering
        \begin{tabular}{|@{} l @{}l @{}|@{}l@{}|@{} l @{}l@{} |}\cline{1-2}\cline{4-5}
            & $\single{G} + \config{\RM_i} \rightarrow \single{R}_i^{1} + \config{\RM_i}$ & &
            & $\single{G} + \config{\RM_i}\rightarrow \single{P}_i^1$
            \\ \cline{1-2}\cline{4-5}

            & $\single{R}_i^j + \single{r}_j \rightarrow \single{C}_i^{j}$ (check for zero) & &
            $\forall p_j \in \PM_i:$ & $\single{z}_{p_j} + \single{P}_i^j \rightarrow \single{p}_j + \single{P}_i^{j+1}$ (remove zero)
            \\ 
            
           $\forall r_j \in \RM_i:$ & $\single{C}_i^j + \single{z}_{r_j} \rightarrow \single{R}_i^{j+1} + \single{z}_{r_j}$ ($r_j$ is zero) & &
            & $\single{p}_j + \single{P}_i^j \rightarrow \single{p}_j + \single{p}_j + \single{P}_i^{j+1}$ (no zero)
            \\
           
            & $\single{C}_i^{j} +  \single{r}_j \rightarrow \single{R}_i^{j+1} + \single{r}_j$ ($r_j$ is not zero) & & &
            
            \\ \cline{1-2}\cline{4-5}
            
            & $\single{R}_i^{|\RM_i|+1} \rightarrow \single{G} + \config{\PM_i}$ & & 
            & $\single{P}_i^{|\PM_i|+1} \rightarrow \single{G}$
             \\ \cline{1-2}\cline{4-5}
           \multicolumn{2}{l}{\textbf{(a)} Using reactants sequentially $\config{\RM_i} \SeqReacts \config{\PM_i}$} & \multicolumn{1}{c}{}& \multicolumn{2}{l}{\textbf{(b)} Creating products sequentially $\config{\RM_i} \SeqProds \config{\PM_i}$}
        \end{tabular}
    \caption{An overview of two techniques. (a) A procedure to consume the reactants of some rule sequentially in order to test whether a rule can be applied and whether zero species are (or need to be) created. This is abbreviated for rule $i$ as $\config{\RM_i} \SeqReacts \config{\PM_i}$. (b) Creating products sequentially by consuming all reactants and then using counter species to create each product with a different rule until all have been created. This may be needed to decrease the number of rule combinations while handling zero species. Note that the zero species $z_{p_j}$ is consumed when $p_j$ is created. If $p_j$ still exists, then it is used as a catalyst. This process is denoted for some rule $i$ as $\config{\RM_i} \protect\SeqProds \config{\PM_i}$.  } \label{tab:react_break}
        \label{tab:prod_break}
        \vspace{-.2cm}
\end{table}
\subsection{$k$-Void Genesis equivalence with Void Genesis}\label{sub:VG_simulations}

\begin{restatable}{lemma}{kvgvg} \label{lem:kvg-vg}
    The $k$-Void Genesis model can simulate the Void Genesis model. 
\end{restatable}
\para{Construction}
Given a Void Genesis model $\crn{C_{VG}}=((\species, \reactions), z)$, let $\Zero(\lambda)=z$, $\forall \lambda \in \Lambda\setminus\{z\}$. Then the $k$-Void Genesis model $\crn{C_{KVG}}=((\species, \reactions), \Zero)$ is equivalent.


\begin{restatable}{lemma}{vgkvg} \label{lem:vg-kvg}
    The Void Genesis model can simulate the $k$-Void Genesis model. 
\end{restatable}

\para{Construction} This result is straightforward using the methods to sequentially consume reactants as previously defined. We slightly modify them with specifics as shown in Table \ref{tab:vgreact_break}. Whenever we consume a reactant $r_j$, we check if the single $z$ species exists and if it does, create the specific $z_{r_j}$ species if $r_j$ is mapped in $\Zero$.

Given a $k$-Void Genesis model $\crn{C_{KVG}}=((\species, \reactions), \Zero)$, we create a VG CRN $\crn{C_{VG}}=((\species', \reactions'), z)$. We let $\species' = \species \cup \{G\} \cup \{z_{\lambda}: \lambda \in \species\setminus\{z\}\} \cup S_R$ where $S_R = \{R_i^j:1 \leq i \leq |\reactions|, 1 \leq j \leq |\RM_i|+1\}$.
We can simulate having specific $z_j$ species for all $j$ species by keeping the $z$ species count at 0, and checking whether we consumed the last $j$. 

\begin{table}[t]
\vspace{-.2cm}
    \centering
    \begin{tabular}{| @{}l  @{}l |@{} l l l}\cline{1-2}
        & $\single{G} + \config{\RM_i} \rightarrow \single{R}_i^{1} + \config{\RM_i}$ & &
        \multicolumn{2}{c}{$\single{G} + \single{a} + 2\single{b} \rightarrow \single{R}_1^1 + \single{a} + 2\single{b}$}  \\ 
        \cline{1-2}\cline{4-5}
        
        \multirow{4}{*}{$\forall r_j \in \RM_i:$} & $\single{R}_i^j + \single{r}_j \rightarrow \single{C}_i^{j}$ (check for zero)
         & & $\single{R}_1^1 + \single{a} \rightarrow \single{C}_1^1$ & $\single{C}_1^2 + \single{b} \rightarrow \single{R}_1^3 + \single{b}$
        \\ 
        
         & \multirow{2}{*}{$\begin{cases}
            \single{C}_i^j + \single{z} \rightarrow \single{R}_i^{j+1} \text{(unmapped species)}\\ 
            \single{C}_i^j + \single{z} \rightarrow \single{R}_i^{j+1} + \single{z}_{r_j} \text{(zero species)}
        \end{cases}$} & &
        $\single{C}_1^1 + \single{z} \rightarrow \single{R}_1^2 + \single{z}_a$ & $\single{R}_1^3 + \single{b} \rightarrow \single{C}_1^3$ 
        \\
        
        & & &
        $\single{C}_1^1 + \single{a} \rightarrow \single{R}_1^2 + \single{a}$ & $\single{C}_1^3 + \single{z} \rightarrow \single{R}_1^4 + \single{z}_b$\\
        
        & $\single{C}_i^{j} +  \single{r}_j \rightarrow \single{R}_i^{j+1} + \single{r}_j$ (count not zero) & &
        $\single{R}_1^2 + \single{b} \rightarrow \single{C}_1^2$ & $\single{C}_1^3 + \single{b} \rightarrow \single{R}_1^4 + \single{b}$
        \\ \cline{1-2}
         & $\single{R}_i^{|\RM_i|+1} \rightarrow \single{G} + \config{\PM_i}$  & & $\single{C}_1^2 + \single{z} \rightarrow \single{R}_1^3 + \single{z}_b$ & $\single{R}_1^4 \rightarrow \single{G} + 2\single{a} + \single{c}$\\  \cline{1-2}
    \multicolumn{2}{l}{\textbf{(a)} Zero checking reactants $\config{\RM_i} \SeqReacts \config{\PM_i}$} & \multicolumn{1}{c}{}& \multicolumn{2}{l}{\textbf{(b)} Example reaction set}
    \end{tabular}

    \caption{(a) The rules for the simulation of $k$-Void Genesis by Void Genesis. For each applicable rule, we create a set of rules that sequentially consume each reactant $r_j$, and if its count is now zero, it creates the $z_{r_j}$ species (unless it is unmapped, then $z$ is simply consumed). Only one of the zero check rules is added based on the mapping, which is why they are grouped.
    (b) The reactions that simulate the rule $\single{a} + 2\single{b} \rightarrow 2\single{a} + \single{c}$ from the $k$ void genesis model within the void genesis model.
    }\label{tab:vgreact_break}\label{tab:vgexample}
    \vspace{-.2cm}
\end{table}

\begin{restatable}{theorem}{kvgvgthm}\label{thm:vg-kvg}
    The Void Genesis model is equivalent under polynomial simulation to the $k$-Void Genesis model. 
\end{restatable}
\subsection{Inhibitory CRNs}
\begin{restatable}{lemma}{icrnkvg}\label{lem:i > VG}
Inhibitory CRNs can simulate any given $k$-VG CRN under polynomial simulation.
\end{restatable}

\para{Construction}
Given a $k$-VG CRN $\crn{C_{KVG}} = ((\species, \reactions), Z_{\emptyset})$, we construct an Inhibitory CRN $\crn{C_{IC}} = ((\species', \reactions'), \mathcal{I})$. Let $\species' = \species \cup \{ I, e_{\lambda_1}, \ldots, e_{\lambda_{\size{\species_1}}} \}$ where $I$ and $e_{\lambda_1}, \ldots, e_{\lambda_{\size{\species_1}}}$ check if any species of $\species_1$ used in a simulated reaction have a resulting count of zero.
Figure \ref{fig:IC->KVG} shows an example of an Inhibitory CRN simulating a $k$-VG CRN with $\species=\{a,b\}$ and a reaction that consumes the species $a$.

\begin{table}[t]
    \centering
    \iCRNiffKVG
    \caption{(a) Reactions for an Inhibitory CRN to simulate any given k-VG CRN. (b) Reactions for a k-VG CRN to simulate any given Inhibitory CRN.}\label{tab:VG <> iCRN}\label{tab:iCRN > VG}\label{tab:VG > i}
    \vspace{-.2cm}
\end{table}

\begin{figure}[t]
\vspace{-.2cm}
    \begin{subfigure}{.56\textwidth}
        \centering
        \includegraphics[width = \textwidth]{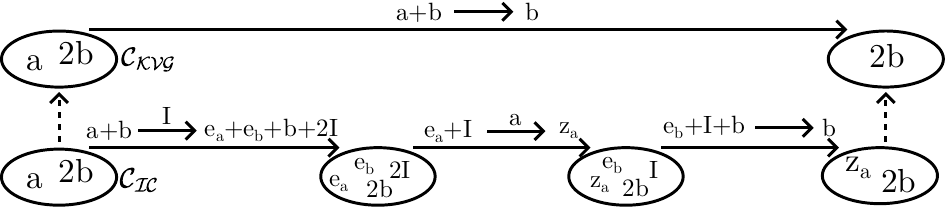}
        \subcaption{$\crn{C_{IC}}\rightarrow\crn{C_{KVG}}$}\label{fig:IC->KVG}
        \label{fig:icrn_vg_construction}
    \end{subfigure}
    \begin{subfigure}{.43\textwidth}
        \centering
        \includegraphics[width = \textwidth]{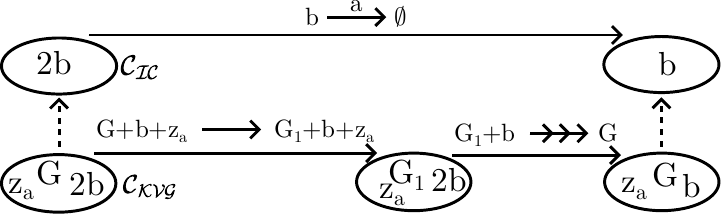}
        \subcaption{$\crn{C_{KVG}}\rightarrow\crn{C_{IC}}$}\label{fig:KVG->IC}
        \label{fig:vg_icrn_construction}
    \end{subfigure}
    \vspace{-.4cm}
    \caption{ (a) A rule application in the \emph{simulated} $k$-void genesis CRN $\crn{C_{KVG}}$ and the equivalent rule application sequence in the \emph{simulating} inhibitory CRN $\crn{C_{IC}}$. The dashed arrows represent mapping a configuration of $\crn{C_{IC}}$ to a configuration of $\crn{C_{KVG}}$. (b) The other direction. }\label{fig:KVG-IC}
    \vspace{-.2cm}
\end{figure}

\begin{restatable}{lemma}{kvgicrn} \label{lem:VG > i}
k-VG CRNs can simulate any Inhibitory CRN under polynomial simulation.
\end{restatable}

\para{Construction}
Given an Inhibitory CRN $\crn{C_{IC}} = ((\species, \reactions), \mathcal{I})$, we construct a $k$-VG CRN $\crn{C_{KVG}} = ((\species', \reactions'), Z_{\emptyset})$. Let $\species' = \species \cup \{G, G_1, \ldots, G_{\size{\reactions}}, P_i^1, \ldots, P_{\size{\reactions}}^{\size{\PM_{\size{\reactions}}}+1}, z_{\lambda_1}, \ldots, z_{\lambda_{\size{\species}}}\}$. Let $Z_i$ represent the set of $z$ species corresponding to inhibitors of a reaction $\reaction_i$. $G$ is consumed to produce $G_i$ when selecting reaction $\reaction_i \in \reactions$, as long as no inhibitor species of $\reaction_i$ are presented (indicted by the complete presence of $Z_i$). $G_i$ then applies $\reaction_i$ and restores $G$ back.
Figure $\ref{fig:KVG->IC}$ shows an example of a $k$-VG CRN simulating an Inhibitory CRN with $\species=\{a,b\}$ and a reaction that consumes $b$ if $a$ is absent.

\begin{restatable}{theorem}{vgicrnthm} \label{thm:vgicrnthm}
    The Void Genesis model is equivalent under polynomial simulation to the Inhibitory model.    
\end{restatable}
\subsection{Coarse-Rate CRNs}

\begin{restatable}{lemma}{crkvg}\label{lem:cr > VG}
Coarse-Rate CRNs can simulate any given k-VG CRN under polynomial simulation.
\end{restatable}
\para{Construction}
Given a $k$-VG CRN $\crn{C_{KVG}} = ((\species,\reactions),Z_\emptyset)$, we construct a Course-Rate CRN $\crn{C_{CR}} = ((\species',\reactions'),rank)$. Let $\species' = \species \cup \{ x, y, G, e_{\lambda_1}, \ldots, e_{\lambda_{\size{\species_1}}} \}$ where $G$ is used to select a reaction and species $x$ and $y$ check if any non-zero species in $\species$ have a count of zero. Finally, $e_{\lambda_1}, \ldots, e_{\lambda_{\size{\species_1}}}$ are used to check the counts of the $\lambda_i \in \species_1$ that are relevant to the simulated reaction, where $\species_1$ is the set of non-zero species in $\species$.
Figure \ref{fig:cr_vg_construction} shows an example of a Coarse-Rate CRN simulating a $k$-VG CRN with $\species=\{a,b\}$ and a reaction that consumes $a$.

\begin{table}[t]
\vspace{-.2cm}
    \centering
    \CRiffVG
    \caption{(a) Reactions for a Coarse-Rate CRN to simulate any given $k$-VG CRN. (b) Reactions for a $k$-VG CRN to simulate any given Coarse-Rate CRN.}\label{tab:cr > VG}\label{tab:VG > cr}
    \vspace{-.2cm}
\end{table}

\begin{figure}[t]
\vspace{-.2cm}
    \centering
    \includegraphics[width=.9\textwidth]{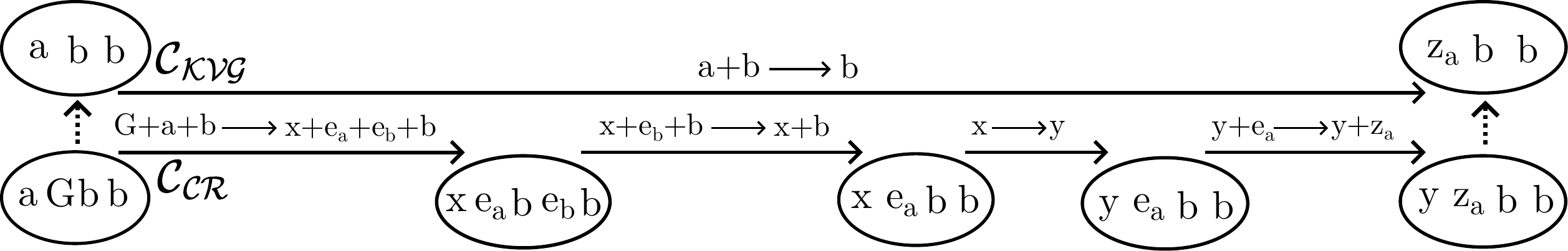}
    \vspace{-.2cm}
    \caption{A rule application in the \emph{simulated} $k$-void genesis CRN $\crn{C_{KVG}}$ and the equivalent rule application sequence in the \emph{simulating} Coarse-Rate CRN $\crn{C_{CR}}$. }
    \label{fig:cr_vg_construction}
    \vspace{-.2cm}
\end{figure}

\begin{restatable}{lemma}{kvgcr} \label{lem:VG > cr}
k-VG CRNs can simulate any given Coarse-Rate CRN under polynomial simulation. 
\end{restatable}
\para{Construction}
Given a Coarse-Rate CRN $\crn{C_{CR}} = ((\species, \reactions), rank)$, we construct a $k$-VG CRN $\crn{C_{KVG}} = ((\species', \reactions'), Z_{\emptyset}')$.
Let $\species' = \species
\cup \{ G, s, t\}
\cup \{g_i, z_{g_i},
r_1,\ldots, r_{\size{\reactions^2}+1},
z_{\lambda_1}, \ldots, z_{\lambda_{\size{\Lambda}}},
P_1^1, \ldots$, $ P_{\size{\reactions^1}}^{\size{\PM_{\size{\reactions^1}}} + 1} :
i\in\{1,\ldots,\size{\reactions^2}\}\}$. $G$ and a random $g_i$ species is consumed to select a random fast reaction $\reaction_i^2\in\reactions^2$, and any missing $g_i$ species is restored with $r_i$ species. Alternatively, if no $g_i$ species is present, then $G$ and all $z_{g_i}$ species are consumed to produce $s$. $s$ selects a random slow reaction $\reaction_i^1\in\reactions^1$ and produces $t$, which restores $G$ and all $g_i$ species. Figure \ref{fig:vg_cr_construction} shows an example of a $k$-VG CRN simulating a Course-Rate CRN with $\Lambda=\{a,b\}$ and a fast reaction that deletes $a$ and a slow reaction that deletes $b$.

\begin{figure}[t]
\vspace{-.2cm}
    \centering
    \includegraphics[width=.8\textwidth]{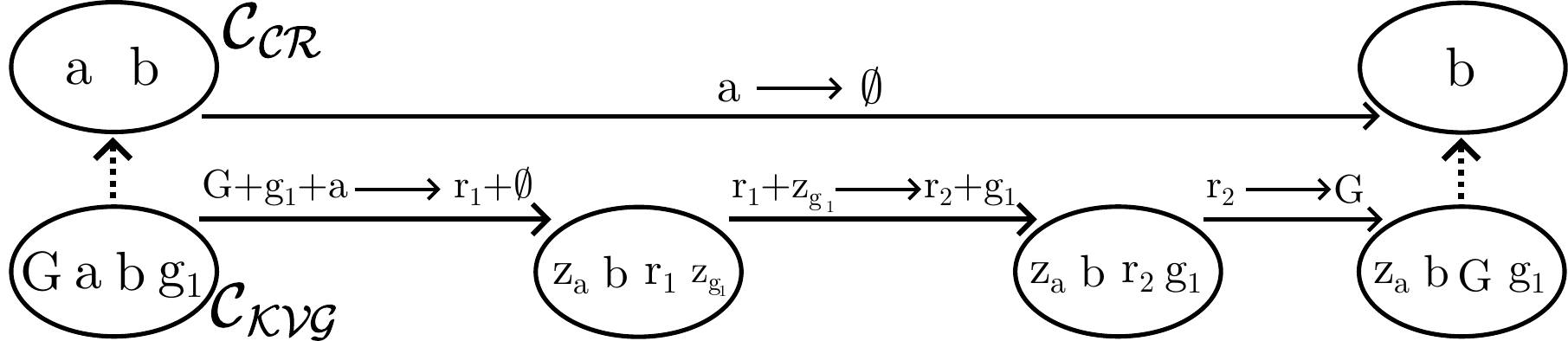}
    \vspace{-.2cm}
    \caption{A rule application in the \emph{simulated} Coarse-Rate CRN $\crn{C_{CR}}$ and the equivalent rule application sequence in the \emph{simulating} $k$-void genesis CRN  $\crn{C_{KVG}}$.}
    \label{fig:vg_cr_construction}
    \vspace{-.2cm}
\end{figure}

\begin{restatable}{theorem}{vgcrthm}
    The Void Genesis model is equivalent under polynomial simulation to the Course-Rate model.    
\end{restatable}
\subsection{Step-Cycle CRNs}
\begin{restatable}{lemma}{sckvg}\label{lem:cycle > VG}
Step-Cycle CRNs can simulate any given k-VG CRN under polynomial simulation.
\end{restatable}

\para{Construction} 
Given a $k$-VG CRN $\crn{C_{KVG}} = ((\species, \reactions), \Zero)$, we construct a Step-Cycle CRN $\crn{C_{SC}} = ((\species', \reactions'), \config{S}_0)$. We let $\species' = \species \cup \{ G, y, x, w, e_{\lambda_1}, \ldots, e_{\lambda_{\size{\species_1}}}, z_{\lambda_1}, \ldots, z_{\lambda_{|\species_2|}} \}$ and $\config{S}_0 = \single{y}$. The $G$ species is used to select a reaction, and $y$ is used to check if any non-zero species in $\species$ have a count of zero. Finally, $e_{\lambda_1}, \ldots, e_{\lambda_{\size{\species_1}}}$ are used to check the counts of the $\lambda_i \in \species_1$ that are relevant to the simulated reaction, where $\species_1$ is the set of non-zero species in $\species$. 
Figure \ref{fig:cycle_vg_construction} shows an example of a Step-Cycle CRN simulating a $k$-VG CRN with $\species=\{a,b\}$ and a reaction that consumes the species $a$.

\begin{table}[t]
\vspace{-.2cm}
    \centering
    \SCiffKVG
    \caption{(a) Reactions for a Step-Cycle CRN to simulate any given k-VG CRN. (b) Reactions for a k-VG CRN to simulate any given Step-Cycle CRN.}\label{tab:cycle > VG}\label{tab:VG > cycle}
    \vspace{-.2cm}
\end{table}

\begin{table}[t]
\vspace{-.2cm}
    \centering
    \SCReactants
    \caption{(a) Checking reactants sequentially $\config{\RM_i} \SeqReacts \config{\PM_i}$ (b) Undoing reaction selection if not enough reactants exist for rule $i$. }\label{tab:check_reactants}\label{tab:undo_reactants}
    \vspace{-.2cm}
\end{table}

\begin{figure}[t]
\vspace{-.2cm}
    \centering
    \includegraphics[width=.9\textwidth]{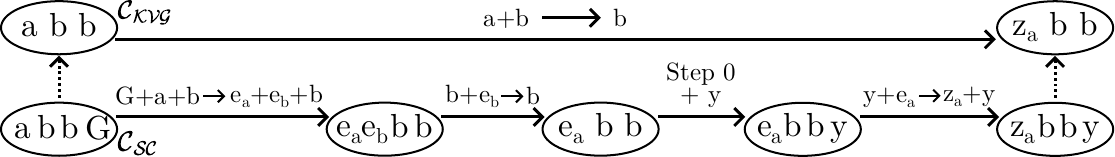}
    \vspace{-.2cm}
    \caption{A rule application in the \emph{simulated} $k$-VG CRN $\crn{C_{KVG}}$ and the equivalent rule application sequence in the \emph{simulating} Step-Cycle CRN $\crn{C_{SC}}$. }
    \label{fig:cycle_vg_construction}
    \vspace{-.2cm}
\end{figure}

\begin{restatable}{lemma}{kvgsc} \label{lem:VG > cycle}
k-VG CRNs can simulate any given Step-Cycle CRN under polynomial simulation.
\end{restatable}

\para{Construction}
Given a Step-Cycle CRN $\crn{C_{SC}} = ((\species, \reactions), (\config{S}_0, \config{S}_1, \ldots, \config{S}_{k-1})$, we construct a $k$-VG CRN $\crn{C_{KVG}} = ((\species', \reactions'), Z_{\emptyset})$. Let $\species' = \species \cup \{ G, \lambda_1', \ldots, \lambda_{|\species|}', z_{\lambda_1}, \ldots, z_{\lambda_{|\species|}}, r_1, \ldots, r_{\size{\reactions} + 1}$, $s, s_0, \ldots, s_{k-1}, t \} \cup \{ g_i, z_{g_i}, R_i^j, R_i^{j^-}, P_i^l : 1 \leq i \leq |\reactions|, 1 \leq j \leq |\RM_i|+1, 1 \leq l \leq |\PM_i|+1 \}$. $G$ and $g_i$ are consumed to select a random reaction. The reactants are checked sequentially, converting each reactant into $\lambda_i'$. Then the $r_i$ species reintroduce any deleted $g_i$ into the system. Species $s_i$ represent step $i$, with species $s$ and $t$ used to transition between steps. 
Figure \ref{fig:vg_cycle_construction} shows an example of a $k$-VG CRN simulating a Step-Cycle CRN with $\species = \{ a, b \}$ and reaction $a + b \rightarrow b$. 

\begin{figure}[t]
\vspace{-.2cm}
    \centering
    \includegraphics[width=.9\textwidth]{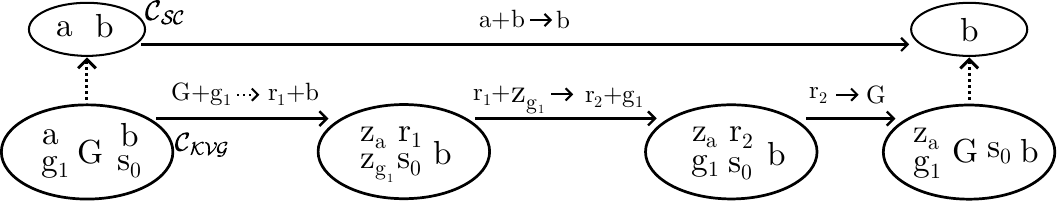}
    \vspace{-.2cm}
    \caption{A rule application in the \emph{simulated} Step-Cycle CRN $\crn{C_{SC}}$ and the equivalent rule application sequence in the \emph{simulating} $k$-VG CRN $\crn{C_{KVG}}$.}
    \label{fig:vg_cycle_construction}
    \vspace{-.2cm}
\end{figure}

\begin{restatable}{theorem}{vgscthm}
    The Void Genesis model is equivalent under polynomial simulation to the Step-Cycle model.     
\end{restatable}

\para{Extension to deletion-only rules} Given the recent result in \cite{Luchsinger:2025:UCNC}, these results extend to give the following corollary.

\begin{restatable}{corollary}{31sccor}
    Even when restricted to at most $(3,1)$ void rules, the Step-Cycle model is equivalent under polynomial simulation to the Void Genesis model.
\end{restatable}

\subsection{Unique-Instruction Parallel Model}

\begin{restatable}{lemma}{uikvg} \label{thm:UI > VG}
UI parallel CRNs can simulate any given VG CRN.
\end{restatable}

\para{Construction} Given a Void-Genesis CRN $\crn{C_{VG}} = ((\species, \reactions), z)$, we construct the UI parallel CRN $\crn{C_{UI}} = (\species', \reactions')$. We create $\species'=\species\cup\{G,E,t^2\}\cup\{e_{\lambda_j},t_j^1,r_j^1,r_j^2,G_i:i\in\{1,\dots,|\reactions|\}, j\in\{1,\dots,|\{\RM_i\}|\}\}$. The set of species $\species'$ each has a role. The global species G, selects a rule $\reaction_1' \in \reactions'$ non-deterministically, and produces a $G_i$ to disallow the reaction 1 running again, a set of species $e_{\lambda}$, and produces the products. The $e_{\lambda}$ species then creates a timer species $t^i_j$, and a checker species $r^1_j$, which is used to check if a species exists in the system. In reaction 3 the rules run in parallel, so the timer species goes down, and at the same time, we check if any checker species have incremented. Reaction 4, reintroduces the consumed species or creates the zero species $z$, based on the checker species. Both reactions produce the $E$, which is later used to reintroduce the $G$ species. Figure $\ref{fig:ui_vg_construction}$ shows an example of a UI CRN simulating a simple VG CRN. 

\begin{figure}[t]
\vspace{-.2cm}
    \centering
    \includegraphics[width=.9\textwidth]{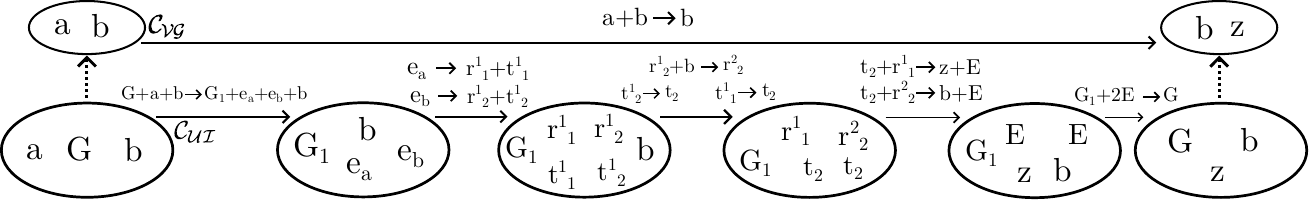}
    \vspace{-.2cm}
    \caption{A rule application in the \emph{simulated} void genesis CRN $\crn{C_{VG}}$ and the equivalent rule application sequence in the \emph{simulating} Unique Instruction CRN $\crn{C_{UI}}$. }
    \label{fig:ui_vg_construction}
    \vspace{-.4cm}
\end{figure}

\begin{table}[t!]
    \vspace{-.2cm}
    \centering
    \UIiffVG
    \caption{(a) Reactions for a UI parallel CRN to simulate any given VG CRN. Here, $\config{e_{\{\RM_i\}}}$ is an $e_{\lambda_j}$ species created for each of the $j$ different reactants in $\RM_i$ (only one $e$ is created if multiple copies of that species are used). 
    (b) Reactions for a kVG CRN to simulate any given UI parallel CRN.}\label{tab:UI > VG}\label{tab:kVG > UI}
    \vspace{-.2cm}
\end{table}

\begin{restatable}{lemma}{kvgui} \label{thm:VG > UI}
$k$-VG CRNs can simulate any given UI Parallel CRN.
\end{restatable}

\para{Construction} Given a UI parallel CRN $\crn{C_{UI}} = (\species, \reactions)$, we construct a $k$-Void-Genesis CRN $\crn{C_{KVG}} = ((\species', \reactions'), \Zero)$ as described. 
Let $\species'=\species\cup\{F\}\cup\{G_i,N_i,I_i,F_i,R_i^j:i\in\{1,\dots,|\reactions|\}, j\in\{1,\dots,|\RM_i|\}\}\cup\{z_{r_j}:r_j\in \RM_i s.t. 1 \leq i \leq |\reactions|\}$. A $G_i$ species exists for each $\reaction_i\in\reactions$, which sequentially consumes the reactants. If successful, it is committed to running $I_i$. If some reactant is missing, it returns the previously consumed reactions and notes that it can not run $N_i$. Both commitments create an $X_i$, and when all $\size{\reactions}|\cdot X$ exist, a maximal set has been chosen and an $F$ is created. This allows $N_i$ or $I_i$ to turn into $F_i$, and output the products if the rule was an $I$. Once all rules are converted to an $F_i$, they all combine to reset the rule selections by creating all the $G_i$'s again. Figure $\ref{fig:kvg_ui_construction}$ shows an example of a $k$-VG CRN simulating a UI CRN with $\species = \{ a, b \}$ and reaction $a + b \rightarrow b$. 

\begin{figure}[t!]
    \centering
    \includegraphics[width=.9\textwidth]{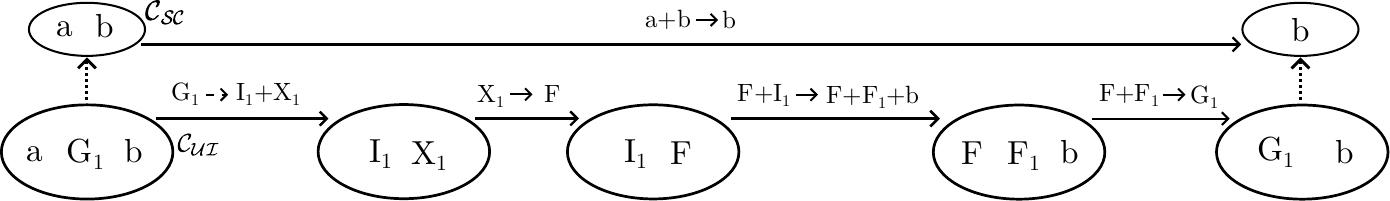}
    \vspace{-.2cm}
    \caption{A rule application in the \emph{simulated} Unique Instruction CRN $\crn{C_{UI}}$ and the equivalent rule application sequence in the \emph{simulating} $k$-void genesis CRN $\crn{C_{KVG}}$.}
    \label{fig:kvg_ui_construction}
    \vspace{-.2cm}
\end{figure}

\begin{restatable}{theorem}{vguithm}
\label{thm:VG=UI}
    The Void Genesis model is equivalent under polynomial simulation to the Unique-Instruction parallel model.
\end{restatable}
\subsection{Register Machine}\label{sub:undecide}

In this section, we show that Void Genesis CRNs are Turing Universal by constructing a simulation of a standard register machine.


\begin{theorem}\label{thm:VG_undecide}
$k$-VG CRNs are Turing Universal.
\end{theorem}
\begin{proof}
Given a Register Machine (RM) with $s_1, \ldots, s_n$ states--each with either a $inc(r_l, s_j)$ or $dec(r_l, s_j, s_k)$ instruction--and $r_1, \ldots, r_m$ registers, we construct the $k$-VG CRN as follows. Each register and state will be represented by a species, and the instructions will be encoded in the rules. The initial configurations should have one copy of $s_1$ (assume $s_1$ is the RM's starting state) and the register species' counts should be equal to the registers they represent.
    
\begin{table}[t]
\centering
        
        
    \begin{tabular}{| l l | l l | }\hline
        \textbf{Instruction} & \textbf{Relevant Rules} & \textbf{Instruction} & \textbf{Relevant Rules} \\ \hline
        \multirow{3}{*}{$s_j: inc(r_i, s_k)$} 
        & $\single{s}_j + \single{r}_i \rightarrow \single{s}_k + \single{r}_i + \single{r}_i$ 
        & \multirow{3}{*}{$s_j: dec(r_i, s_k, s_l)$} 
        & $\single{s}_j + \single{r}_i \rightarrow \single{s}_k$ \\ 
        & $\single{s}_j + \single{z}_{r_i} \rightarrow \single{s}_k + \single{r}_i$  
        & & $\single{s}_j + \single{z}_{r_i} \rightarrow \single{s}_l + \single{z}_{r_i}$ \\
        & & & $\Zero(r_i) = z_{r_i}$ \\  \hline
    \end{tabular}
    
    \caption{\textit{Rules for a Void Genesis CRN to simulate a given Register Machine.}}\label{tab:RM_sim}
\end{table}

Table \ref{tab:RM_sim} shows what rules should be made for each state of a given RM. Because $k$-VG CRNs can simulate any given RM, $k$-VG CRNs are Turing Universal.
\end{proof}

Although it can be inferred from this result and the simulation equivalence results that all the models studied here are Turing Universal, we do not give formal proofs of this due to space.

\section{Conclusion}\label{sec:conc}

In this paper, we demonstrate equivalence through polynomial simulation between 5 natural extensions to the CRN model. We centralize these simulations around the Void Genesis CRN model, as this model's ability to detect zero is one of the simplest augmentations to a regular CRN. We then show that Void Genesis CRNs are Turing Universal, implying that Step-Cycle CRNs, Inhibitory CRNs, Parallel CRNs, and Coarse-Rate CRNs are also Turing Universal. While this work is complete in proving equivalence between these models, there are still several interesting open problems to consider  (some of which are shown in Figure \ref{fig:vg_hub}):

\begin{itemize}
    \item We aim to explore constrained versions of our simulation definition that recover existing notions as special cases. Does restricting the configuration map to be consistent with an underlying species-species mapping immediately results in weak bisimulation? If that underlying map is a total bijective function, does that yield strong bisimulation?
    \item For iCRNs, a rule is inhibited by the existence of one or more species. Our definition effectively uses a logical \texttt{OR} (inhibition is only false when all inhibitor counts are zero). A natural extension is to consider inhibition functions using other logic (e.g., AND - a reaction is inhibited only when all of its inhibitors are present).
    \item Coarse-Rate CRNs are limited to 2 ranks for reactions. A natural generalization of this model is to allow for $k$ different ranks ($k$-rate CRNs). What is the relationship between Coarse-Rate CRNs and $k$-rate CRNs?
    \item Even when limited to only void reactions (rules where no species are created), step CRNs are able to compute threshold circuits \cite{anderson2024steps,anderson2024computing}. Does this suggest that Step-Cycle CRNs, even when limited to only void reactions (void Step-Cycle), are still Turing Universal?  
    \item Are there more efficient ways to simulate these augmented CRNs using VG CRNs?
    \item What is the complexity of reachability in restricted instances of each model, such as \cite{Fu:2025:SAND,Fu:2025:DNA}? As mentioned, we know the complexity with step-cycles \cite{Luchsinger:2025:UCNC}, but deletion-only rules have not been explored in detail in the other models.
\end{itemize}

\newpage
\bibliographystyle{plainurl}
\bibliography{crns}
\newpage
\appendix
\section{Formal Void Genesis}\label{constructions}

\begin{table}[H]
    \centering
        \iCRNsimKVG
        \caption*{Table ~\ref{tab:iCRN > VG} (restated): iCRN simulating $k$-VG}
\end{table}
\icrnkvg*
\begin{proof} For any given $k$-VG CRN $\crn{C_{KVG}}$ and a configuration $\config{C}$ of $\crn{C_{KVG}}$, we construct an Inhibitory CRN $\crn{C_{IC}}$ and configuration $\config{C}'$ that simulates $\crn{C_{KVG}}$ over $\config{C}$ as follows. 

\para{Construction}
Given a $k$-VG CRN $\crn{C_{KVG}} = ((\species, \reactions), Z_{\emptyset})$, we construct an Inhibitory CRN $\crn{C_{IC}} = ((\species', \reactions'), \mathcal{I}')$ as follows.

We first build the species set $\species'$. First, each species from $\species$ is added with no modification. We then create the set of species $e_{\lambda_1}, \ldots, e_{\lambda_{\size{\species_1}}}$ to check if the count of a reactant consumed in a simulated reaction $\reaction_i$ has reached zero. We also make the species $I$ to ensure that each reactant was checked before simulating another reaction.
The configuration of $\crn{C_{IC}}$ $\config{C}'$ is simply the configuration of $\crn{C_{KVG}}$ $\config{C}$ with no changes.

We now construct the reaction set of $\crn{C_{IC}}$ $\reactions'$. First, Reaction $1$ simulates a reaction $\reaction_i \in \reactions$. In addition to consuming $\config{\RM}_i$ and creating $\config{\PM}_i$, Reaction $1$ produces $\config{e}_{\{\RM_i\}}$ and $\size{\{\RM_i\}}$ copies of $\single{I}$ for zero-checking each reactant of $\reaction_i$. The zero-checking itself is performed by Reactions $2$ or $3$. If a copy of a reactant $\lambda_i\in\RM_i$ still remains present, then Reaction $2$ preserves that copy. Otherwise, Reaction $3$ will be fired and produce one copy of $\single{z}_{\lambda_i}$. Note that Reaction $1$ is inhibited by the species $I$, which can only be consumed by Reactions $2$ and $3$, so it cannot be fired again until the zero-checking for all reactants $\RM_i$ is completed.

We now define the configuration mapping of $\crn{C_{IC}}$ as follows. If a copy of the species $I$ \emph{does not exist} in a configuration of $\crn{C_{IC}}$ $\config{C'}$, then $\config{C}'$ maps to a configuration of $\crn{C_{KVG}}$ $\config{C}$ that holds the same count of species $\lambda_i,\forall \lambda_i\in\species$. The intuition of this mapping is that $I$ only exists in a configuration following the simulation of some reaction $\reaction_i\in\reactions$ with Reaction $1$. Then Reactions $2$ and $3$ will be used to zero-check each reactant used in $\reaction_i$, consuming the copies of $I$ in the process. Only once each reactant has been checked (and thus all $I$ copies removed) can Reaction $1$ be fired again to simulate another reaction from $\reactions$.

\[\cmap(\config{C}') =
\begin{cases}
    \sum_{\lambda' \in \species} \config{C}'[\lambda']\cdot\single{\lambda}' & \text{if}\, I \not\in \Set{\config{C}'}\\
    \undefConfig & \text{otherwise.}
\end{cases}\]

\para{$\crn{C_{KVG}}$ follows $\crn{C_{IC}}$} Given configurations $\config{A'}$ and $\config{B'}$ in $\crn{C_{IC}}$ such that $\cmap(\config{A}')$ and $\cmap(\config{B}')$ are defined, we show that $\crn{C_{KVG}}$ follows $\crn{C_{IC}}$. We first show that if $\config{A'} \macrotrans_{\crn{C_{IC}}} \config{B'}$ and $\cmap(\config{A'}) \not= \cmap(\config{B'})$ then $M(\config{A'})\rightarrow_{\crn{C_{IC}}} M(\config{B'})$. Since $\cmap(\config{A}')$ and $\cmap(\config{B}')$ are defined, the species $I$ does not exist in $\config{A'}$ and $\config{B'}$. Because $\cmap(\config{A'}) \not= \cmap(\config{B'})$, there must sequence a sequence of rule applications to transform $\config{A}'$ to $\config{B}'$. Assume a reaction of the form $\config{\RM}_i \xrightarrow{I} \config{e}_{\Set{\RM_i}} + \config{\PM}_i + \size{\Set{\RM_i}} \cdot \single{I}$ can be applied in $\config{A}'$. Let the resulting configuration be $\config{X}'_1=\config{A}'-\config{\RM_i}+\config{e}_{\Set{\RM_i}} + \config{\PM}_i + \size{\Set{\RM_i}} \cdot \single{I}$. By the presence of $\size{\Set{\RM_i}}$ copies of $I$, Reaction $1$ will be inhibited. Then only a sequence of application of Reactions $2$ or $3$ will occur in $\config{X}'_1$, consuming all copies of $\single{e}_{\lambda_i}$ and $\single{I}$ and producing some copies of $\single{z}_{\lambda_i}$. Let the set of $\single{z}_{\lambda_i}$ copies added be $\config{Z}$. Then the final resulting configuration is $\config{B}'=\config{A}'-\config{\RM_i}+\config{\PM_i}+\config{Z}$. Denote the first applied rule $\config{\RM}_i \xrightarrow{I} \config{e}_{\Set{\RM_i}} + \config{\PM}_i + \size{\Set{\RM_i}} \cdot \single{I}$ as $\reaction_i'$. By the construction of $\crn{C_{IC}}$, $\reaction'_i$ must be created from some reaction $\reaction_i\in\reactions$ where $\RM_i\rightarrow\PM_i$.
By the mapping function the configuration of $\crn{C_{KVG}}$ $M(\config{A'})$ contains the same count for all $\lambda_i\in\species$. Then if $\reaction'_i$ is applicable in $\config{A}'$, $\reaction_i$ must also be applicable in $M(\config{A'})$.  
Applying $\reaction_i$ in $M(\config{A'})$ then results in the configuration $M(\config{B'}) = M(\config{B'_t}) + \config{Z}$, where $M(\config{B'_t})=M(\config{A'}) - \RM_i + \PM_i$. Clearly, $M(\config{B'})=\config{B}'$.
Thus, it follows that $\cmap(\config{A'}) \rightarrow_{\crn{C_{KVG}}} \cmap(\config{B'})$, if $\config{A'} \macrotrans_{\crn{C_{IC}}} \config{B'}$ and $\cmap(\config{A'}) \not= \cmap(\config{B'})$.

\para{$\crn{C_{IC}}$ models $\crn{C_{KVG}}$} Assume there exists two configurations of $\crn{C_{KVG}}$ $\config{A}$ and $\config{B}$ such that $\config{A}\
\rightarrow_{C_{KVG}} \config{B}$, where for every configuration in $\crn{C_{IC}}$ for which a mapping is defined, there exist a unique mapping configuration $M(\config{A'})$ in $\crn{C_{KVG}}$. Therefore, for any configuration $\config{A}$ we only have one configuration $\config{A'} \in [\![\config{A}\,]\!]$ which contains the same count of species $\lambda_i \in \species$ as in $\config{A}$. Similarly the configuration $\config{B}$ will have only one unique configuration mapping $\config{B'} \in [\![\config{B}\,]\!]$. Both configurations $\config{A'}$ and $\config{B'}$ can not contain species $I$ because their mapping is defined. By $\textit{$\crn{C_{KVG}}$ follows $\crn{C_{IC}}$}$, there must exist an applicable reaction $\reaction_i \in \reactions$ such that $\config{B} = \config{B_t} + \config{Z}$, in which $\config{B_t}=\config{A}-\reactants+\products$ and $\config{Z} = \sum_{\lambda \in C} \single{z}_\lambda$, where $C =\{\lambda \in \species | \config{A}[\lambda] \not= 0$ and $\config{B_t}[\lambda]=0\}$ for which the all reactants of the reaction are present. If such a reaction is applicable and the system $\crn{C_{IC}}$ is in configuration $\config{A'}$, there will exist a sequence of applicable reactions such that $\config{B} = \config{B_t} + \config{Z}$ where $\config{B_t}=\sum_{\lambda_i\in\species}\config{A'}[\lambda_i]-\reactants+\products$. The resulting configuration in $\crn{C_{IC}}$ after the application of $\reaction_i$ will contain the same count of species $\lambda_i\in\species$ as in $\config{B}$. As defined by the configuration mapping, such a configuration will be $\config{B'}$. Therefore for any two configurations $\config{A}$ and $\config{B}$ such that $\config{A}\rightarrow_{\crn{C_{IC}}}\config{B}$ and $\config{A'} \in [\![\config{A}\,]\!]$ and $\config{B'} \in [\![\config{B}\,]\!]$, $\config{A'}\macrotrans_{\crn{C_{IC}}}\config{B'}$.
   
Because $\crn{C_{KVG}}$ follows $\crn{C_{IC}}$, and $\crn{C_{IC}}$ models $\crn{C_{KVG}}$, we can state that $\crn{C_{IC}}$ simulates $\crn{C_{KVG}}$.

\para{Polynomial Simulation} We now show that the above simulation is polynomial efficient as follows.

\begin{enumerate}
    \item \textbf{polynomial species and rules:} As defined in the construction, along with the original species in $\species$,
    species $I$ and $e_{\lambda_1}, \ldots, e_{\lambda_{\size{\species_1}}}$ are created for $\species'$
    , resulting in $\size{\species'}=2\cdot\size{\species}+1= \mathcal{O}(\size{\species})$ unique species in the construction. Regarding reactions, we create one reaction for each reaction in $\reactions$ and two reactions for each species in $\species_1$.
    We then have $\size{\reactions}+2\cdot \size{\species}$ reactions, thus resulting in $\size{\reactions'} = \mathcal{O}(\size{\reactions} + \size{\species})$ total rules.
    \item \textbf{polynomial rule size:} Given a reaction $\reaction_i\in\reactions$, Reaction $1$ is created, which consumes $\size{\{\RM_i\}}$ reactants and produces $\size{\{\PM_i\}}+2\cdot\size{\{\RM_i\}}$ products.
    \item \textbf{polynomial transition sequences:} To perform a macro transition $M(\config{A'}) \macrotrans_{\crn{C_{KVG}}} M(\config{B'})$, Reaction $1$ is first applied once. Then one of Reactions $2$ or $3$ will be executed for all $\size{\PM_i}$ products, resulting in $\size{\species}$ additional rule applications. Thus, for a sequence of $n$ transitions in $\crn{C_{KVG}}$, the resulting simulating sequence is of length $\mathcal{O}(\size{\species}\cdot n)$.
    \item \textbf{polynomial volume:} If the configuration $\config{C'}$ maps to a configuration $\config{C}$ in $\crn{C_{KVG}}$ then $\config{C'}[\lambda_i] = \config{C}[\lambda_i]$ and the volume of system in $\config{C'}$ is polynomial in the volume of $\crn{C_{KVG}}$. For any other configuration $\config{C}''$ with an undefined mapping, the volume is bounded by the volume of configuration $\config{A'}$ and $\config{B'}$ such that $\config{A'}\macrotrans_{\crn{C_{IC}}}\config{B'}$ through configurations $\config{C}''$, i.e. is polynomial in the simulated system. This is because Reaction $1$ introduces the products $\config{\PM}_i$ and $\size{\RM_i}$ copies of $I$ into $\config{C}''$, and Reactions 2 and 3 will decrease the volume of $\config{C}''$ by a constant amount.
\end{enumerate}

The above simulation only utilizes 1) polynomial species and rules, 2) polynomial rule size, 3) polynomial transition sequences, and 4) polynomial volume. Therefore, $\crn{C_{IC}}$ simulates $\crn{C_{KV}}$ under polynomial simulation.
\end{proof}

\kvgicrn*
\begin{proof}

For any given Inhibitory CRN $\crn{C_{IC}}$ and a configuration of $\crn{C_{IC}}$ $\config{C}$, we show that we can construct a $k$-VG CRN $\crn{C_{KVG}}$ and a configuration of $\crn{C_{KVG}}$ $\config{C}'$ that simulates $\crn{C_{IC}}$ over $\config{C}$.

\para{Construction}
Given an Inhibitory CRN $\crn{C_{IC}} = ((\species, \reactions), \mathcal{I})$, we construct a $k$-VG CRN $\crn{C_{KVG}} = ((\species', \reactions'), Z_{\emptyset})$ as described below.

We use a global leader species to select a reaction and then generate products sequentially as explained in Section \ref{subsec: techniques}. The initial configuration of $\crn{C_{KVG}}$ starts with the global species $G$, which is consumed to produce $G_i$ when selecting reaction $\reaction_i \in \reactions$. The species $P_i^1 \ldots P_i^{\size{\PM_i}+1}$ are generated as described in Table \ref{tab:prod_break} when generating products $p_1, \ldots, p_{\size{\PM_i}}$. Therefore, the species set $\species'$ will contain the original species set $\species$ as well as a single leader species $G$, reaction-specific leaders $G_1 \ldots G_{\size{\reactions}}$, and the species $P_i^1, \ldots, P_{\size{\reactions}}^{\size{\PM_{\size{\reactions}}}+1}$. Along with the leader species, the $\crn{C_{KVG}}$ will also contain one zero species for every species in $\species$. The species set $\species'$ for $\crn{C_{KVG}}$ can then be defined as $\species' = \species \cup \{G, G_1, \ldots, G_{\size{\reactions}}, P_i^1, \ldots, P_{\size{\reactions}}^{\size{\PM_{\size{\reactions}}}+1}, z_{\lambda_1}, \ldots, z_{\lambda_{\size{\species}}}\}$.

\begin{table}[H]
    \centering
    \KVGsimiCRN
    \caption*{Table~\ref{tab:VG > i} (restated): $k$-VG simulating iCRN}
\end{table}

For any reaction $\reaction_i \in \reactions$, let $Z_i$ represent the set of zero species $z_{\lambda_j}, \forall \lambda_j \in \mathcal{I}(\reaction_i)$. $\crn{C_{KVG}}$ represents each reaction $\reaction_i = (\config{\RM_i}, \config{\PM_i})$ in $\crn{C_{IC}}$ by the two reactions as given in Table \ref{tab:VG > i}. The first reaction uses the reactants $\config{\RM_i}$ and the set of zero species $Z_i$ as a catalyst. This reaction is applicable if both conditions for a reaction $\reaction_i$ are true: the reactants $\RM_i$ are present; and the inhibitors are absent, that is, the zero species in $Z_i$ are present. If the first reaction is applied for any $\reaction_i$, then the species $G$ is consumed and a new species $G_i$ is produced making the second reaction applicable. The second reaction in Table \ref{tab:VG > i} represents a sequence of reactions provided in Table \ref{tab:prod_break}. This consumes $G_i$ and generates the products of $\reaction_i$ while consuming any existing zero species corresponding to the products. Once all products have been produced, the species $G$ is added back to the system.

We now define a configuration mapping as follows. If there exists the leader species $G$ in the configuration $\config{C'}$ then the configuration $\config{C'}$ maps to a configuration in $\crn{C_{IC}}$ that contains the same count of species $\lambda'$, $\forall \lambda' \in \species$. This is because the species $G$ exists in the system before (and after) applying a reaction $\reaction_i \in \reactions$. All intermediate configurations $\config{C'}$ that do not contain $G$ do not map to anything in the original system. As described in Table \ref{tab:prod_break} when a new product species $p_j$ is generated, if there exists a zero species $z_{p_j}$, it is consumed making sure there never exists more than one copy of a zero species. Therefore, $\forall z_{\lambda_i}$ in the system, the count of $z_{\lambda_i} \in [0,1]$. Each configuration $\config{C'}$ in $\crn{C_{KVG}}$ for which $\cmap(\config{C'})$ is defined will contain one leader species $G$, a sequence of $\lambda_i \in \species$, and a unique sequence of zero species corresponding to $\lambda_i$ for which $\config{C'}[\lambda_i]=0$. Hence every $\config{C'}$ will have a unique mapping $\cmap(\config{C'})$.

\[\cmap(\config{C}') =
\begin{cases}
    \sum_{\lambda' \in \species} \config{C}'[\lambda']\cdot\single{\lambda}' & \text{if } \single{G} \in \Set{\config{C}'} \\
    \undefConfig & \text{otherwise.} 
\end{cases}\]

The configuration mapping function is polynomial-time computable because when the function is defined (i.e. $\single{G} \in \Set{\config{C}'})$ it computes the sum of vectors $\config{C}'[\lambda']\cdot\single{\lambda}'$ for all species in the simulated system. Therefore, the configuration mapping function is computable in $\mathcal{O}(\size{\species})$. Furthermore, for each configuration $\config{C}$ in $\crn{C_{IC}}$ there exists a configuration $\config{C'}$ in $\crn{C_{KVG}}$ that contains $G$ and $\forall \lambda_i \in \species, \config{C'}[\lambda_i] = \config{C}[\lambda_i]$, and $\config{C'}[z_{\lambda_i}]=1\; \text{if} \config{C}[\lambda_i]=0$. In other words, for every $\config{C}$ in $\crn{C_{IC}}$ there exists a $\config{C'}$ in $\crn{C_{KVG}}$ such that $\cmap(\config{C'}) = \config{C}$.

Now we show that $\crn{C_{KVG}}$ simulates any given $\crn{C_{IC}}$. This is done in two parts. We first show that $\crn{C_{IC}}$ follows $\crn{C_{KVG}}$, and then we prove that $\crn{C_{KVG}}$ models $\crn{C_{IC}}$.

\para{$\crn{C_{IC}}$ follows $\crn{C_{KVG}}$}
$\crn{C_{IC}}$ follows $\crn{C_{KVG}}$ if for any two given configurations $\config{A'}$ and $\config{B'}$ in $\crn{C_{KVG}}$ where $\cmap(\config{A'})$ and  $\cmap(\config{B'})$ are defined, such that $\config{A'} \macrotrans_{\crn{C_{KVG}}} \config{B'}$ and $\cmap(\config{A'}) \not= \cmap(\config{B'})$, then $\cmap(\config{A'})\rightarrow_{\crn{C_{IC}}}\cmap(\config{B'})$. We know that both $\config{A'}$ and $\config{B'}$ contain the leader $G$ because $\cmap(\config{A'})$ and $\cmap(\config{B'})$ are defined. Because $\cmap(\config{A'}) \not= \cmap(\config{B'})$, there exists an applicable reaction $\reaction_i \in \reactions$ for which there will be reactions $\single{G} + \config{\RM_i} + \config{Z_i} \rightarrow \single{G}_i + \config{\RM_i} + \config{Z_i}$ and $\single{G}_{i} + \config{\RM_i} \SeqProds \single{G} + \config{\PM_i}$ that are applicable in $\crn{C_{KVG}}$. And the configuration $\config{A'}$ is the configuration before the reaction $\single{G} + \config{\RM_i} + \config{Z_i} \rightarrow \single{G}_i + \config{\RM_i} + \config{Z_i}$ is applied and $\config{B'}$ is the resulting configuration after the sequence of reactions represented by $\single{G}_{i} + \config{\RM_i} \SeqProds \single{G} + \config{\PM_i}$ is applied. The system will start with the configuration $\config{A'}$ and transition through certain intermediate configurations consisting of species $G_i, P_i^1, \ldots, P_i^{\size{\PM_i}+1}$, finally resulting in configuration $\config{B'}$. In the original system $\crn{C_{IC}}$, the configuration $\cmap(\config{A'})$ containing the same count of species in $\species$ as in $\config{A'}$ will transition to the configuration $\cmap(\config{A'})-\config{\RM_i}+\config{\PM_i}$ as described in Definition \ref{def:inhibitory-dynamics}. The resulting configuration $\cmap(\config{B'}) = \cmap(\config{A'})-\config{\RM_i}+\config{\PM_i}$ in $\crn{C_{IC}}$ will contain the same counts of species in $\species$ as in $\config{B'}$. Hence $\cmap(\config{A'}) \rightarrow_{\crn{C_{IC}}} \cmap(\config{B'})$, if $\config{A'} \macrotrans_{\crn{C_{KVG}}} \config{B'}$ and $\cmap(\config{A'}) \not= \cmap(\config{B'})$.

\para{$\crn{C_{KVG}}$ models $\crn{C_{IC}}$}
Say $\config{A}$ and $\config{B}$ are two configurations in $\crn{C_{IC}}$ such that $\config{A} \rightarrow_{C_{IC}} \config{B}$. As described in the configuration of configuration mapping,  for any configuration $\config{A}$ we only have one configuration $\config{A'} \in [\![\config{A}\,]\!]$ which contains the same count of species $\lambda_i \in \species$ as in $\config{A}$. Similarly the configuration $\config{B}$ will have only one configuration $\config{B'} \in [\![\config{B}\,]\!]$. Both configurations $\config{A'}$ and $\config{B'}$ will contain species $G$ because their mapping is defined. We know that $\config{A} \rightarrow_{C_{IC}} \config{B}$, hence there exists an applicable reaction $\reaction_i \in \reactions$ such that $\config{B} = \config{A}-\config{\RM_i}+\config{\PM_i}$ for which \ inhibitors $\mathcal{I}(\reaction_i)$ are absent. If such a reaction is applicable and the system $\crn{C_{KVG}}$ is in configuration $\config{A'}$ containing $G$, there will exist a sequence of applicable reactions in $\crn{C_{KVG}}$ (representing $\reaction_i$) leading to configuration $\config{B'}$ that also contains $G$ such that $\sum_{\lambda_i\in\species}\config{B'}[\lambda_i]\cdot\single{\lambda}_i = \sum_{\lambda_i\in\species}\config{A'}[\lambda_i]\cdot\single{\lambda}_i - \config{\RM_i}+\config{\PM_i}$. By the definition of configuration mapping, $\config{B'} \in [\![\config{B}\,]\!]$ (i.e., $\cmap(\config{B'}) = \config{B}$). Therefore for any two configurations $\config{A}$ and $\config{B}$ such that $\config{A}\rightarrow_{\crn{C_{IC}}}\config{B}$ and $\config{A'} \in [\![\config{A}\,]\!]$ and $\config{B'} \in [\![\config{B}\,]\!]$, $\config{A'}\macrotrans_{\crn{C_{KVG}}}\config{B'}$.

Because (1) $\crn{C_{IC}}$ follows $\crn{C_{KVG}}$, and (2) $\crn{C_{KVG}}$ models $\crn{C_{IC}}$ using a polynomial-time computable function $\cmap$, we can say that $\crn{C_{KVG}}$ simulates $\crn{C_{IC}}$. 

\para{Polynomial Simulation}
We now show that the above simulation is polynomial efficient as follows.

\begin{enumerate}
    \item \textbf{polynomial species and rules:} As defined in the construction, along with the original species in $\species$, the set $\species'$ contains $\size{\species}$ zero species, $(\size{\reactions}+1)$ leader species, and $\size{\species}\cdot\size{\reactions}$ species to generate products. Hence, $\size{\species'} = 2\cdot\size{\species} + \size{\reactions}+1 + \size{\species}\cdot\size{\reactions}$. Because for any CRN the $\size{\reactions}=\mathcal{O}(\size{\species}^c)$ for some constant $c$, therefore, $\size{\species'}=\mathcal{O}(\size{\species}^c)$. As given in Table \ref{tab:VG > i}, for each reaction in $\reaction$ we have one reaction for selection and $\size{\species}+1$ reactions for sequential product generation (as discussed in Table \ref{tab:prod_break}). Therefore, $\size{\reactions'} = \size{\reactions}\cdot(\size{\species}+2)$ which is $\mathcal{O}(\size{\reactions}^2)$.

    \item \textbf{polynomial rule size:} Each rule $\reaction_i \in \reactions$ of the form $\RM_i \rightarrow \PM_i$ is simulated by $\size{\species}+2$ rules. For the selection rule we add $\mathcal{O}(\size{\RM_i})$ species to the reactants and products, while for sequential products we only add a constant number of species to the reactants and products of $\reaction_i$. Therefore, $\forall \reaction_i \in \reactions$, every rule simulating $\reaction_i$ is polynomial in the size of $\reaction_i$.
    
    \item \textbf{polynomial transition sequences:} For any rule $\reaction_i$ transitioning to a rule $\reaction_j$ in the $\crn{C_{IC}}$, $\mathcal{O}(\size{\species})$ intermediate rules are applied in the $\crn{C_{KVG}}$. Therefore, any sequence of $n$ transitions in the $\crn{C_{IC}}$ is simulated by $\mathcal{O}(\size{\species})\cdot n$ transitions in the $\crn{C_{KVG}}$.
    
    \item \textbf{polynomial volume:} Every configuration $\config{C'}$ in $\crn{C_{KVG}}$ will contain one copy of one of the leader species and at most one copy of the $z_{\lambda_i}$ species $\forall \lambda_i \in \species$ along with the original species in $\species$. If the configuration $\config{C}'$ maps to a configuration $\config{C}$ in $\crn{C_{IC}}$ then $\config{C}'[\lambda_i] = \config{C}[\lambda_i]$ and the volume of system in $\config{C}'$ is polynomial in the volume of $\crn{C_{IC}}$ in configuration $\config{C}$. For any other configuration $\config{C}''$ that maps to $\undefConfig$, the volume is bounded by the volume of configurations $\config{A'}$ and $\config{B}'$ such that $\config{A}'\macrotrans_{\crn{C_{KVG}}}\config{B'}$ through configurations $\config{C}''$, i.e. is polynomial in the simulated system. This is because, to generate the products using the second rule in Table \ref{tab:VG > i}, a copy of each product is added to the system in every one of the $\config{C}''$ configurations.
    
\end{enumerate}
The above simulation only utilizes 1) polynomial species and rules, 2) polynomial rule size, 3) polynomial transition sequences, and 4) polynomial volume. Therefore, $\crn{C_{KVG}}$ simulates $\crn{C_{IC}}$ under polynomial simulation.
\end{proof}

\vgicrnthm*
\begin{proof}
    Lemma \ref{lem:i > VG} and Lemma \ref{lem:VG > i} show that both $C_{IC}$ and $C_{KVG}$ simulate each other under polynomial simulation. Therefore the $k$-VG CRN and Inhibitory CRN models are equivalent under polynomial simulation. Theorem \ref{thm:vg-kvg} proves that the $k$-VG CRN model is equivalent to the Void Genesis model under polynomial simulation. Based on the Transitivity Theorem (Thm. \ref{thm:transitivity}), the Void Genesis model is equivalent to the Inhibitory CRN model under polynomial simulation.
\end{proof}


\crkvg*

\begin{table}[H]
    \centering
    \CRsimKVG
    \caption*{Table~\ref{tab:cr > VG} (restated): Coarse-Rate CRN simulating $k$-VG CRN}
\end{table}
\begin{proof}
For any given $k$-VG CRN $\crn{C_{KVG}}$ and a configuration $\config{C}$, we show that there exists a Course-Rate CRN $\crn{C_{CR}}$ and a configuration $\config{C}'$ that simulates $\crn{C_{KVG}}$ over $\config{C}$.

\para{Construction}
Given a $k$-VG CRN $\crn{C_{KVG}} = ((\species,\reactions),Z_\emptyset)$, we construct a Course-Rate CRN $\crn{C_{CR}} = ((\species',\reactions'),rank')$ as follows.

In addition to $\species$, we first create the species $G$ to select a reaction $\reaction_i\in\reactions$. We also add the species $e_{\lambda_1}, \ldots, e_{\lambda_{\size{\species_1}}}$ to check if any reactant $r_j\in\RM_i$ has reached a count of zero from simulating $\reaction_i$. Finally $x$ and $y$ are constructed to modify $e_{\lambda_i}$ depending on if the new count of $\lambda_i$ is non-zero or zero, respectively.
Then $\species' = \species \cup \{ x, y, G, e_{\lambda_1}, \ldots, e_{\lambda_{\size{\species_1}}} \}$. Let the configuration of $\crn{C_{CR}}$ $\config{C}'$ be $\config{C}$ in addition to one copy of $G$. 

Reaction 1 in Table \ref{tab:cr > VG} simulates $\reaction_i \in \reactions$, which is selected by $\single{G}$. In addition to producing $\PM_i$, the reaction also creates the species $\single{x}$ and $\single{e}_{\lambda_1}, \ldots, \single{e}_{\lambda_{\size{\species_1}}}$. With Reaction 2, $\single{x}$ will then remove a copy of $\single{e}_{\lambda{_i}}$ only if the new count of $\lambda_i$ is not zero. If Reaction 2 can no longer be applied, any remaining copies of $\single{e}_{\lambda{_i}}$ must be represented by a species $\lambda_i$ with a count of zero. Afterwards, Reaction 4 transforms $\single{x}$ into $\single{y}$. Now, through Reaction 3, $\single{y}$ will transform all remaining copies of $\single{e}_{\lambda_i}$ into $\single{z}_{\lambda_i}$. Finally, Reaction 5 converts $\single{y}$ into $\single{G}$, allowing us to choose another reaction from $\reactions$.

We define the mapping function as follows. The intuition is that every configuration in $\crn{C_{CR}}$ maps to a configuration in $\crn{C_{KVG}}$ if the counts of each species are exactly the same, ignoring the added $G$, $x$, and $y$ species in the $\crn{C_{CR}}$ system. Ensuring that species have not reached a count of $0$ is done through the $e_{\lambda_i}$ species, which is why any configuration that contains an $e_{\lambda_i}$ is undefined. 

\[\cmap(\config{C}') =
\begin{cases}
    \undefConfig & \text{if}\, \Set{\config{C}'} \bigcap \, \Set{e_{\lambda_1}, \ldots, e_{\lambda_{\size{\species_1}}}} \neq \emptyset\\
    \sum_{\lambda' \in \species} \config{C}'[\lambda'] \cdot \single{\lambda'} & \text{otherwise}
\end{cases}\]

\para{$\crn{C_{KVG}}$ follows $\crn{C_{CR}}$}
Let $\config{A},\config{B}$ be configurations in $\crn{C_{KVG}}$ and $\config{A'},\config{B'}$ be configurations in $\crn{C_{CR}}$ with $M(\config{A'}) \neq M(\config{B'})$. The initial configuration for $C_{CR}$ starts with the same initial configuration of $C_{KVG}$ and one copy of $G$. It is clear that $M(\config{A'}) = \config{A}$. We first show that if $\config{A'} \macrotrans_{\crn{C_{CR}}} \config{B'}$, and $M(\config{A'}) \neq M(\config{B'})$, then $\config{A} \macrotrans_{\crn{C_{KVG}}} \config{B}$. With only one copy of $G$ in the system, a rule of the form $\single{G} + \config{\RM_i} \rightarrow \single{x}+\single{e}_{\{\RM_i\}} + \config{\PM_i}$ must occur first. As a result of this reaction, we are left with $\single{x}$ along with $\single{e}_{\{\RM_i\}}$ and the products from the chosen reaction. Therefore the rule $\single{x} + \single{\lambda}_i + \single{e}_{\lambda_i} \rightarrow \single{x} + \single{\lambda}_i$ occurs next, in which a copy of $\single{e}_{\lambda_i}$ is consumed if $\lambda_i$ exists in the system. Let the resulting configuration be $\config{X}'_1$. This reaction is then applied for all $\lambda_i \in \config{\RM}_i$. After exhaustively applying the reaction, 
only the new reaction of the form $\single{x} \rightarrow \single{y}$ can be executed. If any count of $\lambda_i$ reached zero, then there still exists a $\single{e}_{\lambda_i}$ respective to $\lambda_i$ that did not get consumed. This allows a reaction of the form $\single{y} + \single{e}_{\lambda_i} \rightarrow \single{y} + \single{z}_{\lambda_i}$ to trigger, creating a $\single{z}_i$ corresponding directly to $\single{z}_i \in \config{B}$. Let this be $\config{X}'_2$. Following the exhaustive application of that reaction, only then can the rule of the form $\single{y} \rightarrow \single{G}$ be applied, thus reaching $\config{B'}$. From the mapping function, since $M(\config{A'})=\config{A}$, each $\lambda_i \in \{\config{\RM_j}\}$ satisfies $\config{A'}[\lambda_i']=\config{A}[\lambda_i]$, where $\lambda_i'$ is the mapped species of $\lambda_i$ in $\crn{C_{CR}}$. Thus, if $\reaction_j'$ can be applied in $\crn{C_{CR}}$, it must be the case that $\reaction_j$ can be applied in $\crn{C_{KVG}}$. It follows that $M(\config{A'}) \rightarrow_{\crn{C_{KVG}}} M(\config{B'})$.

\para{$\crn{C_{CR}}$ models $\crn{C_{KVG}}$} The final part of the proof is to show that if $\config{A} \rightarrow_{\crn{C_{KVG}}} \config{B}$ implies that $\forall \config{A'}\in [\![\config{A}\,]\!]$, $\exists \config{B'}\in [\![\config{B}\,]\!]$ such that $\config{A'} \Rightarrow_{\crn{C_{CR}}} \config{B'}$. Let $\reaction_j \in \reactions$ denote a rule that can be applied to $\config{A}$. By the mapping function, there exists a species $\lambda_i'$ for each $\lambda_i \in \{\config{\RM_j}\}$ satisfying $\config{A'}[\lambda_i']=\config{A}[\lambda_i]$. Thus, as long as $\config{A'} \in [\config{A}\,]\!]$ contains an $G$ species, the rule $\single{G} + \config{\RM_i} \rightarrow \single{e}_{\{\RM_i\}} + \config{\PM_i} \in \reactions'$ must also be executable. From \textit{$\crn{C_{KVG}}$ follows $\crn{C_{CR}}$}, we know that there are at most 3 representative configurations for $\config{A}$: one with a single $x$ species ($\config{X}'_1$), one with an $y$ species but no $z_{\lambda_i}$ (between $\config{X}'_1$ and $\config{X}'_2$), and one with a $y$ species and a discrete amount of $z_{\lambda_i}$ species ($\config{X}'_2$). Assume that $\config{L'}$ was the configuration prior to reaching configuration $\config{X}'_1$. In the 3 cases, we show that configuration $\config{A'}$ is reachable through the sequence $\config{L'} \Rightarrow \config{A_1'} \Rightarrow \config{A_2'} \Rightarrow \config{A'}$. Thus, since $\config{A'}$ is reachable from all other representative configurations and since $\config{A'} \Rightarrow \config{B'}$, $\config{A} \rightarrow_{\crn{C_{KVG}}} \config{B}$ implies that $\forall \config{A'}\in [\![\config{A}\,]\!]$, $\exists \config{B'}\in [\![\config{B}\,]\!]$ such that $\config{A'} \Rightarrow_{\crn{C_{CR}}} \config{B'}$.

\para{Polynomial Simulation} We finally show that the simulation is polynomial. 
\begin{enumerate}
    \item \textbf{polynomial species and rules:}  As described in the construction, the set of species of $\crn{C_{CR}}$ contains $\species$ (from $\crn{C_{KVG}}$), the added $x,y,z$ species and the checker species $e_{\lambda_1}, \ldots, e_{\lambda_{\size{\species_1}}}$. The set of rules only adds $\size{\reactions}+2\cdot\size{\species}$ fast reactions, as well as two slow reactions.

    \item \textbf{polynomial rule size:} For a rule $\reaction \in \reactions$ of size $(m,n)$, a rule of at most size $(m+1,n+m)$ is created for $\reactions'$.
    
    \item \textbf{polynomial transition sequences:} For any rule $\reaction_i$ transitioning to a rule $\reaction_j$ in the $\crn{C_{CR}}$, $\mathcal{O}(\size{\species})$ intermediate rules are applied in the $\crn{C_{KVG}}$. Therefore, any sequence of $n$ transitions in the $\crn{C_{CR}}$ is simulated by $\mathcal{O}(\size{\species})\cdot n$ transitions in the $\crn{C_{KVG}}$.
    
    \item \textbf{polynomial volume:} Each configuration $\config{C'}$ has a defined mapping only if each $e_{\lambda_i}$ species is not present in $\config{C'}$. Thus, any defined mapping will only differ by a constant amount from configuration $M(\config{C}')$. Any intermediate configuration between $\config{C}'$ and some reachable configuration $\config{C}''$ with a defined mapping will have at most an increase of $\mathcal{O}(|\species|)$ volume introduced by the $e_{\lambda_i}$ species. Thus, the difference in volume from $\config{C}'$ will differ only by a polynomial amount.
\end{enumerate}
\end{proof}

\begin{table}[H]
        \centering
        \KVGsimCR
        \caption*{Table~\ref{tab:VG > cr} (restated): $k$-VG CRN simulating Coarse-Rate CRN.}
\end{table}
\kvgcr*
\begin{proof}
For any given Coarse-Rate CRN $\crn{C_{CR}} = ((\species, \reactions), rank)$ and a configuration $\config{C}$ of $\crn{C_{CR}}$, we show that there exists a $k$-VG CRN $\crn{C_{KVG}} = ( (\species', \reactions'), \Zero' )$ that simulates $\crn{C_{CR}}$ over $\config{C}$.

\para{Construction} Given a Coarse-Rate CRN $\crn{C_{CR}} = ((\species, \reactions), rank)$ and a configuration of $\crn{C_{CR}}$ $\config{C}$, we construct a $k$-VG CRN $\crn{C_{KVG}} = ( (\species', \reactions'), \Zero' )$ and a configuration of 
$\crn{C_{KVG}}$ $\config{C'}$ as follows. Let $\reactions^2$ and $\reactions^1$ be the set of all fast (rank 2) and slow (rank 1) reactions of $\crn{C_{CR}}$, respectively.

We first construct the species set of $\crn{C_{KVG}}$ $\species'$. Each species of $\crn{C_{CR}}$ $\lambda_i\in\species$ are added with no modification and we create the new species $G,s,t$,
$g_1, \ldots, g_{\size{\reactions^2}}$,
$z_{g_1}, \ldots, z_{g_{\size{\reactions^2}}}$,
$r_1, \ldots, r_{\size{\reactions^2}+1}$,
$z_{\lambda_1}, \ldots, z_{\lambda_{\size{\Lambda}}}$, and
$P_1^1, \ldots, P_{\size{\reactions^1}}^{\size{\PM_{\size{\reactions^1}}} + 1}$.
We then define $\Zero'$ as the mapping $g_i\rightarrow z_{g_i}$ and $\lambda_i\rightarrow z_{\lambda_i}$. $G$'s role is attempting to apply a random \emph{fast} reaction $\reaction_i^2\in\reactions^2$ by interacting with the respective $g_i$ species. The presence of a zero species for a $g_i$ species ($z_{g_i}$) indicates the corresponding fast reaction cannot be applied in the current configuration by $G$. After a fast reaction is applied, we use each $r_i$ species to restore any missing $g_i$ species in the configuration to allow all fast reactions to be re-applicable again. $s$ is produced when no fast reactions are applicable and is used to attempt applying a random \emph{slow} reaction $\reaction_i^1\in\reactions^1$ instead. $t$ is made if $\reaction_i^1$ was applied and is used to restore the species $G$ and $g_1, \ldots, g_{\size{\reactions^2}}$ back in the configuration in case any fast reaction becomes applicable again following the application of $\reaction_i^1$. Finally, $P_1^1, \ldots, P_{\size{\reactions^1}}^{\size{\PM_{\size{\reactions^1}}} + 1}$ is used to generate the products of slow reactions using the sequential products technique described in Subsubsection \ref{subsec: techniques}; we note that we use $s$ to initiate the procedure and produce $t$ at the end instead of $G$.
We also define the configuration of $\crn{C_{KVG}}$ $\config{C'}$ to be the configuration of $\crn{C_{CR}}$ $\config{C}$ with the addition of one copy of $G$ and one copy for each of $g_1, \ldots, g_{\size{\reactions^2}}$.

We now construct the reaction set of $\crn{C_{KVG}}$ $\reactions'$. We describe each reaction as listed in Table \ref{tab:VG > cr}. Reaction 1 applies a fast reaction $\reaction_i^2\in\reactions^2$. Specifically, $\single{G}$, $\single{g}_i$, and the reactants $\config{\RM_i}$ will be consumed, creating the products $\config{\PM_i}$ and a species $\single{r}_1$. $\single{r}_1$ is then used to iterate through all fast reactions $\reaction_1^2, \ldots, \reaction_n^2$ by transforming into $\single{r}_2$, which then transforms into $\single{r}_3$, and so on, restoring any missing $\single{g}_i$ species if needed (Reactions 2 and 3). Once the $\single{r}_{|\reactions^2|+1}$ species is created, Reaction 4 converts it into $\single{G}$, allowing Reaction 1 to occur again.
If one of the reactants of a fast reaction $\reaction_i^2$ is detected absent (by a $\single{z}_k$ species), the respective $\single{g}_{i}$ species for the reaction is deleted, in turn creating a $\single{z}_{g_j}$ species (Reaction 1b). If all $\single{z}_{g_j}$ species $ \single{z}_{g_1}, \ldots, \single{z}_{g_{\size{\reactions^2}}}$ are present, then no fast reaction is applicable, allowing any slow reaction to be selected instead. To do this, we first consume $\single{G}$ and $\single{z}_{g_1}, \ldots, \single{z}_{g_{\size{\reactions^2}}}$ to create the species $\single{s}$ (Reaction 5). Reaction 7 uses $\single{s}$ to apply a slow reaction $\reaction_i^1$. This reaction uses the sequential products procedure from Subsubsection \ref{subsec: techniques}, meaning that we produce each product of $\reaction_i^1$ sequentially, allowing for the deletion of all $\single{z_{p_j}}$ species that correspond to a product of $\reaction_i^1$. The procedure ends by creating the species $\single{t}$, which restores $\single{G}$ and $\single{g}_1, \ldots, \single{g}_{|\reactions|}$ with Reaction 6, allowing any potentially new applicable fast reactions to be selected over the slow reactions.

We finally define the mapping function $\cmap(\config{C}')$ as follows. If $\single{G}$ exists in a configuration $\config{C}'$, then that configuration maps to a configuration $\config{C}$ in which the count of all species $\lambda_i \in \Lambda$ in $\config{C}$ is equal to the count of all species $\lambda_i' \in \Lambda'$ in $\config{C}'$. The intuition of our mapping is that $\single{G}$ appears in the configuration before and after both a slow and fast reaction application. With a fast reaction $\reaction_i^2$, $\single{G}$ is used to apply it in Reaction 1. Then, after Reaction 4, $\single{G}$ is added back, allowing a new fast reaction to be applied. In order for a slow reaction to be applied, $\single{G}$ is deleted to create the species $\single{s}$ (Reaction 5). When a slow reaction is applied with $\single{s}$, $\single{G}$ will be restored back following Reaction 6.

\[\cmap(\config{C}') =
\begin{cases}
    \sum_{\lambda' \in \species} \config{C}'[\lambda']\cdot\single{\lambda}' & \text{if } \single{G} \in \Set{\config{C}'} \\
    \undefConfig & \text{otherwise.}

\end{cases}\]

\para{$\crn{C_{CR}}$ follows $\crn{C_{KVG}}$} Let $\config{A'}$ and $\config{B'}$ be configurations in $\crn{C_{KVG}}$. We show that if $\config{A'} \macrotrans_{\crn{C_{KVG}}} \config{B'}$ and $M(\config{A'}) \not = M(\config{B'})$, $\cmap(\config{A'})\rightarrow_{\crn{C_{CR}}}\cmap(\config{B'})$. First, $\config{A'}$ and $\config{B'}$ must contain $\single{G}$ in order for both configurations to be defined. However, since $M(\config{A'}) \not = M(\config{B'})$, there must exist an applicable sequence of reactions to transform the count of the species $\lambda' \in \Lambda'$ of $\config{A'}$ to the count of the species $\lambda' \in \Lambda'$ of $\config{B'}$. We consider two cases for this transition:

\textit{Case 1: Fast Reaction.}
Assume a rule of the form $\single{G}+\single{g}_i+\config{\RM_i}\rightarrow\single{r}_1+\config{\PM_i}$ is applicable to $\config{A'}$. Let the resulting new configuration be $\config{X'_1}=\config{A'}-\single{G}-\single{g}_i-\config{\RM_i}+\single{r}_1+\config{\PM_i}$. We then fire the reactions $\single{r}_j + \single{g}_j \rightarrow \single{r}_{j+1} + \single{g}_j$ or $\single{r}_j + \single{z}_{g_j} \rightarrow \single{r}_{j+1} + \single{g}_j$ for all $1, \ldots, |\reactions^2|$, restoring any $\single{g}_j$ species missing in the configuration. Denote the set of restored $\single{g}_j$ species (besides $\single{g}_i$ from the first applied rule) as $\mathcal{G}$ and the set of removed $\single{z}_{g_j}$ species as $\mathcal{Z_G}$. After the reaction $\single{r}_{|\reactions^2|+1} \rightarrow \single{G}$ is ran, the configuration $\config{B'}=\config{A'}-\mathcal{Z_G}-\config{\RM_i}+\mathcal{G}+\config{\PM_i}$ is reached. Clearly, $\config{X'_1}$ and any intermediate configuration between  $\config{X'_1}$ and $\config{B'}$ have an undefined mapping because $\single{G}$ is not present. $\config{A'}$ and $\config{B'}$ then have a defined mapping, but because $\config{B'}=\config{A'}-\mathcal{Z_G}-\config{\RM_i}+\mathcal{G}+\config{\PM_i}$, $M(\config{A'}) \not = M(\config{B'})$.
Denote the first applied rule $\single{G} + \single{g}_i + \config{\RM_i} \rightarrow \single{r}_1 + \config{\PM_i}$ as ${\reaction_i'}^2$. By the construction of $\crn{C_{KVG}}$, ${\reaction_i'}^2$ must be derived from a fast reaction $\reaction_i^2 \in \reactions^2$ in which $\RM_i \rightarrow \PM_i$. If ${\reaction_i'}^2$ is applicable for all $\lambda' \in \Lambda'$ in $\config{A'}$, because $M(\config{A'})=\config{A}$, it must follow that $\reaction_i^2$ is applicable for all $\lambda \in \Lambda$ in $\config{A}$. Applying $\reaction_i^2$ to $\config{A}$ then results in the configuration $\config{B}=\config{A}-\RM_i+\PM_i$. Clearly, $M(\config{B'})=\config{B}$. Thus, $\cmap(\config{A'})\rightarrow_{\crn{C_{CR}}}\cmap(\config{B'})$ when considering fast reactions.

\textit{Case 2: Slow Reaction.}
Assume no rule of the form $\single{G} + \single{g}_i + \config{\RM_i} \rightarrow \single{r}_1 + \config{\PM_i}$ is applicable to $\config{A'}$. For simplicity, assume a $\single{z}_{g_i}$ species already exists for all of $1,\ldots,\size{\reactions^2}$ in $\config{A'}$ and denote this set of all $\single{z}_{g_i}$ species $\mathcal{Z_G}$. 
Then, only the reaction $\single{G} + \single{z}_{g_1} + \ldots + \single{z}_{g_{|\reactions^2|}} \rightarrow \single{s}$ can be applicable to $\config{A'}$. Let the resulting new configuration be $\config{X'_1}=\config{A'}-G- \mathcal{Z_G}+\single{s}$. Assume a rule of the form $\single{s} + \config{\RM_i} \SeqProds \single{t} + \config{\PM_i}$ can be applied to $\config{X'_1}$, which is then followed by the reaction $\single{t} \rightarrow \single{G} + \single{g}_1 + \ldots + \single{g}_{|\reactions^2|}$. Denote the set of all $\single{g}_i$ species $\mathcal{G}$. The resulting configuration is then $\config{B'}=\config{A'}-\mathcal{Z_G}-\config{\RM_i}+\mathcal{G}+\config{\PM_i}$. Similar to Case 1, only $\config{A'}$ and $\config{B'}$ have a defined mapping, but $M(\config{A'}) \not = M(\config{B'})$. 
Denote the applied rule $\single{s} + \config{\RM_i} \SeqProds \single{t} + \config{\PM_i}$ as ${\reaction_i'}^1$.
By the construction of $\crn{C_{KVG}}$, because no rule of the form $\single{G} + \single{g}_i + \config{\RM_i} \rightarrow \single{r}_1 + \config{\PM_i}$ was applicable in $\config{A'}$, it follows from Case 1 that there does not exist a fast reaction ${\reaction_i}^2$ that is applicable in $\config{A}$. Now consider ${\reaction_i'}^1$, which is created from a slow reaction $\reaction_i^1 \in \reactions^1$ in which $\RM_i \rightarrow \PM_i$. Because $\reaction_i^1$ is applicable to $\config{X'_1}$ for all $\lambda' \in \Lambda'$, since $\lambda' \in \Lambda'$ is unmodified across both $\config{A'}$ and $\config{X'_1}$ \emph{and} $M(\config{A'})=\config{A}$, it must follow that $\reaction_i^1$ is applicable for all $\lambda \in \Lambda$ in $\config{A}$. Applying $\reaction_i^1$ to $\config{A}$ then results in the configuration $\config{B}=\config{A}-\RM_i+\PM_i$, and $M(\config{B'})=\config{B}$. Thus, $\cmap(\config{A'})\rightarrow_{\crn{C_{CR}}}\cmap(\config{B'})$ when considering slow reactions.

\para{$\crn{C_{KVG}}$ models $\crn{C_{CR}}$} Let $\config{A}$ and $\config{B}$ be configurations in $\crn{C_{CR}}$, and $\config{A'}$ and $\config{B'}$ be configurations in $\crn{C_{KVG}}$. We show that $\config{A} \rightarrow_{\crn{C_{CR}}} \config{B}$ implies that $\forall \config{A'}\in [\![\config{A}\,]\!]$, $\exists \config{B'}\in [\![\config{B}\,]\!]$ such that $\config{A'} \Rightarrow_{\crn{C_{KVG}}} \config{B'}$. First, consider the case of using fast reactions. Let $\reaction_i^2\in\reactions^2$ be a fast reaction that reaches $\config{B}$ from $\config{A}$. By the mapping function, $M(\config{A'})=\config{A}$, so $\Set{\lambda' \in \Lambda'}=\Set{\lambda \in \Lambda}$. Furthermore, $M(\config{A'})$ is defined, so $\single{G}$ must be present in $\config{A'}$. Then by the construction of $\crn{C_{KVG}}$, there must exist a rule of the form $\single{G}+\single{g}_i+\config{\RM_i}\rightarrow\single{r}_1+\config{\PM_i}$ that is created from a fast reaction $\reaction_i^2 \in \reactions^2$ and is applicable in $\config{A'}$. As shown in \emph{$\crn{C_{CR}}$ follows $\crn{C_{KVG}}$}, applying the rule results in $\config{B'}=\config{A'}-\mathcal{Z_G}-\config{\RM_i}+\mathcal{G}+\config{\PM_i}$, so $M(\config{A'}) \not = M(\config{B'})$. Because $M(\config{B'})=\config{B}$, $\config{A} \rightarrow_{\crn{C_{CR}}} \config{B}$ implies $\forall \config{A'}\in [\![\config{A}\,]\!]$, $\exists \config{B'}\in [\![\config{B}\,]\!]$ such that $\config{A'} \Rightarrow_{\crn{C_{KVG}}} \config{B'}$. A similar case is made for using slow reactions.

Because $\crn{C_{CR}}$ follows $\crn{C_{KVG}}$ and $\crn{C_{KVG}}$ models $\crn{C_{CR}}$, $\crn{C_{KVG}}$ simulates $\crn{C_{CR}}$.

\para{Polynomial Simulation} We finally show that the simulation is polynomial efficient. Let $m$ be the maximum amount of products from any slow reaction, and $n$ the maximum amount of reactants from a fast reaction.

\begin{enumerate}
    \item \textbf{polynomial species and rules:} The species set of $\crn{C_{KVG}}$ is made from $2\cdot\size{\species}+3\cdot\size{\reactions^2}+m\cdot\size{\reactions^1}+4$ created species, resulting in $\mathcal{O}(\size{\species}+\size{\reactions^2}+m\cdot\size{\reactions^1})$ total species. The reaction set of $\crn{C_{KVG}}$ consists of $(3+n)\cdot\size{\reactions^2}+2\cdot\size{\reactions^1}\cdot(m+1)+3$ created reactions, resulting in $(n+1)\cdot\size{\reactions^2}+(m+1)\cdot\size{\reactions^1}$ total reactions.

    \item \textbf{polynomial rule size:} The maximum reactant size for a reaction in $\crn{C_{KVG}}$ is the consumption of all $\size{\reactions^2}\cdot\single{z_{g_i}}$ species for applying a slow reaction. Likewise, the maximum product size is re-introducing all $\size{\reactions^2}\cdot\single{z_{g_i}}$ species after applying the slow reaction.
    
    \item \textbf{polynomial transition sequences:} In the worst case, $M(\config{A'}) \macrotrans_{\crn{C_{KVG}}} M(\config{B'})$ only through a slow reaction $\reaction_i^1\in\reactions^1$. As shown in the \emph{$\crn{C_{CR}}$ follows $\crn{C_{KVG}}$} section, this can require the removal of all $\single{g}_{i}$ species for $1,\ldots,\size{\reactions^2}$ in $\config{A'}$, resulting in an application sequence of length $\mathcal{O}(\size{\reactions^2})$. The remaining reaction sequence is of size $2\cdot m+4$, so the overall transition sequence length is of size $\mathcal{O}(\size{\reactions^2}+m)$.
    
    \item \textbf{polynomial volume:} Consider a configuration $\config{C'}$ that has a defined mapping. $\config{C'}$ then must contain one copy of $\single{G}$, one copy of either $\single{g}_i$ or $\single{z}_{g_i}$ for all $1,\ldots,\size{\reactions^2}$, and all $\config{C'}[\lambda']$ for all $\lambda'\in\Lambda$. Clearly, the volume of $\config{C'}$ is polynomial bounded by the volume of $\config{C}$. For  any macro transition $\config{A'} \macrotrans_{T'} \config{B'}$, the worst case for an intermediate configuration in the case of fast reactions is the introduction of $\size{\reactions^2}+1$ $\single{r_i}$ species. For slow reactions, only $s$ and $t$ are introduced. Thus, $\config{C'}$'s volume remains polynomial bounded by $\config{C}$'s.
\end{enumerate}
\end{proof}

\vgcrthm*
\begin{proof}
    Similar to Theorem \ref{thm:vgicrnthm}, the proof follows from Lemmas \ref{lem:cr > VG} and \ref{lem:VG > cr}, Theorem \ref{thm:vg-kvg} and the Transitivity Theorem \ref{thm:transitivity}. 
\end{proof}


\begin{table}[H]
    \centering
        \SCsimKVG
        \caption*{Table~\ref{tab:cycle > VG} (restated): Step-Cycle CRN simulating a $k$-VG CRN}
\end{table}
\sckvg*
\begin{proof}

For any given $k$-VG CRN $\crn{C_{KVG}} = ( (\species, \reactions), \Zero )$, we show that there exists a Step-Cycle CRN $\crn{C_{SC}}= ((\species', \reactions'), \config{S}_0)$ that simulates it. We start by describing the construction, then detail the formal proof.

\para{Construction} Given a $k$-VG CRN $\crn{C_{KVG}} = ((\species, \reactions), \Zero)$, we construct a Step-Cycle CRN $\crn{C_{SC}} = ((\species', \reactions'), \config{S}_0)$ as described. Let $\species_1, \species_2$ denote the set of non-zero/zero species in $\species$ respectively.
In addition to $\species$, $\crn{C_{SC}}$ has species $G$ to select a reaction and species $y$ to check if any $\lambda_i \in \species_1$ have reached a count of zero. This is accomplished with the help of a unique species for each $\lambda_i \in \species_1$, namely species $e_{\lambda_1}, \ldots, e_{\lambda_{\size{\species_1}}}$, which are produced as part of the selected rule. 
Thus, the initial configuration of $\crn{C_{SC}}$ should have one copy of $G$ in addition to the initial configuration of $\crn{C_{KVG}}$, with $\species' = \species \cup \{ G, y, x, w, e_{\lambda_1}, \ldots, e_{\lambda_{\size{\species_1}}}, z_{\lambda_1}, \ldots, z_{\lambda_{|\species_2|}} \}$, $\config{S}_0 = \single{y}$, and the reactions shown in Table \ref{tab:cycle > VG}.

Reaction $1$ in Table \ref{tab:cycle > VG} simulates $\reaction_i \in \reactions$, where species $G$ essentially `chooses' a reaction. Once a reaction is chosen, species $e_{\lambda_1}, \ldots, e_{\lambda_{\size{\species_1}}}$ are created as part of the product, with each $e_{\lambda_i}$ being removed through reaction $2$ only if $\lambda_i$ does not have a count of zero. Once the system is terminal, a $y$ is added, with reaction $3$ occurring for any $e_{\lambda_i}$ where the count of $\lambda_i$ is zero. As all $e_{\lambda_i}$ are removed or turned into a zero species, a $y$ is then added, enabling reaction $4$ to occur. This reaction `resets' the system by reintroducing the $G$ species, allowing for a new reaction $1$ to occur. If there no longer exists an applicable reaction $1$ (the system is now terminal), the addition of $y$ to the system allows reaction $5$ to be applied, which results in reaction $6$ executing infinitely. This reaction sequence simulates $\crn{C_{KVG}}$ reaching a terminal configuration and prevents the volume from growing infinitely.

Define the mapping function as follows. The intuition is that every step configuration in $\crn{C_{SC}}$ maps to a configuration in $\crn{C_{KVG}}$ if the count of each species $\lambda_i \in \Lambda$ matches the count of the corresponding species $\lambda_i' \in \Lambda'$ and if there exists either 1 $G$ species or 1 $y$ species. Ensuring that species have not reached a count of $0$ is done through the $e_{\lambda_i}$ species, which is why any configuration that contains an $e_{\lambda_i}$ is undefined.  

$$\cmap(\config{C_i}') =
\begin{cases}
    \sum_{\lambda' \in \species} \config{C_i}'[\lambda'] \cdot \single{\lambda}' & \text{if}\, \Set{\config{C_i}'} \bigcap \, \Set{e_{\lambda_1}, \ldots, e_{\lambda_{\size{\species_1}}}} = \emptyset \text{ and } | \{\config{C_i'}\} \bigcap \Set{G, y}| = 1 \\ 
    \undefConfig & \text{otherwise}
\end{cases}$$

The configuration mapping function is polynomial-time computable because when the function is defined, it computes the sum of vectors $\config{C}'[\lambda']\cdot\single{\lambda}'$ for all species in the simulated system. Therefore, the configuration mapping function is computable in $\mathcal{O}(\size{\species})$.
    
\para{$\crn{C_{KVG}}$ follows $\crn{C_{SC}}$} 
To start, note that for any step configuration $\config{C_i}$ in $\crn{C_{SC}}$, $i$ will be $0$. Thus, for simplicity, we exclude $i$ from the notation. Let $\config{A'}$, $\config{B'}$ be step configurations in $\crn{C_{SC}}$ with $M(\config{A'}) \not = M(\config{B'})$.
We first show that if $\config{A'}\macrotrans_{\crn{C_{SC}}} \config{B'}$, and $M(\config{A'})\neq M(\config{B'})$, then $M(\config{A'})\rightarrow_{\crn{C_{KVG}}} M(\config{B_n'})$. 
Since $M(\config{A'})$ and $M(\config{B'})$ are defined, there must not exist any $e_{\lambda_i}$ in either configuration, and either $G$ or $y$ must be present in $\config{A'}$. Assume $G \in \config{A'}$, and a rule of the form $\single{G} + \config{\RM_j} \rightarrow \single{e}_{\{\RM_j\}} + \config{\PM_j}$ is applicable to $\config{A'}$. As a result of this reaction, there are no $G$'s left in the system. Since $y$ is not initially present in $\config{A'}$, the rule $\single{\lambda}_i + \single{e}_{\lambda_i} \rightarrow \single{\lambda}_i$ occurs next, in which each $e_{\lambda_i}$ is consumed if $\lambda_i$ exists in the system. Denote the terminal configuration $\config{I'}$. A $y$ is then added to $\config{I'}$, resulting in the rule $\single{y} + \single{e}_{\lambda_i} \rightarrow \single{y} + \single{z}_{\lambda_i}$ occurring for each $e_{\lambda_i}$ with $\config{I'}[e_{\lambda_i}] > 0$. Denote the terminal configuration $\config{B'}$. 
Consider $2$ cases. If $\config{A'}-\config{\RM_j}$ does not result in a species $\lambda_i$ reaching a count of $0$, then $\config{I'}$ will not contain any $e_{\lambda_i}$ species, resulting in the new configuration $\config{I'}=\config{A'}-\config{\RM_j} - \single{G} + \config{\PM_j}$ with mapping $M(\config{I'}) = \undefConfig$. After the $y$ species is added, the resulting configuration $\config{B'}=\config{I'} + \single{y} = \config{A'}-\config{\RM_j}-\single{G} + \config{\PM_j} + \single{y}$ then has mapping $M(\config{B'}) \not = \undefConfig$. In a similar manner, if $\config{A'}-\config{\RM_j}$ results in some species $\lambda_i$ reaching a count of $0$, then $\config{I'}$ will contain $e_{\lambda_i}$, resulting in the mapping $M(\config{I'})=\undefConfig$. As a $y$ is added, the following configuration $\config{B'}=\config{I'} + \single{y} + \config{\mathcal{Z}} = \config{A'}-\config{\RM_j}-\single{G} + \config{\PM_j} + \single{y} + \config{\mathcal{Z}}$, where $\config{\mathcal{Z}}$ is the configuration of a single count of each $z_{\lambda_i}$ satisfying $\config{I}'[e_{\lambda_i}] > 0$, then has mapping $M(\config{B'}) \neq  \undefConfig$. 
In both cases, we reach a configuration $\config{B'}= \config{A'}-\config{\RM_j}-\single{G} + \config{\PM_j} + \single{y} + \config{\mathcal{Z}}$. It is clear that $M(\config{B'})$ is defined, with $M(\config{B'}) \not = M(\config{A'})$. From our initial assumption, $\config{B'}$ is only reachable from $\config{A'}$ if a rule of the form $\single{G} + \config{\RM_j} \rightarrow \single{e}_{\{\RM_j\}} + \config{\PM_j}$ can be executed. Denote this rule $\reaction_j'$. By construction, this implies there exists a rule $\reaction_j= \config{\RM_j} \rightarrow \config{\PM_j} \in \reactions$. From the mapping function, since $M(\config{A'})=\config{A}$, each $\lambda_i \in \{\config{\RM_j}\}$ satisfies $\config{A'}[\lambda_i']=\config{A}[\lambda_i]$, where $\lambda_i'$ is the mapped species of $\lambda_i$ in $\crn{C_{SC}}$. Thus, if $\reaction_j'$ can be applied in $\crn{C_{SC}}$, it must be the case that $\reaction_j$ can be applied in $\crn{C_{KVG}}$. Applying $\reaction_j$ to $\config{A}$ then results in configuration $\config{B}=\config{A}-\config{\RM_j}+\config{\PM_j} + \config{\mathcal{Z}}$, where $\config{\mathcal{Z}}$ denotes the set of zero species explicitly added by $\crn{C_{KVG}}$. It is clear that $M(\config{B'})=\config{B}$. It follows that $M(\config{A'}) \rightarrow_{\crn{C_{KVG}}} M(\config{B'})$. 
Now assume that $y \in \config{A'}$. Since $M(\config{A'})$ is defined, there must not exist any $e_{\lambda_i}$ and $G$ must not be present. Thus, $\config{A'}$ is terminal. In the next step, a $y$ is added. The rule $\single{y} + \single{y} \rightarrow \single{G}$ must then occur, resulting in new configuration $\config{A_1'}=\config{A'} - \single{y} + \single{G}$ with $M(\config{A_1'})=M(\config{A'})$. It follows from the above that $M(\config{A'}) \rightarrow_{\crn{C_{KVG}}} M(\config{B'})$. 

\para{$\crn{C_{SC}}$ models $\crn{C_{KVG}}$} The final part of the proof is to show that if $\config{A} \rightarrow_{\crn{C_{KVG}}} \config{B}$ implies that $\forall \config{A'}\in [\![\config{A}\,]\!]$, $\exists \config{B'}\in [\![\config{B}\,]\!]$ such that $\config{A'} \Rightarrow_{\crn{C_{SC}}} \config{B'}$. Let $\reaction_j \in \reactions$ denote a rule that can be applied to $\config{A}$. By the mapping function, there exists a species $\lambda_i'$ for each $\lambda_i \in \{\config{\RM_j}\}$ satisfying $\config{A'}[\lambda_i']=\config{A}[\lambda_i]$. Since $M(\config{A'})$ is defined, there must exist either a $G$ or $y$ species. In the latter, we have shown from \textit{$\crn{C_{KVG}}$ follows $\crn{C_{SC}}$} that configuration $\config{A'}$ is reached from $\config{A_1'}$ with $M(\config{A'})=M(\config{A_1'})$. We then consider when $G \in \config{A'}$. Assuming this is the case, from the fact that $\reaction_j$ is executable in $\crn{C_{KVG}}$, the rule $\single{G} + \config{\RM_j} \rightarrow \single{e}_{\{\RM_j\}} + \config{\PM_j} \in \reactions'$ must also be executable in $\crn{C_{SC}}$. As described in \textit{$\crn{C_{KVG}}$ follows $\crn{C_{SC}}$}, configuration $\config{B'}$ is then reachable such that $\config{B'}=\config{A'} - \config{\RM_j} + \config{\PM_j} + \config{\mathcal{Z}}$. Since $M(\config{B'})=\config{B'}$, $\config{A} \rightarrow_{\crn{C_{KVG}}} \config{B}$ implies that $\forall \config{A'}\in [\![\config{A}\,]\!]$, $\exists \config{B'}\in [\![\config{B}\,]\!]$ such that $\config{A'} \Rightarrow_{\crn{C_{SC}}} \config{B'}$.

\para{Polynomial Simulation} We now show that the above simulation is polynomial efficient as follows.
\begin{enumerate}
    \item \textbf{polynomial species and rules:} In addition to the already existing $|\species|$ species from $\crn{C_{KVG}}$, another $|\species|$ species are added, resulting in $|\species'|=\mathcal{O}(|\species|)$ total species. In addition to the already existing $|\reactions|$ rules from $\crn{C_{KVG}}$, another $|\reactions| + |\species|$ rules are added, resulting in $|\reactions'|= \mathcal{O}(|\reactions| + |\species|)$ total rules.

    \item \textbf{polynomial rule size:} For a rule $\reaction \in \reactions$ of size $(m,n)$, a rule of at most size $(m+1,n+m)$ is created for $\reactions'$.
    
    \item \textbf{polynomial transition sequences:} For each rule $\reaction_i \in \reactions$ applied in $\crn{C_{KVG}}$, as shown in \textit{$\crn{C_{KVG}}$ follows $\crn{C_{SC}}$}, rule $1$ produces at most $|\species|$ species that then take part in either rules $2$ or $3$, resulting in at most $|\species|$ additional rule applications. Thus, for a sequence of $n$ transitions in $\crn{C_{KVG}}$, the resulting simulating sequence is of length $\mathcal{O}(|\species| \cdot n)$.
    
    \item \textbf{polynomial volume:} Each configuration $\config{C}'$ has a defined mapping only if each $e_{\lambda_i}$ species is not present in $\config{C}'$. Thus, any defined mapping will only differ by a constant amount from configuration $M(\config{C}')$. Any intermediate configuration between $\config{C}'$ and some reachable configuration $\config{C}''$ with a defined mapping will have at most an increase of $\mathcal{O}(|\species|)$ volume introduced by the $e_{\lambda_i}$ species. Thus, the difference in volume from $\config{C}'$ will differ only by a polynomial amount.
\end{enumerate}
The above simulation only utilizes 1) polynomial species and rules, 2) polynomial rule size, 3) polynomial transition sequences, and 4) polynomial volume. Therefore, $\crn{C_{SC}}$ simulates $\crn{C_{KVG}}$ under polynomial simulation.
\end{proof}

\begin{table}[H]
    \centering
        \KVGsimSC
        \caption*{Table~\ref{tab:VG > cycle} (restated): $k$-VG CRN simulating a Step-Cycle CRN}
\end{table}

\begin{table}[H]
    \centering
    \SCReactants
    \caption*{Table~\ref{tab:check_reactants} (restated): (a) Checking reactants sequentially $\config{\RM_i} \SeqReacts \config{\PM_i}$ (b) Undoing reaction selection if not enough reactants exist for rule $i$.}
\end{table}

\kvgsc*
\begin{proof}

\para{Construction} Given a Step-Cycle CRN $\crn{C_{SC}} = ((\species, \reactions), (\config{S}_0, \config{S}_1, \ldots, \config{S}_{k-1})$, we construct a $k$-VG CRN $\crn{C_{KVG}} = ((\species', \reactions'), Z_{\emptyset})$. In addition to $\species$, $\crn{C_{KVG}}$ has species $G$, $\lambda_1', \ldots, \lambda_{|\species|}'$, $z_{\lambda_1}, \ldots, z_{\lambda_{|\species|}}$, $g_i$, $z_{g_i}$, $r_1, \ldots, r_{\size{\reactions} + 1}$, $R_i^j$, $R_i^{j^-}$, $P_i^l$, $s$, $s_0, \ldots, s_{k-1}$, and $t$, for $1 \leq i \leq |\reactions|, 1 \leq j \leq |\RM_i|+1$ and $1 \leq l \leq |\PM_i|+1$. Species $z_{\lambda_1}, \ldots, z_{\lambda_{|\species|}}$ and $z_{g_1}, \ldots, z_{g_|\reactions|}$ are the zero species; thus, define $Z_\emptyset$ to map each $\lambda_i \to z_{\lambda_i}$ and each $g_i \to z_{g_i}$.
\sloppy Species $G$ is used to select a reaction from the set of species $g_i$, which denotes that rule $i$ is potentially applicable. In the event of a rule application, species $r_1, \ldots, r_{|\reactions|+1}$ is used to reintroduce all missing $g_i$ species. Species $R_i^j$ check to ensure that enough reactants exist to execute rule $\reaction_i$, converting each reactant to a species $\lambda_1', \ldots, \lambda_{|\species|}'$, with species $R_i^{j^-}$ reintroducing the reactants if the rule is not applicable. Species $s$ signifies that no $\reaction \in \reactions$ is applicable, and species $s_0, \ldots, s_{k-1}$ represent the steps. Species $t$ is used to introduce the `steps' corresponding to each step. The initial configuration of $\crn{C_{KVG}}$ should have one copy of $g_1, \ldots, g_{\size{\reactions}}$, $s_0$, and $G$  in addition to the initial configuration of $\crn{C_{SC}}$, with $\species' = \species \cup \{ G, \lambda_1', \ldots, \lambda_{|\species|}', z_{\lambda_1}, \ldots, z_{\lambda_{|\species|}}, r_1, \ldots, r_{\size{\reactions} + 1}, s, s_0, \ldots, s_{k-1}, t \} \cup$  $\{ g_i, z_{g_i}, R_i^j, R_i^{j^-}, P_i^l : 1 \leq i \leq |\reactions|, 1 \leq j \leq |\RM_i|+1, 1 \leq l \leq |\PM_i|+1 \}$ and the reactions shown in Table \ref{tab:VG > cycle}.

Let $S = \{ S_1, \ldots, S_{k-1} \}$ be the set of steps. Reaction $1$ in Table \ref{tab:VG > cycle} attempts to apply $\reaction_i$ by checking if enough of each reactant exists, as shown in Table \ref{tab:check_reactants}. If it is successful, then it simulates the reaction and begins the process of adding $g_1, \ldots, g_{\size{\reactions}}$ back into the system (which is carried out in reactions $2$, $3$, and $4$). If a species in $\reaction_i$ does not exist in the system, then the reactants are reintroduced into the system as shown in Table \ref{tab:undo_reactants}, which also removes $g_i$ from the system, producing $z_{g_i}$. If there are no executable reactions, then reaction $5$ is executed to produce an $s$ species. This $s$ species then reacts with the current `step' of the system, denoted by species $s_i$. Reaction 7 or 8 then runs, introducing the species corresponding to the step. These reactions also introduce the species $t$, which then reintroduces $G$ along with each $g_i$ through reaction $6$. 

Define the mapping function as follows. The intuition is that every configuration in $\crn{C_{KVG}}$ maps to a step configuration in $\crn{C_{SC}}$ if the count of each species $\lambda_i \in \Lambda$ matches the count of the corresponding species $\lambda_i' \in \Lambda'$ and if there exists exactly 1 $G$ species. The step is mapped by the existence of a species $s_1, \ldots, s_{k-1}$, as described below. For simplicity, we represent the step-configuration as a pair, where the first element denotes the configuration and the second element denotes the step.

\[\cmap(\config{C}') =
\begin{cases}
    (\sum_{\lambda' \in \species} \config{C}'[\lambda']\cdot\single{\lambda}', i) & \text{if } \single{G} \in \Set{\config{C}'} \land s_i \in \{ \config{C}' \} \\
    \undefConfig & \text{otherwise}
    
\end{cases}\]

The configuration mapping function is polynomial-time computable because when the function is defined, it computes the sum of vectors $\config{C}'[\lambda']\cdot\single{\lambda}'$ for all species in the simulated system. Therefore, the configuration mapping function is computable in $\mathcal{O}(\size{\species})$. Mapping the configuration to the step is also achieved in $\mathcal{O}(|\species|)$ time, as we are only looking for the species $s_i$ in configuration $\config{C}'$.

\para{$\crn{C_{SC}}$ follows $\crn{C_{KVG}}$} 
Let $\config{A'}$, $\config{B'}$ be configurations in $\crn{C_{KVG}}$ with $M(\config{A'}) \not = M(\config{B'})$. We first show that if $\config{A'}\macrotrans_{\crn{C_{KVG}}} \config{B'}$, and $M(\config{A'})\neq M(\config{B'})$, then $M(\config{A'})\rightarrow_{\crn{C_{SC}}} M(\config{B'})$. Since $M(\config{A}')$ and $M(\config{B}')$ are defined, $G$ must exist in both $\config{A}'$ and $\config{B}'$, and $s_a, s_b$ must exist in $\config{A}'$, $\config{B}'$, respectively. We consider 2 cases. 

\textit{Case 1:} $a = b$. Both $G$ and $s_a$ must exist in $\config{A}'$ and $\config{B}'$. Assume a rule of the form $\single{G} + \single{g}_i \SeqReacts \single{r}_1 +  \config{\PM_i}$ exists and is fully applicable. A new configuration $\config{A}_1' = \config{A}' - \single{G} - \single{g}_i - \config{\RM_i} + \single{r}_1 + \config{\PM}_i$ is then reached without the $G$ species. Reactions $\single{r}_j + \single{g}_j \rightarrow \single{r}_{j+1} + \single{g}_j$ or $\single{r}_j + \single{z}_{g_j} \rightarrow \single{r}_{j+1} + \single{g}_j$ must then fire for $j=1, \ldots, |\reactions|$, which reintroduce any $g_i$ species that were not already in $\config{A}'$. Let $\mathcal{G}$ denote this set of species. Once reaction $\single{r}_{|\reactions|+1} \rightarrow \single{G}$ fires, the configuration $\config{B}' = \config{A}' + \mathcal{G} - \config{\RM}_i + \config{\PM}_i$ is reached. With no $G$ species, it is clear that $M(\config{A_1}) = \undefConfig$, along with any intermediate configuration between $\config{A}_1'$ and $\config{B}'$. From the fact that $\config{B}' = \config{A}' + \mathcal{G} - \config{\RM}_i + \config{\PM}_i$, $M(\config{A}') \neq M(\config{B}')$. Additionally, from the initial assumption, a rule of the form $\single{G} + \single{g}_i \SeqReacts \single{r}_1 +  \config{\PM_i}$ was executable. Denote this rule $\reaction_i'$. By construction, this implies the existence of a rule $\reaction_i \in \reactions$ with $\RM_i \rightarrow \PM_i$. From the mapping function, since $M(\config{A'})=(\config{A}$, $a$), each $\lambda_i \in \{\config{\RM_j}\}$ satisfies $\config{A'}[\lambda_i']=\config{A}[\lambda_i]$, where $\lambda_i'$ is the mapped species of $\lambda_i$ in $\crn{C_{SC}}$. Thus, if $\reaction_i'$ can be applied in $\crn{C_{KVG}}$, it must be the case that $\reaction_i$ can be applied in $\crn{C_{SC}}$. Applying $\reaction_i$ to $\config{A}$ then results in configuration $\config{B}=\config{A}-\config{\RM_j}+\config{\PM_j}$. It is clear that $M(\config{B'})=(\config{B}, a)$. It follows that $M(\config{A'}) \rightarrow_{\crn{C_{SC}}} M(\config{B'})$. 

\textit{Case 2:} $a \neq b$. For simplicity, we start with the assumption that $a + 1 \mod k = b$, that is, $a,b$ are consecutive steps. Assume no rule of the form $\single{G} + \single{g}_i \SeqReacts \single{r}_1 +  \config{\PM_i}$ exists. There must not exist a species $g_i$ in $\config{A}'$. The rule $\single{G} + \single{z}_{g_1} + \ldots + \single{z}_{g_{|\reactions|}} \rightarrow \single{s}$ then executes. Denote the configuration after the execution of this rule $\config{A_1}'$. Since $s, s_a \in \config{A_1}'$, the rule $\single{s} + \single{s}_a \SeqProds \single{t} + \single{s}_{a+1} + \config{S}_a$ or 
$\single{s} + \single{s}_{a} \SeqProds \single{t} + \single{s}_0 + \config{S}_{a}$ must then execute, followed by the rule $\single{t} \rightarrow \single{G} + \single{g}_1 + \ldots + \single{g}_{|\reactions|}$. Let this final configuration be $\config{B}'$. From the sequence of configurations, we get that either $\config{B}' = \config{A}' - \single{s}_a + \single{s}_{a + 1} + \single{S}_{a + 1} + g_1 + \ldots + g_{|\reactions|}$ or $\config{B}' = \config{A}' - \single{s}_a + \single{s}_{0} + \single{S}_{0} + g_1 + \ldots + g_{|\reactions|}$. In both cases, it is clear that $M(\config{B}') \neq M(\config{A}')$ is defined, as either $G, s_{a+1} \in \config{B}'$ or $G, s_{0} \in \config{B}'$ are true. Furthermore, from the initial assumption, no rule of the form $\single{G} + \single{g}_i \SeqReacts \single{r}_1 +  \config{\PM_i}$ exists. This implies that for each rule $i$, $\config{\RM}_i \not \subset \config{A}'$. From the mapping function, since $M(\config{A'})=(\config{A}$, $a$), each $\lambda_j \in \{\config{\RM_i}\}$ satisfies $\config{A'}[\lambda_j']=\config{A}[\lambda_j]$, where $\lambda_j'$ is the mapped species of $\lambda_j$ in $\crn{C_{SC}}$. Thus, if no such reaction can be applied in $\crn{C_{KVG}}$, it must be the case that no reaction can be applied in $\crn{C_{SC}}$. As $\config{A}_a$ is terminal, we then have to reach configuration $\config{B}_b = \config{A}_a + \config{S_b}$. It follows that $M(\config{A'}) \rightarrow_{\crn{C_{SC}}} M(\config{B'})$.

For any $a,b$ that differ by more than 1 step, consider the sequence of configurations $\config{A}', \ldots, \config{B}'$ with defined mappings to get from configuration $\config{A}'$ to configuration $\config{B}'$. For any $2$ configurations in the same step, we get from case $1$ that $M(\config{A'}) \rightarrow_{\crn{C_{SC}}} M(\config{B'})$. For any $2$ configurations in adjacent steps, we get from the above that $M(\config{A'}) \rightarrow_{\crn{C_{SC}}} M(\config{B'})$. Thus, the same holds for any non-consecutive steps $a,b$.

\para{$\crn{C_{KVG}}$ models $\crn{C_{SC}}$} The final part of the proof is to show that if $\config{A}_a \rightarrow_{\crn{C_{SC}}} \config{B}_b$ implies that $\forall \config{A'}\in [\![\config{A}_a\,]\!]$, $\exists \config{B'}\in [\![\config{B}_b\,]\!]$ such that $\config{A'} \Rightarrow_{\crn{C_{SC}}} \config{B'}$. Consider when $a=b$. Let $\reaction_j \in \reactions$ denote a rule that can be applied to $\config{A}_a$. By the mapping function, there exists a species $\lambda_i'$ for each $\lambda_i \in \{\config{\RM_j}\}$ satisfying $\config{A'}[\lambda_i']=\config{A}_a[\lambda_i]$. Since $M(\config{A'})$ is defined, there must exist species $G$ and $s_a$. From the fact that $\reaction_j$ is executable in $\crn{C_{SC}}$, the rule $\single{G} + \single{g}_j \SeqReacts \single{r}_1 +  \config{\PM_j}$ must be executable in $\crn{C_{KVG}}$. As described in \textit{$\crn{C_{SC}}$ follows $\crn{C_{KVG}}$}, configuration $\config{B}'$ is then reached such that $\config{B}' = \config{A}' + \mathcal{G} - \config{\RM}_j + \config{\PM}_j$. Since $M(\config{B'})=\config{B_b}$, $\config{A_a} \rightarrow_{\crn{C_{SC}}} \config{B}_b$ implies that $\forall \config{A'}\in [\![\config{A}_a\,]\!]$, $\exists \config{B'}\in [\![\config{B}_b\,]\!]$ such that $\config{A'} \Rightarrow_{\crn{C_{SC}}} \config{B'}$. A similar argument applies for when $a \neq b$, following from case $2$ from \textit{$\crn{C_{SC}}$ follows $\crn{C_{KVG}}$}.

\para{Polynomial Simulation}
We now show that the above simulation is polynomial efficient as follows. Let $m$ denote the maximum number of reactants in any rule from $\reactions$, and let $k$ denote the number of steps of $\crn{C_{SC}}$.

\begin{enumerate}
    \item \textbf{polynomial species and rules:} In addition to the already existing $|\species|$ species from $\crn{C_{SC}}$, another $1 + 2 \cdot |\species| + 2 \cdot |\reactions| + |\reactions + 1| + 3 \cdot m \cdot |\reactions| + k + 2$ species are added, resulting in $|\species'| = \mathcal{O}((m + 1) \cdot |\reactions| + |\species| + k)$ total species. For each rule in $\reactions$, a total of $3m$ rules are created to `check' the reactants and then `produce' the products. In total, the number of rules becomes $2 + |\reactions| + 1 + 3 \cdot m \cdot |\reactions| + 2 \cdot k$, resulting in $|\reactions'|=\mathcal{O}( (m+1) \cdot |\reactions| + k)$ total rules.

    \item \textbf{polynomial rule size:} For each of the created rules, the maximum number of reactants is $m$ to consume the primed reactants after the check is complete. The maximum number of products is $|\reactions| + 1$ to reintroduce all the $z_{g_i}$ species into the system.
    
    \item \textbf{polynomial transition sequences:} In the worst case, no rule $\reaction_i$ is applicable in $\crn{C_{SC}}$, and so the current step-configuration must transition to the next step. As shown in \textit{$\crn{C_{SC}}$ follows $\crn{C_{KVG}}$}, this implies that each $g_i$ must be attempted and removed from the system, resulting in a sequence of $|\reactions| \cdot 2 \cdot m$ additional rule applications. As the remaining rule applications are constant, for a sequence of $n$ transitions in $\crn{C_{SC}}$, the resulting simulating sequence is of length $\mathcal{O}(|\reactions| \cdot m)$.
    
    \item \textbf{polynomial volume:} Each configuration $\config{C}'$ has a defined mapping only if $G$ and $s_i$ are present in the system. Thus, the volume of $\config{C}'$ will differ only by a constant amount from configuration $M(\config{C}')$. Any intermediate configuration between $\config{C}'$ and some reachable configuration $\config{C}''$ with a defined mapping will have at most an increase of $\mathcal{O}(|\species|)$ volume introduced by the $z_{\lambda}$ species. Thus, the difference in volume from $\config{C}'$ will differ only by a polynomial amount.
    
\end{enumerate}
The above simulation only utilizes 1) polynomial species and rules, 2) polynomial rule size, 3) polynomial transition sequences, and 4) polynomial volume. Therefore, $\crn{C_{KVG}}$ simulates $\crn{C_{SC}}$ under polynomial simulation.
\end{proof}

\vgscthm*
\begin{proof}
    Similar to Theorem \ref{thm:vgicrnthm}, the proof follows from Lemmas \ref{lem:cycle > VG} and \ref{lem:VG > cycle}, Theorem \ref{thm:vg-kvg} and the Transitivity Theorem \ref{thm:transitivity}. 
\end{proof}

\begin{table}[H]
        \centering
        \UIsimsVG
        \caption*{Table~\ref{tab:UI > VG} (restated): UI CRN simulating VG CRN.}
\end{table}
\uikvg*
\begin{proof}
For any given Void-Genesis CRN $\crn{C_{VG}}$ and a configuration of $\crn{C_{VG}}$ $\config{C}$, we show that there exists a UI parallel CRN $\crn{C_{UI}}$ and configuration $\config{C}'$ that simulates $\crn{C_{VG}}$ over $\config{C}$.

\para{Construction} Given a Void-Genesis CRN $\crn{C_{VG}} = ((\species, \reactions), z)$, we construct the UI parallel CRN $\crn{C_{UI}} = (\species', \reactions')$ as follows. We create $\species'=
\species\cup\{G,
G_1,\cdots,G_{\size{\reactions}},
e_{\lambda},\cdots,e_{\lambda_{\size{\species_1}}}, 
t_1^1,\cdots,  \\ t_{\size{\species_1}}^1,
r_j^1,\cdots, r_{\size{\species_1}}^1,
r_j^2,\cdots, r_{\size{\species_1}}^2,
t_2,E\}$. $G$ will be used to select a reaction $\reaction_i\in\reactions$ to simulate. $e_{\lambda},\cdots,e_{\lambda_{\size{\species_1}}}$ are created to initiate checking the counts of each reactant in the chosen reaction $\reaction_i$. $r_j^1,\cdots, r_{\size{\species_1}}^1$ and $r_j^2,\cdots, r_{\size{\species_1}}^2$ are constructed to represent if a reactant has a zero or non-zero count, respectively, following the simulation of $\reaction_i$. $t^1_1,\cdots,t^1_{\size{\species_1}}$ and $t_2$ are ``timer'' species created to help divide the zero-checking process s.t. the involved reactions do not interfere each other. Finally, $G_1,\cdots,G_{\size{\reactions}}$ and $E$ are made to prevent another reaction from being picked until the zero-checking process for all reactants in $\reaction_i$ is complete. Let the configuration of $\crn{C_{UI}}$ $\config{C}'$ be $\config{C}$ with the addition of a single copy of $G$.

First, in Reaction 1, species $\single{G}$ selects a reaction $\reaction_i \in \reactions$ and consumes the reactants $\RM_i$, producing $\single{G}_i$ to prevent Reaction 1 from running again immediately, the set of species $\config{e}_{\{\RM_i\}}$, and the products $\PM_i$. With Reaction 2, each $\single{e}_{\lambda_j}$ copy is transformed into a copy of the species $\single{t}^1_j$ and $\single{r}^1_j$. Note that, under unique-instruction parallel dynamics, all of the new copies are created within a single time-step.
We then move to Reaction 3, which again is composed of two reactions that will run in parallel due to $\crn{C_{UI}}$'s dynamics.
If a copy of a reactant of $\reaction_i$ $\single{\lambda}_j$ remains present, then $\single{r}^1_j$ and $\single{\lambda_i}$ transform into $\single{r}^2_j$. At the same time, all copies of $\single{t}^1_j$ turn into $\single{t}_2$.
Reaction 4 is composed similarly to Reaction 3. $\single{t}_2$ will transform all copies of $\single{r}^2_j$ back into $\single{\lambda}_i$, and any remaining copies of $\single{r}^1_j$ into $\single{z}$. In both cases, the species $\single{E}$ is also produced.
Finally, in Reaction 5, $\single{G}_i$ and $\size{\{\RM_i\}}$ copies of $\single{E}$ are consumed to reintroduce $\single{G}$, allowing the system to choose a new reaction. 

Each configuration $\config{C'}$ in $\crn{C_{UI}}$ for which $\cmap(\config{C'})$ is defined will contain one leader species $G$, a sequence of $\lambda_i \in \species$, and a unique sequence of zero species corresponding to $\lambda_i$ for which $\config{C'}[\lambda_i]=0$. Hence every $\config{C'}$ will have a unique mapping $\cmap(\config{C'})$:

\vspace{-.2cm}
\[\cmap(\config{C}') =
\begin{cases}
    \sum_{\lambda' \in \species} \config{C}'[\lambda']\cdot\single{\lambda}' & \text{if } \single{G} \in \Set{\config{C}'} \\
    \undefConfig & \text{otherwise.}
\end{cases}\]
\vspace{-.2cm}

\para{$\crn{C_{VG}}$ follows $\crn{C_{UI}}$}
Given configurations $\config{A'}$ and $\config{B'}$ in $\crn{C_{UI}}$ such that both $M(\config{A'})$ and $M(\config{B'})$ are defined, we prove that $\crn{C_{VG}}$ follows $\crn{C_{UI}}$. Using the mapping $M$, we define $M(\config{A'}) \in \config{\crn{C_{VG}}}$ as $\sum_{\lambda_i \in \species} A'[\lambda_i] \cdot \single{\lambda_i}$, which is equivalent to $\config{A'} - \single{G}$. Given that $M(\config{A'})$ is equivalent to  $\config{A'}$ minus $\single{G},$ we observe the following: if reaction $\single{G} + \config{\RM_i} \rightarrow \config{G_i} + \config{e_{\RM_i}} + \config{\PM_i}$ is applicable under $\crn{C}_{UI}$ then an equivalent reaction $\reaction_i$ can be applied in $\crn{C_{VG}}$ resulting in $\config{B}_t = M(\config{A'}) - \config{R} +\config{G_i}+ \config{P}$ and $\config{B} = \config{B}_t + \config{Z}$ as defined in Definition 17.

We define the first intermediary configuration as $X'_1 = \config{A'} - \RM_i + \PM_i + \config{e}_{\{\RM_i\}}$. As described in the construction, this simulation proceeds through five reactions, with each reaction producing the next intermediary configuration. With Reaction 2, each $e_{\lambda_j}$ is replaced with its corresponding timer $t_j^1$ and representative $r_j^1$, resulting in the configuration $X'_2 = X'_1 - \config{e}_{\{\RM_i\}} + \config{r^1}_{\{\RM_i\}} + \config{t^1}_{\{\RM_i\}}$. Through Reaction 3, each timer $t^1_j$ is replaced with a copy of $t^2$, and if possible, replaces the representative species $r_j^1$ with $r_j^2$, indicating the presence of $\lambda_j$ in the system. Consequently, the resulting configuration is \(
\config{X}'_3 = \config{X}'_2 
+ \sum_{\lambda_j \in \{ \RM_i \} \setminus Z} \single{r}_2^j 
- \sum_{\lambda_j \in \{ \RM_i \} \setminus Z} \single{r}_1^j 
- \sum_{\lambda_j \in \{ \RM_i \} \setminus Z} \single{\lambda_j} 
+ |\{ \RM_i \}| \cdot \config{t^2} 
- |\{ \RM_i \}| \cdot \config{t^1}
\), where $Z$ denotes the set of all newly detected zero-count species
. Reaction 4 involves two reactions, each consuming one instance of timer $t^2$ along with either $r_j^1$ or $r_j^2$. Since $r_j^1$ indicates the absence of $\lambda_j$, it is replaced by the species $z$. Conversely, $r_j^2$ denotes that $\lambda_j$ is present, and thus a copy of $\lambda_j$ is added back into the system. Each of these two reactions also produces a copy of species $E$ which will be used for Reaction 5. This produces the configuration $\config{X}'_4 = \config{X}'_3 - \sum_{\lambda_j \in {Z}} \single{r_j^1} + \sum_{\lambda_j \in {Z}} \single{z} - \sum_{\lambda_j \in \{ \RM_i \}  \setminus Z} \single{r_j^2} +\sum_{\lambda_j \in \{ \RM_i \}  \setminus Z} \single{\lambda_j}  + |{\{\RM_i\}}|\cdot \single{E}$. Finally, Reaction 5 takes each copy of $E$ created and the rule marker $G_i$ to produce rule selector $G$, resulting in the defined configuration $\config{B'} = \config{X}'_4 + \single{G} - |\config{e}_{\{\RM_i\}}|\cdot \config{E} = \config{A'} - \config{\RM}_i + \config{\PM}_i + \single{G}$. Using configuration mapping $M$ we see that  $M(\config{B}') =  \config{A}' - \config{\RM}_i + \config{\PM}_i$. Given that $\config{B} = M(\config{B}')$ we conclude the following: For all configurations $\config{A}'$ and $\config{B}'$, if there is a macro transition $\config{A}' \macrotrans_{\crn{C_{UI}}} \config{B}'$ and $M(\config{A}') \neq M(\config{B}')$, there is an equivalent rule application $\reaction_i$ such that $M(\config{A}') \rightarrow_{\crn{C_{VG}}} M(\config{B}')$.

\para{$\crn{C_{UI}}$ models $\crn{C_{VG}}$} Given configurations $\config{A}$ and $\config{B}$ in $\crn{C_{VG}}$, we demonstrate that $\crn{C_{UI}}$ models $\crn{C_{VG}}$. As described through our representative configuration mapping, $[\![\config{A}\,]\!]$ consists solely of the configuration $\config{A} = \sum_{\lambda_i \in \species} A'[\lambda_i] \cdot \single{\lambda_i} + \single{G}$. Assuming that rule $\reaction_i$ is applied to $\config{A}$ to produce the configuration $\config{B} = \config{A} - \config{\RM}_i + \config{\PM}_i + Z$, an equivalent macro-transition from $\config{A}'$ exists, yielding the configuration $\config{B}' = \config{A}' - \config{\RM}_i + \config{\PM}_i + \single{Z}$, as established in the proof for $\crn{C_{VG}}$ simulates $\crn{C_{UI}}$. Given that $[\![\config{B}\,]\!]$ corresponds exactly to $\config{B'}$, we conclude that if $\config{A} \rightarrow_{\crn{C_{VG}}} \config{B}$, then for every $\config{A}' \in [\![\config{A}\,]\!]$, there exists a configuration $\config{B}' \in [\![\config{B}\,]\!]$ such that $\config{A}' \Rightarrow_{\crn{C_{UI}}} \config{B}'$.

\para{Polynomial Simulation}
We now show that the above simulation is polynomial efficient as follows.

\begin{enumerate}
\item \textbf{polynomial species and rules:} As defined in the construction of UI Parallel CRNs, the species alphabet $\species'$ of $\crn{C_{UI}}$ consists of $\species \cup \{ G, E, t^2 \} \cup \{
e_{\lambda_j}, t_j^1, r_j^1, r_j^2, G_i : i \in \{1, \ldots, |\reactions|\},\\
j \in \{1, \ldots, |\{\RM_i\}|\} \}$. This results in a total species count of $|\Lambda'| = \mathcal{O}(|\reactions| + |\species|)$. For the ruleset $\reactions'$, as detailed in Table~7, one rule is created for each $\reaction \in \reactions$, along with six distinct rules for each species $\lambda_j \in \species_1$. This yields a total ruleset size of $|\reactions'| = \mathcal{O}(|\reactions| + |\species|)$.

    \item \textbf{polynomial rule size:} For rule $\reaction_i \in \reactions$ of size $(m, n)$, a rule of at most size $(n + 1, n + m+ 1)$. 
    
    \item \textbf{polynomial transition sequences:} For each rule $\reaction\in \reactions$ applied in $\crn{C_{VG}}$, as shown in $\crn{C_{VG}}$ follows $\crn{C_{UI}}$, there is an equivalent macro transition which undergoes 5 transitions. Thus for a sequence of $n$ transitions in $\crn{C_{VG}}$, the resulting simulating sequence in $\crn{C_{UI}}$ is of length  $\mathcal{O}(n).$
    
    \item \textbf{polynomial volume:} Each configuration $\config{C'}$ has a defined mapping only if $G$ is present in the system. Thus, the volume of $\config{C'}$ and the configuration $M(\config{C'})$ will only differ by a constant amount. Any intermediate configuration between $\config{C'}$ and some reachable configuration $\config{C''}$ with a defined mapping will have at most an increase of $\mathcal{O}(|\species|)$ which is the species of the reactants, which then create a checker species. Thus, the difference in volume will differ only by a polynomial amount.  
    
\end{enumerate}

The above simulation only utilizes 1) polynomial species and rules, 2) polynomial rule size, 3) polynomial transition sequences, and 4) polynomial volume. Therefore, $\crn{C_{UI}}$ simulates $\crn{C_{VG}}$ under polynomial simulation.

Because (1) $\crn{C_{VG}}$ follows $\crn{C_{UI}}$, and (2) $\crn{C_{UI}}$ models $\crn{C_{VG}}$ using a polynomial-time computable function $\cmap$, we can say that $\crn{C_{UI}}$ simulates $\crn{C_{VG}}$.
\end{proof}

\begin{table}[H]
        \centering
        \KVGsimsUI
        \caption*{Table~\ref{tab:kVG > UI} (restated): $k$-VG CRN simulating UI CRN.}
\end{table}
\kvgui*
\begin{proof}
    
\para{Construction} Given a Unique-Instruction Parallel CRN $\crn{C_{UI}} = (\species, \reactions)$ and a configuration $\config{C}$, we construct a $k$-Void-Genesis CRN $\crn{C_{KVG}} = ((\species', \reactions'), \Zero)$  and configuration $\config{C}'$ as described. 
We create $\species'=\species\cup\{F\}\cup\{G_i,N_i,I_i,F_i,R_i^j:i\in\{1,\ldots,|\reactions|\}, j\in\{1,\ldots,|\RM_i|\}\}\cup\{z_{r_j}:r_j\in \RM_i s.t. 1 \leq i \leq |\reactions|\}$. Let $\config{C}'$ be $\config{C}$ with one copy for each of $G_1,\cdots,G_{\size{\reactions}}$.

Table \ref{tab:kVG > UI} overviews the general rules. The system must simulate a maximal selection of possible rules to run. First, a $\single{G}_i$ species exists for every rule in $\reactions$, which sequentially consumes the reactants (Reactions 2/3). If each reactant of a rule $\reaction_i$ is present, Reaction 5 produces $\single{I}_i$. If some reactant is missing, Reactions 1/4 refunds the previously consumed reactants and produces $\single{N}_i$ instead. Both processes create an $\single{X}_i$ copy, and when $\size{\reactions}$ copies of $\single{X}_i$ exist, a maximal set has been chosen and $\single{F}$ is created with Reaction 6. This converts a copy of $\single{N}_i$ or $\single{I}_i$ to  $\single{F}_i$ and output the products $\config{\PM}_i$ if given $\single{I}_i$ (Reaction 7).
Once $\size{\reactions}$ copies of $\single{F}_i$ are created, Reaction 9 resets the rule selection process by consuming all $\single{F}_i$ copies and restoring all $\single{G}_i$ copies.

We define a configuration mapping as follows. If there exists the reaction species $G_1,\ldots,G_{|\reactions|}$ in the configuration $\config{C'}$ then the configuration $\config{C'}$ maps to a configuration in $\crn{C_{UI}}$ that contains the same count of species $\lambda'$, $\forall \lambda' \in \species$.  All intermediate configurations $\config{C'}$ that do not contain all $G_i$ species do not map to anything in the original system. 

\vspace{-.2cm}
\[\cmap(\config{C}') =
\begin{cases}
    \sum_{\lambda' \in \species} \config{C}'[\lambda']\cdot\single{\lambda}' & \text{if } \single{G}_1,\cdots,\single{G}_{\size{\reactions}} \in \Set{\config{C}'} \\
    \undefConfig & \text{otherwise.}
\end{cases}\]
\vspace{-.2cm}

The configuration mapping function is polynomial-time computable because when the function is defined, it computes the sum of vectors $\config{C}'[\lambda']\cdot\single{\lambda}'$ for all species in the simulated system. Therefore, the configuration mapping function is computable in $\mathcal{O}(\size{\species})$.

Now we show that $\crn{C_{KVG}}$ simulates any given $\crn{C_{UI}}$. This is done in two parts. We first show that $\crn{C_{UI}}$ follows $\crn{C_{KVG}}$, and then we prove that $\crn{C_{KVG}}$ models $\crn{C_{UI}}$.

\para{$\crn{C_{UI}}$ follows $\crn{C_{KVG}}$}
We show that $\crn{C_{UI}}$ follows $\crn{C_{KVG}}$ in two parts. We start by proving that for any two given configurations $\config{A'}$ and $\config{B'}$ in $\crn{C_{KVG}}$ where $\cmap(\config{A'})$ and  $\cmap(\config{B'})$ are defined, such that $\config{A'} \macrotrans_{\crn{C_{KVG}}} \config{B'}$ and $\cmap(\config{A'}) \not= \cmap(\config{B'})$, then $\cmap(\config{A'})\rightarrow_{\crn{C_{UI}}}\cmap(\config{B'})$. We know that both $\config{A'}$ and $\config{B'}$ contain set of global leaders $G_1, \ldots, G_{\size{\reactions}}$ because $\cmap(\config{A'})$ and $\cmap(\config{B'})$ are defined. Because $\cmap(\config{A'}) \not= \cmap(\config{B'})$, there exists an applicable reaction $\reaction_i \in \reactions$ for which there will be reactions $\single{G}_i + \single{r}_1 \rightarrow \single{R}_i^2$ and $\single{F} + \single{F}_1 + \ldots + \single{F}_{|\Gamma|} \rightarrow \single{G}_i + \ldots + \single{G}_{|\reactions|}$ that are applicable in $\crn{C_{UI}}$. And the configuration $\config{A'}$ is the configuration before the reaction $\single{G}_i + \single{r}_1 \rightarrow \single{R}_i^2$ is applied and $\config{B'}$ is the resulting configuration after the reaction represented by $\single{F} + \single{F}_1 + \ldots + \single{F}_{|\Gamma|} \rightarrow \single{G}_i + \ldots + \single{G}_{|\reactions|}$ is applied. The system will start with the configuration $\config{A'}$ and transition through certain intermediate configurations consisting of species $\single{G}_i, \single{z}_{r_j}, \single{N}_i, \single{X}_i, \single{r}_1, \ldots, \single{r}_j , \single{R}_i^2, \ldots, \single{R}_i^{\size{\reactions}}, \single{I}_i, \single{F},$ and $\single{F}_1, \ldots, \single{F}_{\size{\reactions}}$, finally resulting in configuration $\config{B'}$. In the original system $\crn{C_{UI}}$, the configuration $\cmap(\config{A'})$ containing the same count of species in $\species$ as in $\config{A'}$ will transition to the configuration $\cmap(\config{A'})-\sum_{i=1}^{k} \reactants_i + \sum_{i=1}^{k} \products_i$ as described in Definition \ref{def:unique-dynamic}. The resulting configuration $\cmap(\config{B'}) = \cmap(\config{A'})-\sum_{i=1}^{k} \reactants_i + \sum_{i=1}^{k} \products_i$ in $\crn{C_{UI}}$ will contain the same counts of species in $\species$ as in $\config{B'}$. Hence showing that it follows $\cmap(\config{A'}) \rightarrow_{\crn{C_{UI}}} \cmap(\config{B'})$, if $\config{A'} \macrotrans_{\crn{C_{KVG}}} \config{B'}$ and $\cmap(\config{A'}) \not= \cmap(\config{B'})$.

\para{$\crn{C_{KVG}}$ models $\crn{C_{UI}}$}
Say $\config{A}$ and $\config{B}$ are two configurations in $\crn{C_{UI}}$ such that $\config{A} \rightarrow_{C_{UI}} \config{B}$. As described in the configuration of configuration mapping,  for any configuration $\config{A}$ we only have one configuration $\config{A'} \in [\![\config{A}\,]\!]$ which contains the same count of species $\lambda_i \in \species$ as in $\config{A}$. Similarly the configuration $\config{B}$ will have only one configuration $\config{B'} \in [\![\config{B}\,]\!]$. Both configurations $\config{A'}$ and $\config{B'}$ will contain species the set of global leaders $G_1, \ldots, G_{\size{\reactions}}$ because their mapping is defined. We know that $\config{A} \rightarrow_{C_{UI}} \config{B}$, hence there exists an applicable reaction $\reaction_i \in \reactions$ such that $\config{B}= \config{A}-\sum_{i=1}^{k} \reactants_i + \sum_{i=1}^{k} \products_i$ where we get the maximal by Def \ref{def:unique-dynamic}. If such a reaction is applicable and the system $\crn{C_{KVG}}$ is in configuration $\config{A'}$ containing all of $G_1, \ldots, G_{\size{\reactions}}$, then there will exist a sequence of applicable reactions in $\crn{C_{UI}}$ (representing $\reaction_i$) leading to configuration $\config{B'}$ that also contains all of $G_1, \ldots, G_{\size{\reactions}}$ such that $\sum_{\lambda_i\in\species}\config{B'}[\lambda_i]\cdot\single{\lambda}_i = \sum_{\lambda_i\in\species}\config{A'}[\lambda_i]\cdot\single{\lambda}_i - \sum_{i=1}^{k} \reactants_i + \sum_{i=1}^{k} \products_i$. By the definition of configuration mapping, $\config{B'} \in [\![\config{B}\,]\!]$ (i.e., $\cmap(\config{B'}) = \config{B}$). Therefore for any two configurations $\config{A}$ and $\config{B}$ such that $\config{A}\rightarrow_{\crn{C_{UI}}}\config{B}$ and $\config{A'} \in [\![\config{A}\,]\!]$ and $\config{B'} \in [\![\config{B}\,]\!]$, $\config{A'}\macrotrans_{\crn{C_{KVG}}}\config{B'}$.

\para{Polynomial Simulation}
We now show that the above simulation is polynomial efficient as follows.

\begin{enumerate}
    \item \textbf{polynomial species and rules:} As defined in the construction, along with the original species in $\species$, the set $\species'$ contains the set of leader species which are for every reaction $\size{\reactions}$, the reset species for each reaction $\size{\reactions}$, and other species are noted to be $\size{\reactions}$. The checker species R is the size of the $\size{\reactions} \cdot \size{\species}$ since we are checking each reaction's species. Totaling the amount of species to $\mathcal{O}(\size{\reactions} \cdot \size{\species})$. In the addition to the already existing $\single{\reactions}$ from $\crn{C_{UI}}$, another we use $\size{\reactions} \cdot \size{\species}$, since we use $2 \cdot(\size{\reactions} \cdot \size{\species})$ for rules that contain the $G$ species. For rule selections, we have a total of $\size{\reactions}^2 \cdot \size{\species}$. rules selected and reset are 1 reaction. For creating the products, when a rule didn't run, both of them are $\size{\reactions}$. Totalling the amount of of rules $\size{\reactions'} = \mathcal{O}(\size{\reactions}^2 \cdot \size{\species})$.

    \item \textbf{polynomial rule size:} For a rule $\reaction \in \reactions$ applied in $\crn{C_{UI}}$, we use different rules sizes, from constant size $(2, 2), (2, 3)$. Onces a reaction is selected, the size of the product is based on the about of reactants are in the reaction, making it size of $(2, \size{\RM_i})$. For when rules are selected we use a rule size of  $(\size{\reactions}, 1)$, and for resetting $(\size{\reactions}, \size{\reactions})$. At most we will have $\mathcal{O}(\size{\reactions})$ for the reaction size, and $ \mathcal{O}(\size{\reactions} + \size{\RM_i})$ for the product sizes making it polynomial in rule size.
    
    \item \textbf{polynomial transition sequences:} For any rule $\reaction_i$ trasitioning to a rule $\reaction_{i+1}$ in the $\crn{C_{UI}}$, there are $\mathcal{O}(\size{\reactions} \cdot \size{\RM_i} )$ intermediate rules that are applied in $\crn{C_{UI}}$. Therefore, any sequence of $n$ transitions in the $\crn{C_{UI}}$ is simulated by $\mathcal{O}(\size{\reactions} \cdot \size{\RM_i}) \cdot n$ in the $\crn{C_{KVG}}$.
    
    \item \textbf{polynomial volume:}  Every configuration $\config{C'}$ has a defined mapping only if the set of global leaders $G_1, \ldots, G_{\size{\reactions}}$ are present in the system. Thus the volume of $\config{C'}$ and the configuration $M(\config{C'})$ will only differ by $\mathcal{O}(\size{\reactions})$. Any intermediate configuration between $\config{C'}$ and some reachable configuration $\config{C''}$ with a defined mapping will have at most an increase of $\mathcal{O}(\size{\reactions})$ which is the species of the reactants, which then create a checker species $I_i$. Thus, the difference in volume will differ only by a polynomial amount. 
    
\end{enumerate}

The above simulation only utilizes 1) polynomial species and rules, 2) polynomial rule size, 3) polynomial transition sequences, and 4) polynomial volume. Therefore, $\crn{C_{KVG}}$ simulates $\crn{C_{UI}}$ under polynomial simulation.

Because (1) $\crn{C_{UI}}$ follows $\crn{C_{KVG}}$, and (2) $\crn{C_{KVG}}$ models $\crn{C_{UI}}$ using a polynomial-time computable function $\cmap$, we can say that $\crn{C_{KVG}}$ simulates $\crn{C_{UI}}$.
\end{proof}
\vguithm*
\begin{proof}
    Similar to Theorem \ref{thm:vgicrnthm}, the proof follows from Lemmas \ref{thm:UI > VG} and \ref{thm:VG > UI}, Theorem \ref{thm:vg-kvg} and the Transitivity Theorem \ref{thm:transitivity}. 
\end{proof}

\end{document}